\documentclass{article} 
\usepackage{nips13submit_e,times}
\usepackage{hyperref}
\usepackage{url}
\usepackage{amsmath}
\usepackage{amssymb}
\usepackage{amsthm}

\usepackage{hyperref}
\usepackage{url}
\usepackage{amsmath}
\usepackage{amssymb}
\usepackage{amsthm}

\newtheorem{definition}{Definition}
\newtheorem{lemma}{Lemma}
\newtheorem{assumption}{Assumption}
\newtheorem{prop}{Proposition}
\newtheorem{cor}{Corollary}

\newcommand{\hsa}{\hspace{0.2in}}
\newcommand{\hsb}{\hspace{0.06in}}
\newcommand{\hsc}{\hspace{0.3in}}

\newcommand{\mS}{{\mathcal M}}
\newcommand{\setR}{{\mathcal R}}

\newcommand{\setAd}{{\mathcal A}^{(d)}}
\newcommand{\setAs}{{\mathcal A}^{(s)}}

\newcommand{\Add}{\setAd}
\newcommand{\Ads}{\setAs}

\newcommand{\setC}{{\mathcal C}}

\newcommand{\dC}{C_{\dl}}
\newcommand{\dCk}{C_{\dl,k}}

\newcommand{\ddC}{C_{\dl}}

\newcommand{\ddCs}{C^*_{\dl}}
\newcommand{\Cdd}{{C_{d}}}
\newcommand{\Cds}{{C_{s}}}
\newcommand{\Cddk}{{C_{k,d}}}
\newcommand{\Cdsk}{{C_{k,s}}}

\newcommand{\CI}{{\left (I_C,I_C \right ) }}
\newcommand{\Cs}{{\left (I_{C^*},I_{C^*} \right ) }}
\newcommand{\Cd}{{\left (\Cdd,\Cds \right ) }}

\newcommand{\Csd}{\left (\intsdd,\intssd \right ) }

\newcommand{\sC}{{\mathcal C}}
\newcommand{\sCNash}{{\mathcal C}}
\newcommand{\sCd}{{\mathcal C}}
\newcommand{\sCdNash}{{\mathcal C}}
\newcommand{\sCdd}{{\mathcal C}_{\dl}}
\newcommand{\sCddNash}{{\mathcal C}_{\dl}}
\newcommand{\dsCd}{{\mathcal C(L_C,\dl)}}

\newcommand{\CS}{{\mathcal S}}
\newcommand{\CSNash}{{\mathcal S}^*}
\newcommand{\CSd}{{\mathcal S}}
\newcommand{\CSdNash}{{\mathcal S}^*}
\newcommand{\CSdd}{{\mathcal S}_{\dl}}
\newcommand{\CSddNash}{{\mathcal S}^*_{\dl}}

\newcommand{\dCSd}{{\sCd_\dl(L_C)}}
\newcommand{\dCSdsub}{{\sCd_{\dl_0}(L_C)}}

\newcommand{\commStruc}{\left ( \sC, \{\aSC(y)\}_{y \in \setR},  \{\bSC(\cdot|y)\}_{y \in \setR} \right )}

\newcommand{\commStrucNash}{\left ( \sCNash, \{\aSCs(y)\}_{y \in \setR},  \{\bSCs(\cdot|y)\}_{y \in \setR} \right )}

\newcommand{\dcommStruc}{\left ( \sCd, \{\aSCd(y)\}_{y \in \Add},  \{\bSCd(\cdot|y)\}_{y \in \Ads} \right )}
\newcommand{\dcommStrucNash}{\left ( \sCdNash, \{\aSCsd(y)\}_{y \in \Add},  \{\bSCsd(\cdot|y)\}_{y \in \Ads} \right )}

\newcommand{\ddcommStrucNash}{\left ( \sCddNash, \{\aSCsdd(y)\}_{y \in \Add},  \{\bSCsdd(\cdot|y)\}_{y \in \Ads} \right )}

\newcommand{\Kd}{K^{(d)}}
\newcommand{\Ks}{K^{(s)}}

\newcommand{\Bbl}{\Big ( }
\newcommand{\Bbr}{\Big ) }
\newcommand{\Bsbl}{\Big [ }
\newcommand{\Bsbr}{\Big ] }
\newcommand{\Bsl}{{\Big \{ }}
\newcommand{\Bsr}{{\Big \} }}

\newcommand{\dl}{{\delta}}
\newcommand{\Dl}{{\Delta}}
\newcommand{\dld}{{\delta_d}}
\newcommand{\dls}{{\delta_s}}

\newcommand{\Id}{{I_{C}}}

\newcommand{\intd}{C_\dl}
\newcommand{\intdd}{C_d}
\newcommand{\intsd}{C_s}
\newcommand{\intsdd}{\Cddk}
\newcommand{\intssd}{\Cdsk}

\newcommand{\Ld}{{L(\intd)}}
\newcommand{\LCd}{{L_{C}}}
\newcommand{\Ldd}{{L\big (\intdd \big )}}

\newcommand{\LSd}{{L^*_{\dC}}}
\newcommand{\midd}{{mid(\intd)}}
\newcommand{\middd}{{mid(\intdd)}}

\newcommand{\aC}{{\alpha_{C}}}

\newcommand{\aSC}{{\alpha_{\setC}}}
\newcommand{\aSCs}{{\alpha^*_{\sCNash}}}

\newcommand{\aCd}{{\alpha_{\dC}}}
\newcommand{\aCsd}{{\alpha^*_{\dC}}}
\newcommand{\aSCd}{{\alpha_{\sCd}}}
\newcommand{\aSCsd}{{\alpha^*_{\sCd}}}
\newcommand{\aCdNash}{{\aCsd}}

\newcommand{\aSCdd}{{\alpha_{\sCdd}}}
\newcommand{\aSCsdd}{{\alpha^*_{\sCdd}}}

\newcommand{\bSC}{{\beta_{\setC}}}
\newcommand{\bSCs}{{\beta^*_{\sCNash}}}
\newcommand{\bCd}{{\beta_{\dC}}}
\newcommand{\bCsd}{{\beta^*_{\dC}}}
\newcommand{\bSCd}{{\beta_{\sCd}}}
\newcommand{\bSCsd}{{\beta^*_{\sCd}}}
\newcommand{\bCdNash}{{\bCsd}}

\newcommand{\bSCdd}{{\beta_{\sCdd}}}
\newcommand{\bSCsdd}{{\beta^*_{\sCdd}}}

\newcommand{\xd}{{x^*_\dl}}
\newcommand{\xCd}{{x^*_{\dC}}}
\newcommand{\xs}{{x^*}}

\newcommand{\Dy}{{\Delta^*(y)}}
\newcommand{\Ddy}{{\Delta^*_\dl(y)}}

\newcommand{\PCd}{{P_{\dC}}}

\newcommand{\QCd}{{Q_{\dC}}}
\newcommand{\QCsd}{{Q^*_{\dC}}}
\newcommand{\QCdNash}{{\QCsd}}

\newcommand{\UCd}{{U^{(d)}_C}}
\newcommand{\UCs}{{U^{(s)}_C}}

\newcommand{\FCd}{{F^{(d)}_C}}
\newcommand{\FCs}{{F^{(s)}_C}}
\newcommand{\FsCd}{{F^{(d)}_{C^*}}}

\newcommand{\UCSdNash}{{U^{(d)}_{\CSNash}}}
\newcommand{\UCSsNash}{{U^{(s)}_{\CSNash}}}

\newcommand{\UCdd}{{U^{(d)}_{\dC}}}
\newcommand{\UCsd}{{U^{(s)}_{\dC}}}
\newcommand{\FCdd}{{F^{(d)}_{\dC}}}
\newcommand{\FCsd}{{F^{(s)}_{\dC}}}
\newcommand{\FsCdd}{{F^{(d)}_{\dC^*}}}

\newcommand{\UCSdd}{{U^{(d)}_{\CSdd}}}
\newcommand{\UCSsd}{{U^{(s)}_{\CSdd}}}
\newcommand{\UCSddNash}{{U^{(d)}_{\CSddNash}}}
\newcommand{\UCSsdNash}{{U^{(s)}_{\CSddNash}}}

\newcommand{\erNash}{$\epsilon-$equilibrium }
\newcommand{\erNASH}{$\epsilon-$Equilibrium }
\newcommand{\erNashn}{$\epsilon-$equilibrium}

\newcommand{\dlO}{\dl_0(L_C)}

\newcommand{\tl}{\tilde}

\newcommand{\ty}{\tl y}

\newcommand{\tdC}{\tl \dC}
\newcommand{\tintdd}{{\tl C_d}}
\newcommand{\tintsd}{{\tl C_s}}
\newcommand{\tCd}{(\tintdd,\tintsd)}

\newcommand{\tUCdd}{{U^{(d)}_{\tdC}}}
\newcommand{\tFCdd}{{F^{(d)}_{\tdC}}}

\title{Community Structures in Information Networks for a Discrete Agent Population}

\author{Peter Marbach\\ Department of Computer Science\\ University of Toronto \vspace{0.5in}}

\nipsfinalcopy 

\begin{document}


\maketitle

\begin{abstract}
Communities are an important feature of social networks. The goal of this paper is to propose a mathematical model to study the community structure in social networks. For this, we consider a particular case of a social network, namely information networks. We assume that there is a population of agents who are interested in obtaining content. Agents differ in the type of content they are interested in. The goal of agents is to form communities in order to maximize their utility for obtaining and producing content. We use this model to characterize the structure of communities that emerge as a Nash equilibrium in this setting. The work presented in this paper generalizes results in the literature that were obtained for the case of a continuous agent model, to the case of a discrete agent population model. We note that a discrete agent set reflects more accurately real-life information networks, and are needed in order to get additional insights into the community structure, such as for example the connectivity (graph structure) within in a community, as well as information dissemination within a community. 

\end{abstract}

\section{Introduction}\label{section:introduction}

In this paper we consider a  particular type of social network, which we refer to as an \emph{information network}, where agents (individuals) share/exchange information. Sharing/exchanging of information is an important aspect of social networks, both for social networks that we form in our everyday lives, as well as for online social networks such as Twitter. 

The work in~\cite{continuous_model_ita,continuous_model_arxiv} presents a model to study communities in information networks where agents produce (generate) content, and consume (obtain) content. Furthermore, the model allows agents to form communities in order to share/exchange content more efficiently, where agents obtain a certain utility for joining a given community. Using a game-theoretic framework, \cite{continuous_model_ita,continuous_model_arxiv} characterizes the community structures that emerge in information networks as Nash equilibria. More precisely, \cite{continuous_model_ita,continuous_model_arxiv} considers a particular family of community structures, and shows that (under suitable assumptions) there always exists a community structure that is a Nash equilibrium. 

An interesting outcome of the analysis in~\cite{continuous_model_ita,continuous_model_arxiv} is that, albeit being very simple, the model in~\cite{continuous_model_ita,continuous_model_arxiv} indeed is able to provide interesting insights into the microscopic structure of information communities. For example, the characterization of how content is being produced, i.e. which content each agent in a community produces, indeed matches what has been experimentally observed in real-life social networks.

However, the analysis in~\cite{continuous_model_ita,continuous_model_arxiv} was based on a simplified agent model, i.e. it considered the limiting case of a very dense agent population. As a result, rather then having a discrete set of agent, the agent population was given by a continuous agent density. This assumption simplifies the analysis, and provides the right insights regarding the structure of communities that emerge in an information network. However, the model has its limitations. In particular, as the agent population is given by a continuous density, it does not lend itself readily to study additional properties of communities in information networks such as for example the dynamics of how information communities form, as well as the network (graph) structure within an information community.

In this paper we address this issue by considering a discrete agent population model, i.e. a model whether the information network consists of a finite set of agents. As such, the model considered captures more accurately the situation of actual information networks that do consists of a finite set of agents. We then extend the analysis in~\cite{continuous_model_ita,continuous_model_arxiv} to the discrete agent population model, and show that the results in~\cite{continuous_model_ita,continuous_model_arxiv} can also be established for the discrete agent model. In particular, the results in~\cite{continuous_model_ita,continuous_model_arxiv} are recovered for the limiting case where the size of agent population approaches infinity.

In summary, the contribution of the paper is a technical contribution in that we extend the result in~\cite{continuous_model_ita,continuous_model_arxiv} that were obtained for a continuous agent model to the case of a discrete agent population model. While the analysis shares some similarities with the one in~\cite{continuous_model_ita,continuous_model_arxiv}, and we use some of the results in~\cite{continuous_model_ita,continuous_model_arxiv} for our analysis, the extension to the case of a discrete agent population is technically quite involved and non-trivial. We also note that the discrete agent case is an important case as it more accurately reflects real-life information networks. In addition, results for this case are needed in order to study additional questions such as for example the connectivity (graph structure) within in a community, as well as information dissemination within a community. Indeed, the motivation for this paper was that we wanted to study these questions, and the results of~\cite{continuous_model_ita,continuous_model_arxiv} were not applicable to this case.

The rest of the paper is structured as follows. In Section~\ref{section:related} we discuss existing work that is most closely related to the model and analysis presented in this paper. In Section~\ref{section:content_model} and Section~\ref{section:agent_model} we define the mathematical model that we use for our analysis. In Section~\ref{section:results} we present our main results, and  in Section~\ref{subsection:discussion_results} we discuss the insights obtained from the presented model and analysis, as well as future research. Due to space constraints, all the proofs are given in the appendix.
  
\section{Related Work}\label{section:related}
As our work builds on earlier work presented in~\cite{continuous_model_ita,continuous_model_arxiv}, the summary of related work in~\cite{continuous_model_ita,continuous_model_arxiv} also applies to this paper. For the reader's convenient, we provide here the overview as well.

There exists extensive work, both experimental and theoretical, on the macroscopic properties of social network graphs such as the small world phenomena, shrinking diameter, and power-law degree distribution. By now there are several mathematical models that describe well these properties; examples of such models are Kroenecker graphs~\cite{Kronecker} and geometric protean graphs~\cite{geo-p}. The difference between this body of work and the model presented here is that these models on the social network graph a) do not explicitly model and analyze community structures, and b) focus on macroscopic properties of the social network graph rather than microscopic properties of communities in social networks. 

 There is also a large body of work on community detection algorithms (see for example~\cite{survey_clustering} for a survey) including minimum-cut methods, hierarchical clustering, Givran-Newman algorithm, modulartiy maximization, spectral clustering, and many more. In this paper we are not so much concerned with detecting communities (clusters) in a social network (complex graph), but with modeling and characterizing the community structure that emerges in an information network. An interesting approach to community detection is taken in~\cite{game} by Chen et al. who use a game-theorectic approach to detect overlapping communities in social networks. In~\cite{game} it is assumed that there already exists an underlying social graph for the social network, and the utility that agents obtain when joining a community depends on a) which other agents joined the community and b) the underlying social graph. This is different from the approach in this paper where the utility depends on content production and consumption. In addition, the goal of the work in~\cite{game} is to develop an algorithm to detect overlapping communities, whereas in this paper we aim at characterizing the microscopic structure of information communities.

Related to the analysis in this paper is the work on content forwarding and filtering in social networks~\cite{goel,filtering,hegde}. In particular the work by Zadeh, Goel and Munagala~\cite{goel}, and the work by Hegde, Massoulie, and Viennot~\cite{hegde}. In~\cite{goel}, Zadeh, Goel and Munagala consider the problem of information diffusion in social networks under a broadcast model where content forwarded (posted) by a user is seen by all its neighbors (followers, friends) in the social graph. For this model, the paper~\cite{goel} studies whether there exists a network structure and filtering strategy that leads to both high recall and high precision. High recall means that all users receive all the content that they are interested in, and high precision means that all users only receive content they are interested in. The main result in~\cite{goel} shows that this is indeed the case under suitable graph models such as for example Kronecker graphs. In~\cite{hegde},  Hegde, Massoulie, and Viennot study the problem where users are interested in obtaining content on specific topics, and study whether there exists a graph structure and filtering strategy that allows users to obtain all the content they are interested in. Using a game-theoretic framework (flow games), the analysis in~\cite{hegde} shows that under suitable assumptions there exists a Nash equilibrium, and selfish dynamics converge to a Nash equilibrium. The main difference between the model and analysis in~\cite{goel,hegde} and the approach in this paper is that model and analysis in~\cite{goel,hegde} does not explicitly consider and model community structures,  and the utility obtained by users under the models in~\cite{goel,hegde} depends only on the content that agents receive, but not on the content agents produce. The work by Gupta et al.~\cite{filtering} does take into account how content forwarding (re-posting) by a user affects the utility of neighbors (followers, friends), and studies how different filtering strategies  affect content dissemination. The main result in~\cite{filtering} shows that under suitable social graph models (random graph models) content propagation exhibits a threshold behavior: ``high-quality'' content spreads throughout the network whereas ``low-quality'' content does not. Again, this work does not explicitly consider and model community structures. Despite these differences, we believe that the analysis and results presented in~\cite{goel,filtering,hegde} are very relevant to the model and results presented in this paper, and can potentially be used to extend our results presented here to study content filtering and distribution. Doing this extension is interesting future research.

There exists an interesting connection between the modeling assumption made by Zadeh, Goel and Munagala in~\cite{goel} and by Hegde, Massoulie, and Viennot in~\cite{hegde}, and a result obtained in~\cite{continuous_model_ita,continuous_model_arxiv} and in this paper (Proposition~\ref{prop:dnash}, Section~\ref{section:dnash}). Both papers~\cite{goel,hegde} make the modeling assumption that users produce content only on a small subset of content that they are interested in receiving. Zadeh, Goel and Munagala support this assumption in~\cite{goel} through experimental results obtained on Twitter data that shows that Twitter users indeed tend to produce content on a narrower set of topics than they consume. The results presented in~\cite{continuous_model_ita,continuous_model_arxiv} and in this paper provide a formal validation/explanation for this assumption as it shows that under the proposed model it is optimal for agents (users) to produce content on a small subset of the content type that they are interested in consuming. This result illustrates that the proposed model is able to capture and explain important microscopic properties of information networks and communities.

\section{Mathematical Model}\label{section:content_model}
In this section we introduce the mathematical model that we use for our analysis. The model is an extension of the model in~\cite{continuous_model_ita,continuous_model_arxiv} to the case of a discrete agent population.
We first present the model for the content that agents can produce in an information community. We assume that each content item that is being produced belongs to a particular content type. One might think of a content type as a topic, or particular interest that agents have. Furthermore, we assume that there exists  a structure that relates the different content types with each other. In particular, we assume that  there exists a measure of ``closeness'' between content types that characterizes how strongly related two content types are. For example ``basketball'' and ``baseball'' are both sports and one would assume that these two topics are stronger related with each other than for example ``basketball'' and ``mathematics''.

More formally, we assume that the type of a content item is given by a point $x$ in a metric space $\mS$. The assumption that content types lie in a metric space allows for a natural way to compare the closeness between content types. In particular, the distance between two content types $x,y \in \mS$ is then given by the distance measure $d(x,y)$, $x,y \in \mS$, of the metric space $\mS$.


We next model the content that a given agent is interested in. The model is based on the following intuition. We assume that agents have a main interest, i.e. for each agent there exists a content type $x \in \mS$ that they are most interested in. Moreover, agents are interested in more than one content type, i.e. agents are not only interested in getting information for their main interest $x \in \mS$ but also other type of content. However, the further away a given content type is from their center of interest, the less interested they are in this content type.

To model this situation, we associate with each agent that consumes content a center of interest $y \in \mS$. The center of interest of a given agent is the content type (topic) that an agent is most interested in. Given the center of interest  $y \in \mS$ of an agent, the interest of agent $y$ in content of type  $x \in \mS$ is then given by
$$ p(x|y) = f(d(x,y)),$$
where $d(x,y)$ is the distance between the center of interest $y$ and the content type $x$, and $f:[0,\infty) \mapsto [0,1]$ is a non-increasing function. The interpretation of the function $p(x|y)$ is as follows: when agent $y$ reads a  content of type $x$, then it finds it interesting with probability $p(x|y)$ given by
$$p(x|y) = f(d(x,y)).$$
As the function $f$ is non-increasing, this model captures the intuition that agent an $y \in \mS$ is more interested in content that is close to its center of interest $y$. 

We will refer in the following to an agent by its center of interest $y \in \mS$, i.e. when we refer to an agent $y \in \mS$ we refer to the agent whose center of interest is equal to $y$. 

Next consider a given agent that is producing content, and let $y \in \mS$ be the center of interest of the agent. The ability of agent $y$ to produce content of type $x \in \mS$ is then given by
$$ q(x|y) = g(d(x,y)),$$
where $g:[0,\infty) \mapsto [0,1]$ is a non-increasing function.  The interpretation of this function is as follows. If agent $y$ produces content of type $x$, then the content will  be relevant to content type $x$ with probability $q(x|y)$ given by
$$q(x|y) = g(d(x,y)).$$
As the function $g$ is non-increasing, this model captures the intuition that agent an $y \in \mS$ is better at producing content that is close to its center of interest $y$.

For our analysis, we make additional assumptions on the functions $f$ and $g$ that we use to define the functions $p(\cdot | y)$ and   $q(\cdot | y)$. To state these assumptions, we introduce a few more definitions.

Throughout the paper, we use the following notation. Given a real-valued function $f: \mS \mapsto R$ on a metric space $\mS$, we define the support  $supp(f)$ by
$$supp(f) = \bar A$$
where
$$A = \{ x \in \mS | f(x) \neq 0\},$$
and $\bar A$ is the closure of $A$.

Given real-valued function $f: \mS \mapsto R$ on a metric space $\mS$, we say that $f$ is symmetric with respect to  $y \in \mS$ if for $x,x' \in \mS$ such that
$$ d(x,y) = d(x',y),$$
we have that
$$ f(x) = f(x').$$

Using these definitions, we make the following assumptions for the function $f$ and $g$.
\begin{assumption}\label{ass:fg}
  The function $f: [0,L] \mapsto [0,1]$ is strictly decreasing and three times continuously differentiable on $[0,L]$, the first three derivatives are bounded first derivative on $[0,L]$, and we have that
  $$f'(0) < 0.$$  
Furthermore, the function $f$ is locally strictly concave, i.e. there exists a constant $b$, $0 < b \leq L$, such that 
$$f''(x) < 0, \qquad x \in [0,b].$$
The function $g: [0,L] \mapsto [0,1]$ is non-increasing on $[0,L]$, and strictly concave and twice  continuously differentiable on its support $supp(g)$ with
$$g(0) > 0$$
and
$$g'(0) = 0.$$
\end{assumption}
These assumptions on the functions $f$ and $g$ are a technical assumptions used in the proofs of our results.

\subsection{Content Space $\setR$}\label{section:setR}
In addition, we consider for our analysis a metric space that has a particular structure. More precisely, we consider a one dimensional metric space with the torus metric. The reason for using  this structure is that it simplifies the analysis and allows us to obtain simple expressions for our results, that can easily been interpreted. 

More formally, we consider in the following one-dimensional metric space  for our analysis. The metric space is given by an interval  $\setR = [-L, L) \in R$, $0 < L$, with the torus metric, i.e. the distance between two points $x,y \in \setR$ is given by 
$$d(x,y) = ||x-y|| = \min \{ |x-y|, 2L - |x-y|\},$$
where $| x |$ is the absolute value of $x \in (-\infty,\infty)$.

Note that we have that
$$||x-y|| \leq L, \qquad x,y \in \setR.$$
Furthermore, we have the following two properties,
\begin{enumerate}
\item for $x,y \in \setR$, the addition of $x$ and $y$ is given by
$$ x+ y =
\left \{ \begin{array}{ll}
x +y, & x+y \in [-L,L). \\
-2L + x + y, & x+y \geq L. \\
2L + x + y, & x+y < -L.
\end{array} \right .$$
\item  we have that
$$x < y, \qquad x,y \in \setR.$$
if there exists a point $b$, $b \in (0,L)$, such that
$$ x = y - b.$$
\end{enumerate}

Using the torus metric for the content space $\setR$ eliminates ``border effects'', in the sense are no points that have a ``special'' position as it would be for example the case if we would an interval $[-L,L]$ as the content space. This simplifies the analysis, and leads to simpler expressions for our results.

\section{Agent Population Model}\label{section:agent_model}

We assume that there exists a finite number of agents that exclusively produce content, as well as a finite number of agents that exclusively consume content. More precisely, the set of agents that consume content is given by finite set $\setAd \subset \setR$, of $\Kd$, $\Kd \geq 2$, agents that consume content, and a set $\setAs \subset \setR$, of $\Ks$, $\Ks \geq 1$, agents that produce content. 
The subscript $d$ in $\Cdd$  refers to ``demand'', and the  subscript $s$ in $\Cds$ refers to ``supply''.
For our analysis, we assume that the agents that consume content are ``uniformly distributed'' at distance $\dld$,
$$ \dld = \frac{2L}{\Kd},$$
over $\setR$, and we denote this set of agents by $\Add$. More precisely, the set $\Add$ of agents that consume content is given by 
$$\setAd = \Add = \{y_1,...,y_{\Kd}\}$$
such that
$$ y_{k+1} = y_k + \dld, \qquad k=1,...,\Kd-1.$$

Similarly, we assume that the agents that produce content are ``uniformly distributed'' at distance $\dls$,
$$ \dls = \frac{2L}{\Ks},$$
over $\setR$, and we denote this set of agents by $\Ads$. More precisely, the set $\Ads$ of agents that produce content is given by 
$$\setAs = \Ads = \{y_1,...,y_{\Ks}\}$$
such that
$$ y_{k+1} =  y_k + \dls, \qquad k=1,...,\Kd-1.$$

While we consider here the case where there exists a set of content producers, and a set content consumers, the results in this paper can easily be extended to the case where each agent both produces and consumes content. The results obtained in this paper also hold for this case, requiring only notational changes in the proofs.

\subsection{Information Community}\label{section:information_community}
Having defined the set $\Add$ of agents that consume content, and the set $\Ads$ of agents that produce content, we model a information community as follows. An information community $\dC = (\Cdd,\Cds)$ is defined  by the set $\Cdd\subseteq \Add$ of agents that consume content in the community $\dC$, and the set $\Cds \subseteq \Ads$ of agents that produce content in the community $\dC$. 

For a given a community $\dC= (\Cdd,\Cds)$, we assume that each agent $y \in \Cds$  can decide how much effort it puts into producing content of type $x$. We model this situation as follows. We let $\bCd(x|y)$, $x \in \setR$,  be the rate at which agent $y$ produces content of type $x$ in community $\dC$ where $\bCd(\cdot|y)$ is a non-negative function, i.e. for $y \in \Cds$ we have that
$$\bCd(x|y) \geq 0, \qquad x \in \setR.$$
The total rate (over all agents $y \in \Cds$) at which content of type $x$ is generated in the community $\dC$ is then given by
$$\bCd(x) = \sum_{y \in \Cds}  \bCd(x|y).$$
Recall that content of type $x$ that is generated by agent $y \in \Cds$ is relevant to $x$ with probability $q(x|y)$, and that total rate (over all agents $y \in \Cds$) at which relevant content is generated in $C$ is equal to
$$\QCd(x) = \sum_{y \in \Cds}  \bCd(x|y) q(x|y).$$
We refer to $\QCd(x)$ as the content supply function of the community $\dC$.

For the content consumption we assume that each $y \in \Cdd$ can decide on the fraction of  time it allocates to consume content that is being produced in community $\dC$. Let $\aCd(y)$, $0 \leq \aCd(y) \leq 1$,  be the fraction of time with which agent $y \in \Cdd$ consumes content in community $\dC$, and let the function $\PCd(x)$ be given by
$$\PCd(x) = \sum_{y \in \Cdd}  \aCd(y) p(x|y).$$
We refer to $\PCd(x)$ as the content demand function of  community $\dC$.

For the reward and cost for consuming content, we make the following assumptions. Agents pay a cost $c$ for consuming a content item, where $c$ is a processing cost that reflects the effort/time required by an agent to  process (read) a content item (and decide whether it is of interest or not). If the content item is of interest, then the agent receives a reward equal to 1; otherwise the agent receives a reward equal to 0. The  rate $\mu_C(x|y)$ at which an agent $y \in \Cds$ receives content of type $x$ in community $\dC$ that is of interest to agent $y$ is then given by
\begin{eqnarray*}
\mu_{\dC}(x|y) &=& \aCd(y) \sum_{z \in \Cds} \bCd(x|z) q(x|z) p(x|y) \\
&=& \aCd(y) \QCd(x) p(x|y).
\end{eqnarray*}
The rate at which agent $y$ reads content of type $x$ in the community $C$ is given by 
$$\aCd(y)\bCd(x).$$
As agent $y$ pays a cost of $c$ for each content item, the cost rate that is incurred on agent $y$ is
$$\aCd(y)\bCd(x).$$
Combining the above results, the (time-average) utility  rate (``reward minus cost'') for content consumption $\UCdd(y)$ of agent $y \in \Cdd$ in community $\dC$ is given by
$$ \UCdd(y) = \aCd(y) \int_{\setR} \Bsbl \QCd(x) p(x|y) - \bCd(x) c \Bsbr dx.$$

The (time-average) utility rate for content production  $\UCsd(y)$ of agent $y \in \Cds$ in community $\dC$ is given by
\begin{eqnarray*}
\UCsd(y) = \int_{\setR} \bCd(x|y) \sum_{z\in \Cdd} \aCd(z) \Bsbl  q(x|y) p(x|z)- c \Bsbr dx \\
= \int_{\setR} \bCd(x|y) \Bsbl q(x|y) \PCd(x) - \aCd c \Bsbr dx.
\end{eqnarray*}
where
$$\aCd = \sum_{y \in \Cdd} \aCd(y).$$
This utility rate has the following interpretation. Note that $[q(x|y) p(x|z)- c]$ is the expected reward that agent $z \in \Cdd$ receives from content of type $x$ that is being produced by agent $y \in\ \Cds$. Therefore in an economic setting  $[q(x|y) p(x|z)- c]$ is the amount (price) that agent $z$ is willing to pay agent $y$ for obtaining from agent $y$ content of type $x$. In this sense, one interpretation of the utility rate of a content producer $y \in \Cds$ is that it reflects  the revenue that $y$ would obtain for the content that $y$ produces in the community $\dC$. An alternative interpretation, and the one we adopt in this paper, is that the utility rate of content producer $y \in \Cds$ reflects the reputation, or ``reputation score'', of agent $y$ in the community  $\dC$, i.e. it captures how beneficial the contributions of  a content producer $y$ are for the community $\dC$.

\subsection{Community Structure}\label{section:dnash}
In the previous section  we introduced the model of a community in an information network, and characterized the utility rates that agents obtain when consuming  and producing content  in a given community. Next, we consider the situation where there are agents form several communities in an information network. In particular, we define for this situation how agents decide on how to produce and consume content in the different communities. We refer to the set of communities in an information network, combined with the characterization of how  agents decide on how to produce and consume content in the different communities, as a community structure in an information network.

More precisely, a community structure in an information network is given by a triplet $\CS = \dcommStruc$  where $\sCd$ is a set of communities $\dC = (\Cdd,\Cds)$ as defined in the previous section, and 
$$\aSCd(y) = \{ \aCd(y) \}_{\dC \in \setC} \mbox{ and } \bSCd(y) = \{ \bCd(\cdot |y) \}_{\dC \in \setC}$$
indicate the  rates that content consumers and producers allocate to the different communities $\dC \in \sCd$.

In the following require that each agent is part of at least one community in a community structure. We refer to such a community structure as a covering community structures. More formally, we assume that each agent $y \in \Add$, $y \in \Ads$ belong to at least one community $\dC \in \sCd$, and we have that
$$\cup_{\dC \in \sCd} \Cdd = \Add$$
and
$$\cup_{\dC \in \sCd} \Cds = \Ads.$$
Furthermore, we assume that the total content consumption and production rates of each agent can not exceed a given threshold, and we have that 
$$|| \aSCd(y) || = \sum_{\dC \in \sCd} \aCd(y) \leq E_p, \qquad y \in \Add,$$
where
$$0 < E_p \leq 1,$$
and
$$|| \bSCd(y) || = \sum_{\dC \in \sCd} || \bCd(\cdot|y) || \leq E_q, \qquad y \in \Ads,$$
where
$$ || \bCd(\cdot|y) || = \int_{x \in R}  \beta_C(x|y) dx$$
and
$$0 < E_q.$$
Finally, we require that for a given community structure   $\dcommStruc$, we have for $\dC=(\Cdd,\Cds) \in \sCd$ that
$$ \aCd(y) > 0, \qquad y \in \Cdd,$$
and
$$|| \bCd(y)|| > 0, y \in \Cds,$$
i.e. each agent  $y \in \Cdd$ has a positive content consumption rate $\aCd(y)$ in community $C$, and each agent $y \in \Cds$ has a positive total content production rate $|| \bCd(y)||$ in community $\dC$.

\subsection{\erNASH}
Having defined a community structure in an information network, we next consider the situation where agents consume and produce content in the different communities  in order to maximize their utility rates. For this situation we use a game-theoretic approach to characterize the community structures that emerge in an information network. In particular, we  characterize the community structure that emerges using the concept of an approximate Nash equilibrium, or a \erNashn. To do that, we describe how agents decide on  how to consume and produce content in the different communities  in order to maximize their utility rates.

Given a community structure $\CSd = \dcommStruc$, let $\UCSdd(y)$, $y \in \Add$, the be the total utility rate (over all communities) that agent $y$ receives under this community structure. More precisely, let $\UCSdd(y)$, $y \in \Add$, be given by
$$ \UCSdd(y) = \sum_{\dC \in \setC} \UCdd(y) =  \sum_{\dC \in \sCd} \aCd(y) \int_{\setR} \Bsbl \QCd(x) p(x|y) - \bCd(x) c \Bsbr dx$$
where
$$ \QCd(x) = \sum_{y \in \intsd} \bCd(x|y) q(x|y),$$
and
$$ \bCd(x) = \sum_{y \in \intsd} \bCd(x|y).$$
Similarly,
let $\UCSsd(y)$, $y \in \Ads$, the be the total utility rate (over all communities) that agent $y$ receives under this community structure. More precisely, let $\UCSsd(y)$, $y \in \Ads$, be given by
$$\UCSsd(y) = \sum_{\dC \in \sCd} \UCsd(y) =  \sum_{\dC \in \setC} \int_{\setR} \bCd(z|y) \Bsbl q(x|y)\PCd(x) - \aCd c \Bsbr dx $$
where
$$\PCd(x) = \sum_{y \in \intsd} \aCd(y) p(x|y),$$
and
$$ \aCd = \sum_{y \in \intsd} \aCd(y).$$

Using these definitions, we  analyze for a given a community structure
$\CSd$ 
the situation where an agent $y \in \Add$ changes its rate allocation from $\aSCd(y)$ to $\aSCd'(y)$ given by
$$\aSCd'(y)=  \{ \aCd'(y) \}_{\dC \in \sCd}$$
such that
$$ \sum_{\dC \in \sCd} \aCd'(y) \leq E_p.$$
More precisely,
let $ \UCSdd(\aSCd'(y)|y)$ be the utility that agent $y \in \Add$ obtains under the new allocation $\aSCd'(y)$ (while all other agents keep their rate allocation fixed) given by
$$  \UCSdd(\aSCd'(y)|y) =
\sum_{\dC \in \sCd} \aCd'(y) \int_{x \in \mS} \left [ \sum_{z \in C_s}  \bCd(x|z) \Bsbl q(x|z) p(x|y) - c \Bsbr \right ] dx.$$



Similarly, given a community structure $\CSd$ we  analyze the situation where an agent $y \in \Ads$ changes its rate allocation $\bSCd(y)$ to $\bSCd'(y)$ given by
$$\bSCd'(y)=  \{ \bCd'(\cdot|y) \}_{\dC \in \sCd},$$
such that
$$ || \bSCd'(y) || \leq E_q.$$
Let $ \UCSsd(\bSCd'(y)|y)$ be the utility rate that agent $y$ receives under the new allocation $\bSCd'(y)$ (while all other agents keep their rate allocation fixed) given by
$$  \UCSsd(\bSCd'(y)|y) = \sum_{\dC \in \sCd} \int_{x \in \mS}  \bCd'(x|y) \left [ \sum_{z \in C_d} \aCd(y) \Bsbl q(x|y)p(x|z)  - c \Bsbr \right ] dx.$$

Given a community structure  $\CSd = \dcommStruc$, ideally agents want to choose allocations for consuming and producing content in order to maximize their utility rates.
Here we use a slightly weaker criteria of  an approximate Nash equilibrium, or a \erNash, where  agents change their current allocations only if the new allocations provides an increase in their utility rate at least by a factor $\epsilon$, $\epsilon>0$.


More formally, we call a community structure
$$\CSdNash = \dcommStrucNash$$
a \erNash if  
\begin{enumerate}
\item[a)] for all agents  $y \in \Add$ we have that
$$\UCSddNash(\aSCd(y)|y) - \UCSddNash(y) < \epsilon,$$ 
where
$ \aSCd(y) = \arg \max_{\aSCd'(y): || \aSCd'(y) || \leq E_p}  \UCSddNash(\aSCd'(y)|y)$, 
\item[b)] for all agents $y \in \Ads$ we have that
$$\UCSsdNash(\bSCd(y)|y) - \UCSsdNash(y) < \epsilon,$$
where
$ \bSCd(y) = \arg \max_{\bSCd'(y): || \bSCd'(y) || \leq E_q}  \UCSsdNash(\bSCd'(y)|y)$.
\end{enumerate}
In the following we study whether there exists a \erNash, and characterize the community structure of a \erNashn. For our analysis we consider a  particular metric space, and agent population model, that we describe in the next subsection.

\section{Discrete Interval Community and Community Structure $\dCSd$}\label{section:dinterval_community}
For our analysis of an \erNash, we focus on a particular class of community structures, to which we refer as a discrete interval community structure$\dsCd$. In this section, we formally introduce this class of community structures.

\subsection{Discrete Interval Community}
To formally define a discrete interval community structure $\dsCd$, we use the following definitions. Given an interval $I_C$ in $\setR$, let
$$|I_C| = \int_{I_C} dy$$
be the length of the interval, let
$$ L_C = \frac{|I_C|}{2}$$
be the half-length of the interval $I_C$, and let
$$mid(I_C) = \frac{1}{I_C} \int_{I_C} z dy$$
be the midpoint of the interval $I_C$.
Note that when $I_C$ is a closed interval then we have that
$$I_C = [ mid(I_C) - L_C, mid(I_C) + L_C ].$$

Given an interval $I_C$ in $\setR$ with half-length $L_C$ and midpoint $mid(I_C)$,
a discrete set  $\intd$ on $I_C$ with distance $\dl$, $\dl>0$, is then given by a set $\intd$, 
$$ \intd = \{y_1,...,y_K\},$$
such that
\begin{enumerate}
\item[a)] $K \geq 2$,
\item[b)] $ y_{k+1} - y_k = \dl$, $k=1,...,K-1,$ or
  $$ y_{k+1} = y_k + \dl, \qquad k=1,...,K-1,$$ 
\item[c)] $ ||y_1 - (mid(I_C) - L_C)|| \leq \dl$, and
\item[d)]$||y_K - (mid(I_C) + L_C)||  \leq \dl$.
\end{enumerate}
Note that this definition implies that the $K$ points in $\intd$ are evenly spread out over the interval $I_C$ with distance $\delta$ between neighboring points. Furthermore, we have that the endpoints $y_1$ and $y_K$ in $\intd$ are at a distance of at most $\delta$ from the endpoints $mid(I_C) - L_C$ and $mid(I_C)+L_C$ of the interval $I_C$. 

Note that given an interval $I_C \subset \setR$, and the set of agents $\Add$ and $\Ads$ that consume and produce content, respectively, we  have that
$$\intdd = \Add \cap I_C$$
is a discrete set on $I_C$ with distance $\dld$, and
$$\intsd = \Ads \cap I_C$$
is a discrete set on $I_C$ with distance $\dls$.

Given an interval $I_C$ on $\setR$ and a discrete set
$$\intd = \{y_1,...,y_K\}$$
on $I_C$ with distance $\dl$,
the midpoint $\midd$ of $\intd$ is given by
$$\midd =  \frac{1}{K} \sum_{k=1}^K y_k = \frac{1}{K} \sum_{y \in \intd } y = \frac{y_1 + y_K}2.$$
Furthermore, the half-length $\Ld$ of the discrete set $\intd$ is given by
$$ \Ld = \frac{y_K - y_1}{2},$$
and we have that
$$ y_1 = \midd - \Ld$$
and
$$y_K = \midd + \Ld.$$
Finally, note that we have  that 
$$\Ld \leq L_C$$
and
$$ | \Ld - L_C | \leq \dl.$$

Using the definition of an discrete interval set, we define a discrete interval community as follows. 
Given an interval  $I_C$  on $\setR$, a discrete interval community  on $I_C$ with distance $\dl$, $\dl>0$, is then given by a community $\dC= (\intdd,\intsd)$ such that both sets  $\intdd$ and $\intsd$ are discrete sets on $I_C$ with distance smaller than $\dl$.

In the following, we say that a community $\dC = \Cd$ is a discrete interval community with distance $\dl$, if there exists a interval $I_C$ on $\setR$ such that $\dC=\Cd$ is a discrete interval community on $I_C$ with distance $\dl$. 

Similarly, we say that a community structure $\dcommStruc$ is a discrete interval community structure with distance $\dl$  if every community in $\dC \in \sCd$ is an interval community with distance $\dl$. 

Finally, we say that a \erNash $\dcommStrucNash$ is a discrete interval \erNash  with distance $\dl$ if $\dcommStrucNash$ is a discrete interval community with distance $\dl$.

\subsection{Community Structure $\dCSd$}
Using the above notation, we define the class $\dCSd$ of discrete interval community structures with distance $\dl$ as follows.

\begin{definition}
Let $L_C$ be such that
$$0 < 2L_C < \min \{ b, L \},$$
where  $b$ is the constant of Assumption~\ref{ass:fg}. The class $\dCSd$ then consists of all community structures
$$\dcommStruc$$
where
$$\sCd = \{\dCk \}_{k=1,...,K}$$
that are defined on a set of agents $\Add$ with distance $\dld$ such that
$$ 0 < \dld < \dl,$$
and set of agents $\Ads$ with distance $\dls$ such that
$$0 < \dls < \dl,$$
and have the following properties.

There exists set a $\{I_{C_k} \}_{k=1,...,K}$  of mutually non-overlapping intervals in $\setR$ of length $2 L_C$, i.e. we have that 
\begin{enumerate}
\item[1)] $ I_{C_k} \cap I_{C_{k'}}  = \emptyset$, $k \neq k'$, 
\item[2)] $\cup_{k=1,..,K} I_{C_k} = \setR$,
\item[3)] $ |I_{C_k}| = 2 L_C$, $k=1,..,K$,
\end{enumerate}
such that the community $\dCk = \Cd \in \sCdNash = \{\dCk \}_{k=1,...,K}$, is given by
  $$\intdd = \Add \cap I_{C_k}$$
  and
  $$\intsd = \Ads \cap I_{C_k}.$$
  Furthermore, for $\dC = \Cd \in \sCd = \{\dCk \}_{k=1,...,K}$ we have that
\begin{enumerate}
\item[a)] $\aCd(y) = E_p$, $y \in \intdd$, and
\item[b)] $ \bCd(x|y) = E_q \delta(\xd(y) - x)$, $y \in \intsd, x \in \setR$, 
where
$$ \xd(y) = \arg \max_{x \in \setR} \left [ q(x|y)\PCd(x) \right ] = \arg \max_{x \in \setR} \left [ q(x|y) E_p \sum_{y \in \intdd} p(x|y) \right ]$$
and $\delta(\cdot)$ is the Dirac delta function.  
\end{enumerate}
\end{definition}  
Note that each community structure
$$\CSd = \dcommStruc \in \dCSd$$
consists of a set of discrete interval communities on an interval of equal length of $2 L_C$. Furthermore, we have that each agent $y \in \Add$, and each agent $y \in \Ads$, belongs to exactly one community in $\sCd$. Finally, given a community $\dC = \Cd \in \sCd$, each agent $y \in C_d$ allocates all its effort to consume content in the community $C$, and each each $y \in C_s$ allocates all its effort to produces content of type $\xd(y)$.

In the next section, we show that there always exists a  class $\dCSd$ of discrete interval community structures with distance $\dl$ such that all community structures $\CSd \in \dCSd$ are a \erNashn.

\section{Results}\label{section:results}
In this section we present  the main results of our analysis. The proofs for the results are given in the appendix.

Our first results shows that there always exists a \erNashn. 

\begin{prop}\label{prop:dnash}
Suppose that
$$ f(0) g(0) - c > 0,$$
then there exists a $L_C$,
$$0 < 2L_C < \min \{ b, L \},$$
where  $b$ is the constant of Assumption~\ref{ass:fg}, such that the following is true. For every $\epsilon > 0$ there exists a $\dl > 0$ such that all discrete interval community structures  $\CSdNash \in \dCSd$ are a \erNashn.
\end{prop}

Proposition~\ref{prop:dnash} states that there always exists a \erNash $\CSdNash$ given that  distance $\dld$, and $\dls$, of the agents sets $\Add$, and $\Ads$, are small enough, i.e. we have that $\dld,\dls < \dl$,
where $\dl$ is as given  in the statement of Proposition~\ref{prop:dnash}.

The result of Proposition~\ref{prop:dnash} that a \erNash consists of discrete interval communities captures that intuition that the content that is being produced and a consumed in a community should be ``aligned'' in the sense that the content producers in the community  should have a high ability to produce content that the content consumers in the community are most interested in.

Furthermore, from the definition of the class $\dCSd$ we have that under a Nash equilibrium as given in Proposition~\ref{prop:dnash} each content producer $y \in \intsd$ in a given community $\dC = \Cd$ focuses on producing a single type of content given by
$$\xd(y) =  \arg \max_{x \in \setR} q(x|y)\PCd(x).$$
This result is interesting as experimental results suggest that this property indeed holds in real-life information networks. We will discuss this in more details in Section~\ref{subsection:discussion_results}.

In the following we characterize in more details the properties of a \erNash as given by Proposition~\ref{prop:dnash}.

\subsection{Optimal Content Type $\xd(y)$}
 Proposition~\ref{prop:dnash} states that under a discrete interval \erNash as given in Proposition~\ref{prop:dnash}, each agent $y \in \Ads$ produces a single type of content. More precisely, if $\dC = \Cd$ is a community in a discrete interval \erNash $\CSdNash$ as given by Proposition~\ref{prop:dnash}, then we have for
$y \in \intsd$
that
$$ \bCd(x|y) = E_q \delta(\xd(y) - x), \qquad x \in \setR,$$
where
$$ \xd(y) = \arg \max_{x \in \setR} q(x|y)\PCd(x).$$

In this subsection we characterize in more details the function $\xd(y)$  for a given community $\dC = \Cd$ in a \erNashn.
In particular, we have the following result.
\begin{prop}\label{prop:xd}
Let $\dCSd$ be a class of discrete interval community structures with distance $\dl$ as given by Proposition~\ref{prop:dnash}, i.e. we have that all community structures  $\CSdNash \in \dCSd$ are a \erNashn.
Then for  every $\Dl_{x^*}$, $0 < \Dl_{x^*} < L_C$, there exists a class $\dCSdsub \subseteq \dCSd$, $0 < \dl_0 \leq \dl$, of discrete interval community structures with distance $\dl_0$ such that for all community structures
$$\dcommStrucNash \in \dCSdsub$$
the following is true.
Given a community $\dC = \Cd \in \sCdNash$, let the interval $\Id  =   [ mid(I_C) - L_C, mid(I_C) + L_C) \subset \setR$ as given in Proposition~\ref{prop:dnash}, i.e. we have that
$$\intdd = \Add \cap \Id
\mbox{ and }\;
\intsd = \Ads \cap \Id.$$
Then the solution $x^*_{\dl}(y)$ to the optimization problem 
$$ \xd(y) = \arg \max_{x \in \setR} q(x|y) \PCd(x), \qquad y \in \Id,$$
where
$$\PCd(x) = E_p \sum_{y \in \intdd} p(x|y),$$
has the properties that
\begin{enumerate}
\item[(a)] there exists a unique optimal solution $\xd(y)$ for $y \in \Id$.

\item[(b)] for $y \in [mid(\Id) - \LCd, mid(\Id) - \Dl_{x^*}]$, we have that 
$$ \xd(y)  \in (y, mid(\Id)) \cap supp(q(\cdot|y)).$$

\item[(c)] for  $y \in [mid(\Id) + \Dl_{x^*},mid(\Id) + \LCd]$,  we have that 
$$ \xd(y)\in (mid(\Id),y) \cap supp(q(\cdot|y)).$$

\item[(d)] the function $\xd(y)$ is strictly increasing and differentiable on 
$$[ mid(\Id) - \LCd, mid(\Id) - \Dl_{x^*}] \cup [mid(\Id) + \Dl_{x^*}, mid(\Id) + \LCd].$$
\end{enumerate}
\end{prop}

Proposition~\ref{prop:xd}  states that the function $\xd(y)$ is strictly increasing on $\Id \backslash (mid(\Id) - \Dl_{x^*}, mid(\Id) + \Dl_{x^*})$. This result implies that two different agents $y,y' \in \Id \backslash (mid(\Id) - \Dl_{x^*}, mid(\Id) + \Dl_{x^*})$, $y \neq y'$, produce different type of contents, i.e. we have that
$$\xd(y) \neq \xd(y'), \qquad y \neq y', y,y' \in \Id \backslash (mid(\Id) - \Dl_{x^*}, mid(\Id) + \Dl_{x^*}).$$
This result is interesting as it states that under each agent $y$ in $\Id \backslash (mid(\Id) - \Dl_{x^*}, mid(\Id) + \Dl_{x^*})$ produces a unique content type $\xd(y)$, i.e. we have that the content type $\xd(y)$  is not produced by any other agent in $\Id \backslash (mid(\Id) - \Dl_{x^*}, mid(\Id) + \Dl_{x^*})$.

Another interesting aspect of Proposition~\ref{prop:xd} is that the function $\xd(y)$ has the property that
$$ \xd(y)  \in (y, mid(\Id)) \cap supp(q(\cdot|y))$$
and
$$ \xd(y)\in (mid(\Id),y) \cap supp(q(\cdot|y)).$$
This result states that agents $y \in \Id \backslash (mid(\Id) - \Dl_{x^*}, mid(\Id) + \Dl_{x^*})$  produce content that is closer to the center of interest $mid(\Id)$ of the community $C$ than their center of interest $y$. To get a more detailed understanding of how agents adapt the type of content that they produce towards the center of interest $mid(\Id)$ of the community $\dC$, we next study the function $\Ddy$ given by
$$ \Ddy = || y - \xd(y)||, \qquad y \in \Id.$$
The function $\Ddy$  characterizes the absolute value of the ``displacement'' of the optimal content $\xd(y)$ that agent $y$ produces, and content $y$ that the agent is best at producing which is equal to content type $y$. Or in other words, the function $\Ddy$ characterizes by  how much an agent $y$ adapts its content $\xd(y)$ towards the center of interest of  the community $\dC$, i.e. by how much agent $y$ produces content $\xd(y)$ that is closer to the center of interest $mid(\Id)$ of the community than its own center of interest.

In addition, the function $\Ddy$ can be used to characterize the quality of the optimal content $\xd(y)$ that agent $y$ produces as we have that
$$ q(\xd(y)|y) = g( || y - \xd(y)||) = g(\Ddy).$$

We have the following result for the function $\Ddy$.
\begin{prop}\label{prop:Ddy}
Let $\dCSd$ be a class of discrete interval community structures with distance $\dl$ as given by Proposition~\ref{prop:dnash}, i.e. we have that all community structures  $\CSdNash \in \dCSd$ are a \erNashn.
Then for  every $\Dl_{x^*}$, $0 < \Dl_{x^*} < L_C$, there exists a class $\dCSdsub \subseteq \dCSd$, $0 < \dl_0 \leq \dl$, of discrete interval community structures with distance $\dl_0$ such that for all community structures
$$\dcommStrucNash \in \dCSdsub$$
the following is true.
Given a community $\dC = \Cd \in \sCdNash$, let the interval $\Id  =   [ mid(I_C) - L_C, mid(I_C) + L_C) \subset \setR$ as given in Proposition~\ref{prop:dnash}, i.e. we have that
$$\intdd = \Add \cap \Id
\mbox{ and }\;
\intsd = \Ads \cap \Id.$$

Then the function $\Ddy$ given by
$$ \Ddy = || y - \xd(y)||, \qquad y \in \Id,$$
where
$$\xd(y)= \arg\max_{x \in \setR} q(x|y) \PCd(x),$$
is  strictly decreasing and differentiable on  $[mid(\Id) - \LCd, mid(\Id) -  \Dl_{x^*}]$, and strictly increasing and differentiable on  $[mid(\Id)  + \Dl_{x^*}, mid(\Id)]+\LCd]$.
\end{prop}

Proposition~\ref{prop:Ddy} states that  the function $\Ddy$ is  strictly decreasing and differentiable on  $[mid(\Id) - \LCd, mid(\Id) -  \Dl_{x^*}]$, and strictly increasing and differentiable on  $[mid(\Id)  + \Dl_{x^*}, mid(\Id)]+\LCd]$.
This implies that the further away an agent is from the center of interest $mid(\Id)$, the more it will ``adapt'' the content it produces towards to the center $mid(\Id)$ of the interval $\Id$, i.e. the larger $\Ddy$ will be. In addition, this result implies that the  further away an agent is from the center of interest $mid(\Id)$, the lower the quality is the content that the agents produces. To see this, recall that by Assumption~\ref{ass:fg}, $g$ is decreasing on $supp(g)$ and we have that the quality of the content that agent $y$ produces is given by
$$ q(\xd(y)|y) = g( || y - \xd(y)||) = g(\Ddy).$$

We discuss the results of Proposition~\ref{prop:xd} and Proposition~\ref{prop:Ddy} in more details in Section~\ref{subsection:discussion_results}.

\subsection{Properties of the Content Demand Function $\PCd(x)$}
We next characterize the properties of the content demand function $\PCd(x)$ of a discrete interval community $\dC = \Cd$ under a discrete interval \erNash as given in Proposition~\ref{prop:dnash}. We have the following result.

\begin{prop}\label{prop:PCd}
Let $\dCSd$ be a class of discrete interval community structures with distance $\dl$ as given by Proposition~\ref{prop:dnash},  i.e. we have that all community structures  $\CSdNash \in \dCSd$ are a \erNashn.
Then for  every $\Dl_P$, $0 < \Dl_P < L_C$, there exists a class $\dCSdsub \subseteq \dCSd$, $0 < \dl_0 \leq \dl$, of discrete interval community structures with distance $\dl_0$ such that for all community structures
$$\dcommStrucNash \in \dCSdsub$$
the following is true.
Given a community $\dC = \Cd \in \sCdNash$, let the interval $\Id  =   [ mid(I_C) - L_C, mid(I_C) + L_C) \subset \setR$ as given in Proposition~\ref{prop:dnash}, i.e. we have that
$$\intdd = \Add \cap \Id
\mbox{ and }\;
\intsd = \Ads \cap \Id.$$
Then the demand function
$$\PCd(x)  =  \sum_{y \in I_{\dl,C}} \aCdNash(y) p(x|y) =  E_p\sum_{y \in I_{\dl,C}}  p(x|y), \qquad x \in \setR,$$
 has the properties that
\begin{enumerate}
\item[(a)] $\PCd(x)$ is symmetric with respect to $\middd$. 
\item[(b)] $\PCd(x)$ is strictly increasing on the interval $ [ mid(\Id) - L_C, mid(\Id) - \Dl_P]$, and strictly decreasing on the interval $[mid(\Id) + \Dl_P, mid(\Id) + L_C)$.
\item[(c)] $\PCd(x)$ is strictly concave in $x$ on the interval $I_C$.  
\end{enumerate}
\end{prop}

Note that Proposition~\ref{prop:PCd} implies that 
$$\arg \max_{x \in \setR} \PCd(x) \in [mid(\Id) - \Dl_P, mid(\Id) + \Dl_P],$$
i.e. the most popular content is close to the center of interest of the community $\dC$. Furthermore,  we have that the further away a  content type $x$ is from the center of interest of the community,  the less popular it is.

\subsection{Properties of the Content Supply Function $\QCsd(x)$}
Next characterize the properties of the content demand function $\QCsd(x)$ of a discrete interval community $\dC = \Cd$ under a discrete interval \erNash as given in Proposition~\ref{prop:dnash}. We have the following result.

\begin{prop}\label{prop:QCd}
Let $\dCSd$ be a class of discrete interval community structures with distance $\dl$ as given by Proposition~\ref{prop:dnash},  i.e. we have that all community structures  $\CSdNash \in \dCSd$ are a \erNashn.
Then there exists a class $\dCSdsub \subseteq \dCSd$, $0 < \dl_0 \leq \dl$, of discrete interval community structures with distance $\dl_0$ such that for all community structures
$$\dcommStrucNash \in \dCSdsub$$
the following is true.
Given a community $\dC = \Cd \in \sCdNash$, let the interval $\Id  =   [ mid(I_C) - L_C, mid(I_C) + L_C) \subset \setR$ as given in Proposition~\ref{prop:dnash}, i.e. we have that
$$\intdd = \Add \cap \Id 
\mbox{ and }\;
intsd = \Ads \cap \Id.$$
Then the content supply function $\QCdNash(x)$, $x \in \setR$, is given by 
$$\QCsd(x) = E_q \sum_{y \in \intsd} \dl\big (x-\xd(y)) q(\xd(y)|y \big )$$
where $\xd(y) = \arg \max_{x \in \setR}  q(x|y) \PCd(x)$,
and  $\dl(\cdot)$ is the Dirac delta function, and we have that
$$supp(\QCsd(\cdot)) \subseteq [mid(\Id) - \LSd, mid(\Id) + \LSd]$$
where $0< \LSd < \LCd$.
\end{prop}

Proposition~\ref{prop:QCd} states that
$$supp(\QCsd(\cdot)) \subseteq [mid(\Id) - \LSd, mid(\Id) + \LSd]$$
where $0< \LSd < \LCd$.
This result implies  that the content type that is being produced by agents in the community $\dC$ is a strict subset of the interval $\Id$. As a result, there is no overlap in the content produced in different communities under a discrete interval \erNash as given by Proposition~\ref{prop:dnash}. We discuss this results in more details in Section~\ref{subsection:discussion_results}.

\subsection{Properties of the Utility Function $\UCdd(y)$ and $\UCsd(y)$
}  
Finally, we study the properties of the utility rate function for content consumption $\UCdd(y)$,  and the utility rate function  for content production $\UCsd(y)$ for an interval community $\dC = \Cd$  under a  discrete interval \erNash as given by Proposition~\ref{prop:dnash}. 

We first study the properties of the utility rates for content consumption   $\UCdd(y)$ for  an interval community $\dC = \Cd$  under a discrete interval \erNash as given by Proposition~\ref{prop:dnash}.

\begin{prop}~\label{prop:UCdd}
Let $\dCSd$ be a class of discrete interval community structures with distance $\dl$ as given by Proposition~\ref{prop:dnash}, i.e.  i.e. we have that all community structures  $\CSdNash \in \dCSd$ are a \erNashn.
Then for  every $\Dl_U$, $0 < \Dl_U < L_C$, there exists a class $\dCSdsub \subseteq \dCSd$, $0 < \dl_0 \leq \dl$, of discrete interval community structures with distance $\dl_0$ such that for all community structures
$$\dcommStrucNash \in \dCSdsub$$
the following is true.
Given a community $\dC = \Cd \in \sCdNash$, let the interval $\Id  =   [ mid(I_C) - L_C, mid(I_C) + L_C) \subset \setR$ as given in Proposition~\ref{prop:dnash}, i.e. we have that
$$\intdd = \Add \cap \Id \mbox{ and }\;
\intsd = \Ads \cap \Id.$$
Then the  utility rate function for content consumption $\UCdd(y)$  given by
$$ \UCdd(y) =   E_p E_q \sum_{z \in \intsd} \Big [ p(\xd(z)|y) q(\xd(z)|z) - c \Big ]dx, \qquad y \in \intdd,$$
where
$$\xd(y) = \arg \max_{x \in \setR}  q(x|y) \PCd(x),$$
has the following properties.
\begin{enumerate}
\item[a)]  For
  $$ y,y' \in \intdd \cap [mid(\Id) - \LCd,mid(\Id) - \Dl_U],$$
  such that
  $y > y'$,
  we have that
$\UCdd(y) > \UCdd(y')$.
\item[b)] For
  $$y,y' \in \intdd \cap [mid(\Id) + \Dl_U, mid(\Id) + \LCd],$$
  such that
  $y < y'$,
  we have that
$\UCdd(y) > \UCdd(y')$.   
\end{enumerate}
\end{prop}
Proposition~\ref{prop:UCdd} states that the closer an agent $y \in \intdd$ is to the center of interest of the community, the higher  a higher utility rate it receives. This is an interesting result as it suggest that the utility rate might can be used to rank agents in an information community.  We discuss this in more details in Section~\ref{subsection:discussion_results}.

We next study the properties of the utility rates for content production $\UCsd(y)$ for  an discrete interval community $\dC = \Cd$  under a  discrete interval \erNash as given by Proposition~\ref{prop:dnash}.

\begin{prop}~\label{prop:UCsd}
Let $\dCSd$ be a class of discrete interval community structures with distance $\dl$ as given by Proposition~\ref{prop:dnash}, i.e. we have that all community structures  $\CSdNash \in \dCSd$ are a \erNashn.
Then for  every $\Dl_U$, $0 < \Dl_U < L_C$, there exists a class $\dCSdsub \subseteq \dCSd$, $0 < \dl_0 \leq \dl$, of discrete interval community structures with distance $\dl_0$ such that for all community structures
$$\dcommStrucNash \in \dCSdsub$$
the following is true.
Given a community $\dC = \Cd \in \sCdNash$, let the interval $\Id  =   [ mid(I_C) - L_C, mid(I_C) + L_C) \subset \setR$ as given in Proposition~\ref{prop:dnash}, i.e. we have that
$$\intdd = \Add \cap \Id \mbox{ and }\;
\intsd = \Ads \cap \Id.$$
Then the  utility rate function for content production $\UCsd(y)$ given by
$$ \UCsd(y) =  E_p E_q \Big [  q\big (\xd(y)|y \big ) \PCd(\xd(y))- \aCd c \Big ], \qquad  y \in \intsd,$$
has the following properties.
\begin{enumerate}
\item[a)]  For
  $$ y,y' \in \intdd \cap [mid(\Id) - \LCd,mid(\Id) - \Dl_U],$$
  such that
  $y > y'$,
  we have that
$\UCsd(y) > \UCsd(y')$.
\item[b)] For
  $$y,y' \in \intdd \cap [mid(\Id) + \Dl_U, mid(\Id) + \LCd],$$
  such that
  $y < y'$,
  we have that
$\UCsd(y) > \UCsd(y')$.   
\end{enumerate}
\end{prop}
Similar to  Proposition~\ref{prop:UCdd},  Proposition~\ref{prop:UCsd} states that the closer an agent $y \in \intdd$ is to the center of interest of the community, the higher  a higher utility rate it receives.  Again, this result suggest that the utility rate might can be used to rank agents in an information community.  We discuss this in more details in Section~\ref{subsection:discussion_results}.

\section{Conclusions}~\label{subsection:discussion_results}
In this paper we considered an information network with a discrete set of agents, and characterized for this model community structures that are a \erNash. The analysis extends the results of~\cite{continuous_model_ita,continuous_model_arxiv} that were obtained for a continuous agent population model to the case of a discrete agent population model. In particular, in the limiting case of a very dense agents population $\Add$ and $\Ads$, i.e. when the distances $\dld$ and $\dls$ approach 0, the results presented in this paper recover the results obtained in~\cite{continuous_model_ita,continuous_model_arxiv}. As such, the case of a continuous population model that is easier to analyzed and characterize provides the right intuition for the discrete population model which captures better the real-life information networks.

An interesting aspect of the results in this paper, and inf~\cite{continuous_model_ita,continuous_model_arxiv}, is that they indeed provide insights into properties of communities in real-life information networks. Below we provide an overview of these properties which are essentially the same as obtained in~\cite{continuous_model_ita,continuous_model_arxiv}.  

Proposition~\ref{prop:dnash} states that each content producer in a community focuses on generating exactly one type of content, i.e we have that that an agent $y \in \intsd$ produces the content unique content type $\xd(y)$ given by
$$\xd(y) =  \arg\max_{x \in \setR} q(x|y) \PCd(x), \qquad y \in \intsd.$$
This is an interesting results as this property/behavior has indeed been observed in real-life social networks. In particular, Zadeh, Goel and Munagala provide experimental results using data obtain from  Twitter that shows  Twitter users produce content on a very narrow set of topics, and consume content on a large set of topics~\cite{goel}. The interpretation of this experimental result based on the analysis presented in this paper is that this behavior is indeed optimal, and that the topics for which a given Twitter user produces content is the optimal content for this user to produce.

Proposition~\ref{prop:xd}~and~Proposition~\ref{prop:Ddy} state that agents adapt the content they produce towards  center of interest of a community they belong to. This is an interesting result, as  it provides insight regarding the issue of ``homophily versus adaption'' in social networks. One aspect of communities in social networks is that members of the community tend to have similar interests, and behave similarly. One question is whether is due to homophily, i.e. do individuals form a community because they have the same interest and behave similarly, or whether this is due to adaption, i.e. is due that through the interaction in a community its members develop similar interests and behave similarly. The results of this paper suggest that both mechanism are at work in information communities. In particular, homophily (similar interests) prompts individuals to join the same community. This is captured by the result of Proposition~\ref{prop:xd} which states the communities  of a \erNash as given in Proposition~\ref{prop:dnash} are given by intervals, and hence the center of interests of agents in a given community belong to the same interval and are close together. Or in other words, under a \erNash agents that belong to the same community indeed have similar interests. Moreover, when an content producing agent joins a community then the agent adapts the content it produces from its center of interest towards the center of interest of the community. This is captured by the result of Proposition~\ref{prop:xd} which states that an agent in community $\dC$ does not produce content $y$ for which they have the best ability, i.e. content type that is the center of  interest of the agent, but agent $y$ adapts the content type $\xd(y)$ that is produces towards the center of interest of the community.

Proposition~\ref{prop:PCd}  states that the content type that is the most popular is the content type that is close to the the center of interest of the community. Moreover, it states that the further away a  content type is from the center of interest of a community, the less popular it is in the community. This result is interesting as it confirms in a formal way the intuition that communities have a ``center of interest'', i.e. there is a focus (center of interest) of the community and the closer a type of content (topic) is to that center of interest, the more popular (in demand) the content will be. While this property of a community may seem ``intuitively obvious'', it is interesting to observe that this property is indeed recovered (confirmed) using the proposed model.

Proposition~\ref{prop:QCd}  implies that under a \erNash there is no overlap between the content that is being produced in different communities, i.e. each community produces a distinct set of content types. This is an interesting result as it suggest that each information community can be identified by  unique ``core content'' that is only produced in this community. In particular, this result suggests that each  real-life community has a ``core interest'' that is unique to the community, and hence identifies this community. Moreover, it seems that this property should be useful when trying to discover (identify) communities in information networks, i.e. one should be able to exploit this property to potential design more efficient algorithms to discover communities in information networks. How, and whether, this is possible is future research.

Finally, Proposition~\ref{prop:UCdd}~and~\ref{prop:UCsd} state that  the agents that are close to  the center of interest of a  community  obtain the highest utility rate for content consumption and production; and the further away an agent $y$ is from the center of interest of the community, the smaller a utility it will receive. This result suggest that there is a ``natural'' way to rank members in a community, where a member is ranked higher if it receives a higher utility. In particular, if it is possible to observe the utilities that community members receive, then it is indeed possible to rank community members in this way. How, and when, it is
possible to observe (estimate)  utilities that community members receive, and how this information can be used to design algorithms to analysis information communities,  is future research.

In ongoing research, we use the model and results presented in this paper to study the  connectivity (graph structure) within in a community, as well as information dissemination within a community. In addition  use the model and results presented in this paper to study local social search algorithms to detect core users in a community, as well as local algorithms to detect information communities.

\bibliography{info15}{}
\bibliographystyle{plain}

\appendix














\newpage
\section{Property of Discrete Sets on an Interval $I_C$ with Distance $\dl$}
In this appendix we derive a result for discrete sets on an interval $I_C$ with distance $\dl$ which we use in our analysis. For this, recall the definition of a discrete set on an interval $I_C$ with distance $\dl$. In particular, recall that given an interval $I_C$ on $\setR$,  a discrete set  $\intd$ on $I_C$ with distance $\dl$, $\dl>0$, is given by 
$$ \intd = \{y_1,...,y_K\} \subset I_C$$
such that
\begin{enumerate}
\item[a)] $K \geq 2$,
\item[b)] $ y_{k+1} - y_k = \dl$, $k=1,...,K-1,$
\item[c)] $ ||y_1 - (mid(I_C) - L_C)|| \leq \dl$,
\item[d)]$||y_K - (mid(I_C) + L_C)||  \leq \dl$,
\end{enumerate}
where  $mid(I_C)$ is the mid-point, and $L_C$ is the half-length, of the interval $I_C$.

Given an interval $I_C$ on $\setR$ and a discrete set
$$\intd = \{y_1,...,y_K\}$$
on $I_C$ with distance $\dl$,
recall that the midpoint $\midd$ of $\intd$ is given by
$$\midd =  \frac{1}{K} \sum_{k=1}^K y_k = \frac{1}{K} \sum_{y \in \intd } y = \frac{y_1 + y_K}2.$$
Furthermore, recall that  the half-length $\Ld$ of the discrete set $\intd$ is given by
$$ \Ld = \frac{y_K - y_1}{2},$$
and we have that
$$ y_1 = \midd - \Ld$$
and
$$y_K = \midd + \Ld.$$
Finally, note that we have  that 
$$\Ld \leq L_C$$
and
$$ | \Ld - L_C | \leq \dl.$$

We then obtain the following  result which states that the midpoint $\midd$ of the discrete set $\intd$ on $I_C$ is not too far away from the midpoint $ mid(I_C)$ of the set $I_C$. 
\begin{lemma}\label{lemma:midd}
Consider an interval $I_C \subset \setR$. Then for all discrete sets $\intd$ on $I_C$ with distance $\dl$ we have that
$$ || \midd - mid(I_C)|| \leq \dl.$$
\end{lemma}

\begin{proof}
Consider a discrete set $\intd$ on the interval $I_C$ with distance $\dl$ given by
$$I_{\dl,C} = \{y_k\}_{k=1,...,K}.$$
By definition we have that
$$  ||y_1 - (mid(I_C) - L_C)|| \leq \dl$$
and
$$  || y_K - (mid(I_C) + L_C))|| \leq \dl.$$
This implies that
\begin{eqnarray*}
 || \midd - mid(I_C)||
 &=& \left \Vert \frac{y_1 + y_K}{2} - \frac{ mid(L_C) - L_C + mid(I_C) + L_C }{2} \right \Vert \\
 &\leq& \left \Vert \frac{y_1 -  (mid(L_C) - L_C)}{2} \right \Vert
 + \left \Vert  \frac{y_K -  (mid(L_C) + L_C)}{2} \right \Vert \\
 & \leq& \dl.
\end{eqnarray*}
The result then follows.
\end{proof}

\newpage
\section{Riemann Sum Approximation}
In this appendix we review properties of the Riemann sum approximation of an integral, and apply them to the demand function $\PCd(x)$ for a discrete interval community $C$ with distance $\dl$. 

Given an integral
$$\int_a^b f(x) dx, \qquad a,b \in \setR,$$
the Riemann sum approximation of the integral is given by
$$ \sum_{n=1}^N f \Big (a + i \dl(N) \Big ) \dl(N),$$
where $N$ is an integer and
$$\dl(N) = \frac{|a-b|}N.$$

We have the following well-known result for the approximation error of the Riemann sum.
\begin{prop}\label{prop:Riemann}
If the function $f(x)$ is differentiable on the interval $[a,b]$ and there exists a constant $M$ such that
$$ |f'(x)| \leq M, \qquad x \in [a,b],$$
then we have that
$$ \lim_{N \to \infty}  \sum_{n=1}^N f\Big (a + i \dl(N) \Big ) \dl(N) = \int_a^b f(x) dx,$$
and
$$ \left | \int_a^b f(x) dx -  \sum_{n=1}^N f\big (a + i \dl(N) \big ) \dl(N) \right | \leq \frac{M|a-b|^2}{2 N} = \frac{M|a-b|}{2} \dl(N).$$
\end{prop}
This result states that the smaller $\dl(N)$, or the larger the number of points $N$, the more accurate is the Riemann sum approximation of a given interval.

We can apply this result to the  demand function $\PCd(x)$ for a discrete interval community $\dC$ with distance $\dl$ as follows. Given an interval $I_C$ and a discrete interval community $\dC= \Cd$ on $I_C$ with distance $\dl$, we have that  $ \dld \PCd(x)$ given by
$$ \dld \PCd(x)  =  \dld \sum_{y \in \intdd}  \aCd(y) p(x|y), \qquad x \in \setR,$$
can be interpreted as  a  Riemann sum approximation of the integral
$$ P_C(x) = \int_{I_C}  \aC(y) p(x|y) dy.$$
More precisely, we have the following result for the case where
$$\aCd(y) = E_p, \qquad y \in \intdd$$
and
$$\aC(y) = E_p, \qquad y \in I_C.$$
\begin{lemma}\label{lemma:Riemann_PCd}
Let $E_p$, $0< E_p \leq 1$, be a given constant. Furthermore let $f$ be the function in Assumption~\ref{ass:fg}, and let $M$ be such that
$$ |  f'(x) | < M, \qquad x \in [0,L].$$
Given an interval $I_C \subset \setR$ and a discrete community $\dC =\Cd$ on $I_C$ with distance $\dl$ we have for $ \dld \PCd(x)$ given by
$$ \dld \PCd(x)  =  \dld \sum_{y \in \intdd} E_p p(x|y), \qquad x \in \setR,$$
and $P_C(x)$, $x \in \setR$, given by
$$ P_C(x) = \int_{I_C}  E_p p(x|y) dy, \qquad x \in \setR,$$
that
$$ \Big | \dld \PCd(x) - P_C(x) \Big |  \leq  2 E_p (M L_C + 1) \dld <  2 E_p (M L_C + 1) \dl$$
where $L_C$ is the half-length of the interval $I_C$.
\end{lemma}

\begin{proof}
By definition we have for a discrete interval community $\dC = \Cd$ on $I_C$ with distance $\dl$  that
$$ \dld < \dl,$$
and in order to prove the lemma it suffices to show that
$$ \Big | \dld \PCd(x) - P_C(x) \Big |  \leq  2 E_p (M L_C + 1) \dld.$$

Recall that the set $\intdd$ is given by a set of $K$ points $\{y_1,...,y_K\}$ in the interval $I_C$ such that
$$ y_{k+1} - y_k = \dld, \qquad k=1,...,K-1.$$
We then have that
$$\dld \PCd(x)
= \dld  \sum_{y \in \intdd}  E_p p(x|y), \qquad x \in \setR $$
is the Riemann sum approximation for the integral
$$ \hat P_C(x) = \int_{y_1 - \dld}^{y_K}  E_p p(x|y) dy.$$
From Proposition~\ref{prop:Riemann}, we then have that
\begin{equation}\label{eq:PC_Riemann_error}
\Big | \dld \PCd(x) - \hat P_C(x) \Big | \leq  \frac{E_p M \Bbl 2 \Ldd + \dld \Bbr}{2} \dld
\end{equation}
where
$$\Ldd = \frac{y_K - y_1}{2}.$$
Note that by definition we have that
$$ \Ldd \geq \dld$$
and it follows that
$$ 2 \Ldd + \dld \leq 4 \Ldd.$$
Combining this result with Eq.~\eqref{eq:PC_Riemann_error}, it follows that
$$ \left | \dld \PCd(x) - \hat P_C(x) \right | \leq  2 E_p M \Ldd  \dld.$$
By definition we have that
$$ y_1 - \dld \leq mid(I_C) - L_C \leq y_1,$$
and we obtain that
$$ mid(I_C) - L_C - (y_1 - \dld) \leq \dld,$$
Similarly, we have that
$$ y_K \leq mid(I_C) + L_C  \leq y_K + \dld,$$
and we obtain that
$$ mid(I_C) + L_C - y_k \leq \dld.$$
It then follows that
$$ | \hat P_C(x) - P_C(x)| = E_p \left | \int_{y_1 - \dld}^{mid(I_C) - L_C}  p(x|y) dy \right | + E_p \left | \int_{y_K}^{mid(I_C)+L_C}  p(x|y) dy \right |,$$
and
$$ | \hat P_C(x) - P_C(x)| \leq 2 E_p\dl.$$
Using the above results, we obtain that
\begin{eqnarray*}
| \dld \PCd(x) - P_C(x) |
&\leq&  \left | \dld \PCd(x) - \hat P_C(x) \right | +  \left | \hat P_C(x) - P_C(x) \right | \\
&\leq& 2 E_p M \Ldd \dld + 2 E_p \dld \\
&\leq& 2 E_p M L_C \dld + 2 E_p \dld \\
&=&  2 E_p (M L_C + 1) \dld.
\end{eqnarray*}
The result of the lemmas then follows.
\end{proof}

\section{Properties of $\PCd(x)$}\label{app:PCd}
In this appendix we derive properties of the demand function $\PCd(x)$ of a discrete interval community $C$ with distance $\dl$, and prove Proposition~\ref{prop:PCd}.

More precisely, given a discrete  interval community community $\dC=\Cd$ with distance $\dl$,
we characterize the demand function $\PCd(x)$ given by 
$$\PCd(x) =  \sum_{y \in \intdd}  \aCd(y) p(x|y), \qquad x \in \setR.$$
Recall that for a given discrete set  $\intd$,
$$\intd = \{y_1,...,y_K\},$$
on $I_C$ with distance $\dl$, we have that
$$ \Ld = \frac{y_K - y_1}{2}$$
and 
$$ \midd = \frac{y_1 + y_K}{2}.$$
Finally, recall that the metric space $\setR$ that we use in our analysis is given by an interval  $[-L, L) \in R$, $L > 0$, with the torus metric.

In the following we will focus on the case where agents allocate all their consumption rate to the community $C$, i.e. we have that
$$\aCd(y) = E_p > 0, \qquad y \in \intdd.$$

We then have the following two results.

\begin{lemma}
Consider  a constant  $E_p$, $0 < E_p \leq 1$, and an interval
$$I_C =   [ mid(I_C) - L_C, mid(I_C) + L_C)  \subset \setR.$$
Then for every $\Dl_P >0$ there exists a $\dlO >0$ such that for all discrete interval communities $\dC =\Cd$ on $I_C$ with distance $\dl < \dlO$ we have that
$$ \Big | P_{c}(x) - \dld \PCd(x) \Big | < \Dl_P, \qquad x \in \setR,$$
where
$$\PCd(x) =  \sum_{y \in \intdd}  E_pp(x|y)$$
and
$$ P_C(x) = \int_{I_C} E_p p(x|y) dy.$$
\end{lemma}

\begin{proof}
By Lemma~\ref{lemma:Riemann_PCd} we have that
$$ | \dld \PCd(x) - P_C(x) |  <  2 E_p (M L_C + 1) \dl$$
where $M$ is such that
$$ \left | f'(x) \right | < M, \qquad x \in [0,L].$$

The result of the lemma then follows by setting
$$ \dlO < \frac{\Dl_P}{ 2 E_p (M L_C + 1) }.$$
\end{proof}

\begin{lemma}\label{lemma:PCd}
Consider a constant $E_p$, $0 < E_p \leq 1$, and an interval
$$I_C  =   [ mid(I_C) - L_C, mid(I_C) + L_C) \subset \setR.$$
Then for every $\Dl_P$, $0 < \Dl_P < L_C$, there exists a $\dlO > 0$  such that for all discrete interval communities  $\dC =\Cd$ on $I_C$ with distance $\dl < \dlO$ the following is true.
If we have that
$$ \aCd(y) = E_p, \qquad y \in \intdd,$$
and 
$$L_C < \frac{L}{2},$$
then the demand function
$$\PCd(x)  =  \sum_{y \in \intdd} \aCd(y) p(x|y) =  E_p \sum_{y \in \intdd}  p(x|y), \qquad x \in \setR,$$
has the properties that
\begin{enumerate}
\item[(a)] $\PCd(x)$ is twice differentiable on $\setR$.
\item[(b)] $\PCd(x)$ is symmetric with respect to $\middd$. 
\item[(c)] $\PCd(x)$ is strictly increasing on the interval $ [ mid(I_C) - L, mid(I_C) - \Dl_P]$, and strictly decreasing on the interval $[mid(I_C) + \Dl_P, mid(I_C) + L)$.
\item[(d)] $\PCd(x)$ is strictly concave in $x$ on the interval $[mid(I_C) - L_C,mid(I_C)+L_C]$.
\end{enumerate}
\end{lemma}

\begin{proof}
We first prove part (a) and (b) of the lemma.

The result that $\PCd(x)$ is twice differentiable on $\setR$ follows directly from Assumption~\ref{ass:fg} which implies that the function $p(x|y)$ is twice differentiable in $x$ on $\setR$.

Under the assumption that
$$ \aCd(y) = E_p >0, \qquad y \in \intd,$$
the result that  $\PCd(x)$ is symmetric around $y_0 = \midd$ follows by the same argument given in~\cite{continuous_model_arxiv} to proof that the function
$$P_C(x) = E_p \int_{y \in I_C} p(x|y) dy$$
is symmetric around $mid(I_C)$.

We next consider  part (c) of the lemma.
If
$$ \alpha_C(y) = E_p >0, \qquad y \in I_C,$$
and 
$$L_C < \frac{L}{2},$$
then we have from the analysis in~\cite{continuous_model_arxiv} that $P_C(x)$ given by
$$P_C(x) = E_p \int_{y \in I_C} p(x|y) dy,$$
is strictly increasing on the interval $[ mid(I_C) - L, mid(I_C))$, and strictly decreasing on the interval $(mid(I_C), mid(I_C) + L)$. From the analysis in~\cite{continuous_model_arxiv}, we also have that the derivative of $P_C(x)$ is continuous on $\setR$. It then follows that there exists a constant $B_1 > 0$ such that
$$ \frac{d}{dx} P_C(x) > B_1, \qquad x \in [mid(I_C) - L, mid(I_C) - \Dl_P]$$
and 
$$ \frac{d}{dx} P_C(x) < - B_1, \qquad x \in [mid(I_C) + \Dl_P, mid(I_C) + L).$$
Furthermore, by Assumption~\ref{ass:fg}, there exists constants $M_1$ and $M_2$ such that
$$\left | \frac{d}{dx} p(x|y) \right | < M_1, \qquad x \in \setR, y \in I_C,$$
and
$$\left | \frac{d^2}{d^2x} p(x|y) \right | < M_2, \qquad x \in \setR, y \in I_C.$$
By the same argument as given in the proof for Lemma~\ref{lemma:Riemann_PCd}, it then follows that 
$$ \left | \frac{d}{dx} P_C(x) -  \dl \frac{d}{dx} \PCd(x)\right | < 2 E_p (M_2L_C+M_1) \dl.$$

Let
$$ \dl^{(1)}_0 = \frac{B_1}{2} \cdot \frac{1}{2 E_p (M_2 L_C+M_1)}.$$
It then follows that for $\dl$ such that
$$ 0 < \dl <  \dl^{(1)}_0,$$
we have that
$$\dl  \frac{d}{dx} \PCd(x) > \frac{B_1}{2} > 0, \qquad x \in [mid(I_C) - L, mid(I_C) - \Dl_P]$$
and 
$$ \dl \frac{d}{dx} \PCd(x) < - \frac{B_1}{2} < 0, \qquad x \in [mid(I_C) + \Dl_P, mid(I_C) + L).$$
This result implies that for every discrete interval community  $\dC=\Cd$ on $I_C$ with distance $\dl < \dl^{(1)}_0$ we have that $P_C(x)$ is strictly increasing on the interval $[mid(I_C) - L, mid(I_C) - \Dl_P]$, and strictly decreasing on the interval $[mid(I_C) + \Dl_P, mid(I_C) + L)$.

It remains to prove part (d) of the lemma. From the analysis in~\cite{continuous_model_arxiv} we have that the function $P_C(x)$ is strictly concave on $[mid(I_C) - L_C,mid(I_C)+L_C]$, and hence there exists a constant $B_2 > 0$ such that
$$\frac{d^2}{dx^2} P_C < - B_2, \qquad x \in [mid(I_C) - L_C,mid(I_C)+L_C].$$

By Assumption~\ref{ass:fg}, there exists a constant $M_3$ such that
$$\left | \frac{d^3}{d^3x} p(x|y) \right | < M_3, \qquad x,y \in I_C,$$
and by the same argument as given in the proof for Lemma~\ref{lemma:Riemann_PCd} it follows that
$$ \left | \frac{d^2}{dx^2} P_C(x) -  \dl \frac{d^2}{dx^2} \PCd(x) \right | < 2 E_p (M_3 L_C+M_2) \dl.$$
where $M_2$ is such that
$$\left | \frac{d^2}{d^2x} p(x|y) \right | < M_2, \qquad x \in \setR, y \in I_C.$$

Let
$$ \dl^{(2)}_0 = \frac{B_2}{2} \cdot \frac{1}{2 E_p ( M_3 L_C+M_2)}.$$
It then follows that for $\dl$ such that
$$ 0 < \dl <  \dl^{(2)}_0,$$
we have that
$$\dl  \frac{d^2}{dx^2} P_{\dl,C}(x) < - \frac{B_2}{2} < 0, \qquad x \in [mid(I_C) - L_C,mid(I_C) + L_C].$$
This result implies that for every discrete interval community  $\dC=\Cd$ on $I_C$ with distance $\dl < \dl^{(2)}_0$ 
we have that $\PCd(x)$ is strictly concave in $x$ on the interval $[mid(I_C) - L_C,mid(I_C)+L_C]$.

The result of the lemma then follows by setting
$$\dlO = \min\{ \dl^{(1)}_0,  \dl^{(2)}_0 \}.$$

\end{proof}

\section{Proof of Proposition~\ref{prop:PCd}}\label{app:proof_PCd}
In this appendix we prove Proposition~\ref{prop:PCd} using the results of Appendix~\ref{app:PCd}. Recall that Proposition~\ref{prop:PCd} characterizes the content demand function $\PCd(x)$ for a discrete interval community  $\dC = \Cd$ that is part of a \erNash as given by  Proposition~\ref{prop:dnash}.

We obtain the result of  Proposition~\ref{prop:PCd} using the results of Appendix~\ref{app:PCd} as follows. Let
$$\dcommStrucNash$$
be a \erNash as given by Proposition~\ref{prop:dnash}. Then we have that for every discrete interval community $\dC = \Cd \in \sCdNash$ that
$$ \aCd(y) = E_p, \qquad y \in \intdd,$$
and
$$ L_C < \frac{L}{2},$$
where $L_C$ is the length of the interval $I_C$ such that
$$ \intdd = \Add \cap I_C.$$
As a result, the properties of a discrete interval community  $\dC = \Cd \in \sCdNash$ satisfy the assumptions in Lemma~\ref{lemma:PCd}. Proposition~\ref{prop:PCd} then follows immediately from  Lemma~\ref{lemma:PCd}.

\newpage
\section{Properties of $\xd(y)$, $\Ddy$ and $q(\xd(y)|y)$}\label{app:xd_Ddy}
In this section, we consider a discrete interval community  $\dC =\Cd$ on $I_C$ with distance $\dl$, and study the corresponding optimization problem
$$\max_{x \in \setR} q(x|y)\PCd(x), \qquad y \in I_C,$$
where
$$\PCd(x) =  \sum_{y \in \intdd}  \aCd(y) p(x|y), \qquad x \in \setR.$$
In particular, we characterize  when the above problem has an unique optimal solution $\xd(y)$, and study the function $\xd(y)$, $y \in I_C$.

Before doing this, we consider in the next subsection an interval community $C=(I_C,I_C)$ as defined in~\cite{continuous_model_arxiv}. That is, we consider a community $C = (C_d,C_s)$ where the set of content consumers and content producers is equal to the  interval $I_C \subset \setR$, i.e. we have that
$$C_d = C_s = I_C.$$
For this (continuous) interval community $C=(I_C,I_C)$, we characterize the function $\xs(y)$ given by
$$\xs(y) =  \arg \max_{x \in \setR} q(x|y) P_C(x), \qquad y \in I_C,$$
where
$$P_C(x) = E_p \int_{I_C} p(x|y) dy.$$
We then use these proprieties to characterize the function $\xd(y)$, $y \in I_C$.


\subsection{Properties of $\xs(y)$}

For a given interval community $C=(I_C,I_C)$, let the function $\xs(y)$ be given by
$$\xs(y) =  \arg \max_{x \in \setR} q(x|y) P_C(x), \qquad y \in I_C,$$
where
$$ P_C(x) = E_p \int_{I_C} p(x|y) dy, \qquad x \in \setR.$$
In the following we show that if for a given $y \in I_C$ we have that the distance
$$ || \xs(y) - x||$$
is large, then the difference between the values of $q(\xs(y)|y)P_C(\xs(y))$ and $ q(x|y)P_C(x)$ must  also be  large.
More precisely, we have the following result.
\begin{lemma}\label{lemma:qP_C}
Consider a constant $E_p$, $0 < E_p \leq 1$, and an interval community $C = (I_C,I_C)$ where $I_C  =   [ mid(I_C) - L_C, mid(I_C) + L_C) \subset \setR$ is an interval in $\setR$. If we have that
$$ \alpha_C(y) = E_p >0, \qquad y \in I_C,$$
and 
$$L_C < \frac{L}{2},$$
then the following is true.
For every $\Dl_x$, $0 < \Dl_x < \frac{L}{2}$,  there exists a $\Dl_Q > 0$ such that  if
$$ ||\xs(y)-x|| \geq \Dl_x, \qquad x \in \setR, y \in I_C,$$
then we have that
$$\big | q(\xs(y)|y)P_C(\xs(y)) - q(x|y)P_C(x) \big |  \geq  \Dl_Q,$$
where
$$\xs(y) =  \arg \max_{x \in \setR} q(x|y) P_C(x), \qquad y \in I_C,$$
and
$$ P_C(x) = E_p \int_{I_C} p(x|y) dy, \qquad x \in \setR.$$
\end{lemma}

\begin{proof}
Without loss of generality we assume in the following that the set $I_C$ is given by
$$I_C = [mid(I_C) - L_C, mid(I_C) + L_C].$$
From the analysis in~\cite{continuous_model_arxiv} we have that the function $q(x|y)P_C(x)$ is symmetric with respect to $mid(I_C)$, and it suffices to prove the result for the case where
$$ y \in [mid(I_C) - L_C, mid(I_C)].$$

By assumption we have that
$$ \alpha_C(y) = E_p >0, \qquad y \in I_C,$$
and 
$$L_C < \frac{L}{2}.$$
From the analysis in~\cite{continuous_model_arxiv} we then have that the optimization problem
$$\xs(y) =  \arg \max_{x \in \setR} q(x|y) P_C(x), \qquad   y \in [mid(I_C) - L_C, mid(I_C)],$$
has a unique optimal solution, and
$$\xs(y) \in [y,mid(I_C)] \cap supp(q(\cdot|y)), \qquad y \in [mid(I_C) - L_C, mid(I_C)].$$

For $y \in [mid(I_C) - L_C, mid(I_C)]$ and $\Dl_x$ as given in the statement of the lemma,  let $x_l(y,\Dl_x)$ and  $x_h(y,\Dl_x)$ be given by
$$x_l(y,\Dl_x) = \xs(y) - \Dl_x$$
and
$$x_h(y,\Dl_x) = \xs(y) + \Dl_x.$$
Note that under the assumption that
$$0 < \Dl_x < \frac{L}{2}$$
and
$$ L_C < \frac{L}{2},$$
we have that
$$x_l(y,\Dl_x) \in (mid(I_C)-L,\xs(y))$$
and
$$x_h(y,\Dl_x) \in (\xs(y),mid(I_C)+L).$$

From the analysis in~\cite{continuous_model_arxiv}, we have that the function $\xs(y)$ is continuous in $y$ on $ [mid(I_C) - L_C, mid(I_C)]$, and it follows that the functions  $x_l(y,\Dl_x)$ and $x_h(y,\Dl_x)$ are continuous in $y$ on  $ [mid(I_C) - L_C, mid(I_C)]$. Furthermore, the functions  $x_l(y,\Dl_x)$ and $x_h(y,\Dl_x)$ are continuous in $\Dl_x$ on $[0,L/2)$ with
$$\lim_{\Dl_x \to 0} x_l(y,\Dl_x) = \lim_{\Dl_x \to 0} x_h(y,\Dl_x) = \xs(y), \qquad y \in [mid(I_C) - L_C, mid(I_C)].$$

Using the functions  $x_l(y,\Dl_x)$ and $x_h(y,\Dl_x)$ that we defined above, let the functions $F_l(y,\Dl_x)$ and  $F_h(y,\Dl_x)$ be given by
$$F_l(y,\Dl_x) = q(\xs(y)|y)P_C(\xs(y)) -  q(x_l(y,\Dl_x)|y)P_C(x_l(y,\Dl_x))$$
and
$$F_h(y,\Dl_x) = q(\xs(y)|y)P_C(\xs(y)) -  q(x_h(y,\Dl_x)|y)P_C(x_h(y,\Dl_x)).$$
From the analysis in~\cite{continuous_model_arxiv} we have that the function
$$ q(\xs(y)|y)P_C(\xs(y)), \qquad y \in  [mid(I_C) - L_C, mid(I_C)],$$
is  continuous in $y$  on $I_C$. Combining this result with the properties of the functions  $x_l(y,\Dl_x)$ and $x_h(y,\Dl_x)$ that we obtained above, it follows that  $F_l(y,\Dl_x)$ and  $F_h(y,\Dl_x)$ are continuous in $y$ on  $[mid(I_C) - L_C, mid(I_C)]$, as well as continuous in $\Dl_x$ on $ [0,\frac{L}{2})$, with
$$ \lim_{\Dl_x \to 0} F_l(y,\Dl_x) =   \lim_{\Dl_x \to 0} F_h(y,\Dl_x) = 0, \qquad y \in I_C.$$
From the analysis in~\cite{continuous_model_arxiv} we have that there exists a unique optimal solution $\xs(y)$ for the optimization problem
$$\xs(y) =  \arg \max_{x \in \setR} q(x|y) P_C(x), \qquad y \in [mid(I_C) - L_C, mid(I_C)],$$
and it follows that for
$$ 0 <\Dl_x < \frac{L}{2}$$
we have
$$ F_l(y,\Dl_x) > 0, \qquad y \in [mid(I_C) - L_C, mid(I_C)],$$
and
$$ F_h(y,\Dl_x) > 0, \qquad y \in [mid(I_C) - L_C, mid(I_C)].$$
Let the function $F(y,\Dl_x)$ be given by 
$$F(y,\Dl_x) = min \Bsl  F_l(y,\Dl_x),F_h(y,\Dl_x) \Bsr.$$
From the above results it follows that for
$$ 0 <\Dl_x < \frac{L}{2}$$
we have
\begin{equation}\label{eq:F_pos}
F(y,\Dl_x) > 0, \qquad y \in  [mid(I_C) - L_C, mid(I_C)].
\end{equation}
Finally, we obtain from the above results that the function $F(y,\Dl_x)$ is continuous in $y$ on  $[mid(I_C) - L_C, mid(I_C)]$, as well as continuous in $\Dl_x$  on $ [0,\frac{L}{2})$, with
$$ \lim_{\Dl_x \to 0} F(y,\Dl_x) = 0, \qquad y \in I_C.$$

Having made the above definitions, we can now prove the result of the lemma. We proceed as follows. 
Given $\Dl_x$, $ 0 <\Dl_x < \frac{L}{2}$, let  the function $\Dl_Q(\Dl_x)$ be given by
$$\Dl_Q(\Dl_x) = \min_{y \in [mid(I_C) - L_C, mid(I_C)]} F(y, \Dl_x).$$
Note that the function $\Dl_Q(\Dl_x)$ is well defined as the function $F(y, \Dl_x)$ is continuous in $y$ on $I_C$.

We then have that from Eq.~\eqref{eq:F_pos} that
$$\Dl_Q(\Dl_x) > 0, \qquad  0 <\Dl_x < \frac{L}{2}.$$

For $y \in [mid(I_C) - L_C, mid(I_C)]$ and $x \in \setR$ such that 
$$ ||\xs(y)-x|| = \Dl_x, \qquad y \in  [mid(I_C) - L_C, mid(I_C)],$$
we obtain from the above results that
$$\big | q(\xs(y)|y)P_C(\xs(y)) - q(x|y)P_C(x) \big |  \geq  \Dl_Q(\Dl_x).$$
To see this, not that in this case we have by definition that
$$ F(y,\Dl_x) \geq   \min_{z \in [mid(I_C) - L_C, mid(I_C)]} F(z, \Dl_x) = \Dl_Q(\Dl_x)$$
and
\begin{eqnarray*}
\Big | q(\xs(y)|y)P_C(\xs(y)) - q(x|y)P_C(x) \Big | \geq  \min \Bsl F_l(y,\Dl_x),F_h(y,\Dl_x) \Bsr =  F(y,\Dl_x).
\end{eqnarray*}
It then follows that
$$\Big | q(\xs(y)|y)P_C(\xs(y)) - q(x|y)P_C(x) \Big | \geq \Dl_Q(\Dl_x).$$

To complete the proof, it remains to show that the function $\Dl_Q(\Dl_x)$ is non-decreasing in $\Dl_x$ on $(0, L/2)$, or equivalently that the function  $F_l(y,\Dl_x)$ and  $F_h(y,\Dl_x)$ are non-decreasing in $\Dl_x$ on $(0, L/2)$. 
This is indeed the case, as from the analysis in~\cite{continuous_model_arxiv} we have that for
$$ y \in [mid(I_C) - L_C, mid(I_C)]$$
the function $q(x|y)P_C(x)$ is non-decreasing in $x$ on $[\xs(y) - L/2, \xs(y)]$, and non-increasing in $x$ on $[\xs(y), \xs(y) + L/2]$.

This completes the proof of the lemma.
\end{proof}

We use Lemma~\ref{lemma:qP_C} in the next subsection to characterize the function $\xd(y)$. 


\subsection{Properties of $\xd(y)$}
In this subsection we characterize the function $\xd(y)$.  Our first result for the function $\xd(y)$ states that if a  discrete interval community $C = \Cd$ on a given interval $I_C$ has a small enough distance $\dl$, then the difference
$$ || \xs(y) - \xd(y) ||, \qquad y \in I_C,$$
is small. More precisely, we have the following result.
\begin{lemma}\label{lemma:xd_xs}
Consider a constant $E_p$, $0< E_p \leq 1$, and an interval
$$I_C  =   [ mid(I_C) - L_C, mid(I_C) + L_C) \subset \setR.$$ 
For every $\Dl_{x^*}$, $0 < \Dl_{x^*}$, there exists a $\dlO >0$ such that for all discrete interval communities  $\dC =\Cd$ on $I_C$ with distance $\dl < \dlO$ the following is true.
If 
$$\aCd(y) = E_p, \qquad y \in \intd$$
and
$$L_C < \frac{L}{2},$$
then we have that
$$ || \xs(y) - \xd(y) || < \Dl_{x^*}, \qquad y \in I_C,$$
where
$$ \xd(y) = \arg \max_{x \in \setR} q(x|y) \PCd(x)$$
with
$$ \PCd(x) = E_p \sum_{y \in \intdd} p(x|y),$$
and 
$$ \xs(y) = \arg \max_{x \in \setR} q(x|y) P_C(x),$$
with
$$ P_C(x) = E_p \int_{I_C} p(x|y) dy, \qquad x \in \setR.$$
\end{lemma}
Note that in the statement of the Lemma, the parameter $\dlO$ depends on the half-length $L_C$ of the interval $I_C$.

\begin{proof}
Let $\Dl_{x^*}$ be as given in the statement of the lemma. Then by Lemma~\ref{lemma:qP_C} there exists a $\Dl_Q >0$ such that if
$$ ||\xs(y)-x || \geq \Dl_{x^*}, \qquad y \in I_C,$$
then we have that
$$\Big | q(\xs(y)|y)P_C(\xs(y)) - q(x|y)P_C(x) \Big |  \geq \Dl_Q.$$
Using this $\Dl_Q$, let $\Dl_P$ be given by
$$ \Dl_P = \frac{1}{2}  \Dl_Q.$$
By Lemma~\ref{lemma:Riemann_PCd}, there exists a $\dlO>0$ such that for all discrete interval communities on $I_C$ with distance $\dl$, $0 < \dl < \dlO$, we have that
$$ \Big | P_{c}(x) - \dld P_{\dl,C}(x) \Big | < \Dl_P, \qquad x \in \setR.$$

Using these definitions of $\Dl_Q$, $\Dl_P$, and $\dlO$, we prove the result of the lemma by contradiction. That is we assume that  there exists a discrete interval communities  $C=(\intdd,\intsd)$  on $I_C$ with distance $\dl < \dlO$ and a $y \in I_C$ such that
$$ || \xs(y) - \xd(y) || \geq \Dl_{x^*},$$
and show that this leads to a contradiction. 

Suppose that there exists a discrete interval communities  $C=(\intdd,\intsd)$  on $I_C$ with distance $\dl < \dlO$ and a $y \in I_C$ such that
$$ || \xs(y) - \xd(y) || \geq \Dl_{x^*}.$$
Note that for this community  $\dC=\Cd$ and for this $y$, we have that
$$ q(\xd(y)|y)\PCd(\xd(y)) -  q(\xs(y)|y)\PCd(\xs(y)) \geq 0,$$
as by definition we have that
$$  \xd(y) = \arg \max_{x \in \setR} q(x|y) \PCd(x).$$
This implies that
\begin{eqnarray*}
&&  \dld \Bsbl q(\xd(y)|y)\PCd(\xd(y)) -  q(\xs(y)|y)\PCd(\xs(y)) \Bsbr \\
&=&   
\dld q(\xd(y)|y) \PCd(\xd(y)) -  q(\xd(y)|y)P_{C}(\xd(y)) + \cdots \\
&& +  q(\xd(y)|y)P_{C}(\xd(y)) - q(\xs(y)|y)P_{C}(\xs(y)) + \cdots \\
&& +  q(\xs(y)|y)P_{C}(\xs(y)) - \dld q(\xs(y)|y) \PCd(\xs(y))\\
&\geq& 0,
\end{eqnarray*}
or
\begin{eqnarray*}
&& q(\xd(y)|y) \Big [ \dld \PCd(\xd(y)) - P_C(\xd(y)) \Big ] 
+  q(\xs(y)|y) \Big [ P_{C}(\xs(y)) - \dld \PCd(\xs(y)) \Big ] \\
&& \geq  q(\xs(y)|y)P_{C}(\xs(y)) -  q(\xd(y)|y)P_{C}(\xd(y)).
\end{eqnarray*}
Note that
$$  q(\xs(y)|y)P_{C}(\xs(y)) -  q(\xd(y)|y)P_{C}(\xd(y)) \geq 0$$
as we have that
$$  \xs(y) = \arg \max_{x \in \setR} q(x|y) P_C(x),$$
and we obtain that
\begin{eqnarray}\nonumber
&&  q(\xd(y)|y) \Big [ \dld \PCd(\xd(y)) - P_C(\xd(y)) \Big ] 
+  q(\xs(y)|y) \Big [ P_{C}(\xs(y)) - \dld \PCd(\xs(y)) \Big ] \\ \label{eq:qP_pos}
&& \geq  q(\xs(y)|y)P_{C}(\xs(y)) -  q(\xd(y)|y)P_{C}(\xd(y)) \geq 0.
\end{eqnarray}
By construction we have for $\dl$, $0 < \dl < \dlO$, that
\begin{eqnarray}\label{eq:Dl_P} 
&&  q(\xd(y)|y) \Big ( \dld \PCd(\xd(y)) - P_C(\xd(y)) \Big ) 
+  q(\xs(y)|y) \Big ( P_{C}(\xs(y)) - \dld \PCd(\xs(y)) \Big ) \\ \nonumber 
&& \leq q(\xd(y)|y) \Big | \dld \PCd(\xd(y)) - P_C(\xd(y)) \Big | 
+  q(\xs(y)|y) \Big | P_{C}(\xs(y)) - \dld \PCd(\xs(y)) \Big |\\ \nonumber
&& < 2 \Dl_P.
\end{eqnarray}
Recall that we chose $\Dl_Q$ and $\Dl_P$ such that
$$  2 \Dl_P = \Dl_Q,$$
and Eq.~\eqref{eq:qP_pos}~and~Eq.~\eqref{eq:Dl_P} imply that for all $\dl$ such that
$$ 0 < \dl < \dlO$$
we have that
\begin{equation}\label{eq:Dl_Q1}
q(\xs(y)|y)P_{C}(\xs(y)) -  q(\xd(y)|y)P_{C}(\xd(y)) < 2 \Dl_P = \Dl_Q.
\end{equation}
However, the result of Eq.~\eqref{eq:Dl_Q1} leads to a contradiction. To see this, note that we chose $\Dl_Q$ such that for the case where
$$||\xs(y) - \xd(y)|| \geq \Dl_x^*$$
we have that
\begin{equation}\label{eq:Dl_Q2}
q(\xs(y)|y)P_{C}(\xs(y)) -  q(\xd(y)|y)P_{C}(\xd(y)) \geq \Dl_Q,
\end{equation}
where we use that fact that
$$ q(\xs(y)|y)P_{C}(\xs(y)) \geq  q(\xd(y)|y)P_{C}(\xd(y))$$
as we have that
$$  \xs(y) = \arg \max_{x \in \setR} q(x|y) P_C(x).$$
The condition in Eq.~\eqref{eq:Dl_Q2} which we obtain by construction contradicts  the result of Eq.~\eqref{eq:Dl_Q1}, which we obtained under the assumption that there exists a discrete interval communities  $C=(\intdd,\intsd)$  on $I_C$ with distance $\dl < \dlO$ and a $y \in I_C$ such that
$$ || \xs(y) - \xd(y) || \geq \Dl_{x^*}.$$

This implies that  there can not exist a discrete interval communities  $C=(\intdd,\intsd)$  on $I_C$ with distance $\dl < \dlO$ and a $y \in I_C$ such that
$$ || \xs(y) - \xd(y) | \geq \Dl_{x^*}.$$
The result of the lemma then follows.
\end{proof}

The next result provides additional properties for the  function $\xd(y)$. In particular, it provides conditions under which the optimization problem
$$ \xd(y) = \arg \max_{x \in \setR} q(x|y) \PCd(x), \qquad y \in I_C,$$
has a unique solution.

\begin{lemma}\label{lemma:xd}
Consider  a constant $E_p$, $0< E_p \leq 1$, and an interval
$$I_C  =   [ mid(I_C) - L_C, mid(I_C) + L_C) \subset \setR.$$ 
For every $\Dl_{x^*}$,
$$ 0< \Dl_{x^*} < L_C,$$
there exists a $\dlO >0$ such that for all discrete interval communities   $\dC=\Cd$ on $I_C$ with distance $\dl < \dlO$ the following is true.
If we have that 
$$ \aCd(y) = E_p, \qquad y \in \intdd,$$
and 
$$L_C < \frac{L}{2},$$
then the solution $\xd(y)$ to the optimization problem 
$$ \xd(y) = \arg \max_{x \in \setR} q(x|y) \PCd(x), \qquad y \in I_C,$$
where
$$\PCd(x) = E_p \sum_{y \in \intdd} p(x|y),$$
has the properties that
\begin{enumerate}
\item[(a)] there exists a unique optimal solution $\xd(y)$, $y \in I_C$, given by the unique solution to the equation
$$  q'(x|y) \PCd(x) +  q(x|y) \PCd'(x) = 0,  \qquad x \in I_C,$$
where $\PCd'(x)$ is the derivative of $\PCd(x)$ with respect to $x$. 
\item[(b)] for $y \in [mid(I_C) - L_C, mid(I_C) - \Dl_{x^*}]$, we have that 
$$ \xd(y)  \in (y, mid(I_C)) \cap supp(q(\cdot|y)).$$

\item[(c)] for  $y \in [mid(I_C) + \Dl_{x^*},mid(I_C) + L_C]$,  we have that 
$$ \xd(y)\in (mid(I_C),y) \cap supp(q(\cdot|y)).$$

\item[(d)] the function $\xd(y)$ is strictly increasing and differentiable on the interval
$$[ mid(I_C) - L_C, mid(I_C) - \Dl_{x^*}] \cup [mid(I_C) + \Dl_{x^*}, mid(I_C) + L_C].$$
\end{enumerate}
\end{lemma}

\begin{proof}
Let
$$y_0 = mid(I_C),$$
and let  $\Dl_{x^*}$ as given in the statement of the lemma. From the analysis in~\cite{continuous_model_arxiv} we have for $\xs(y)$, $y \in I_C$, given by
$$\xs(y) =  \arg \max_{x \in \setR} q(x|y) P_C(x), \qquad y \in I_C,$$
where
$$ P_C(x) = E_p \int_{I_C} p(x|y) dy, \qquad x \in \setR.$$
that
$$ \xs(y_0 - \Dl_{x^*}) \in (y_0 - \Dl_{x^*}, y_0),$$
and
$$ \xs(y_0 + \Dl_{x^*}) \in (y_0, y_0  + \Dl_{x^*}).$$
For $\Dl_{x^*}$ as given in the statement of the lemma,  let $\Dl_b$ be given by
$$ 2 \Dl_b =  y_0 -  \xs(y_0 - \Dl_{x^*}),$$
i.e. we have that
$$ \xs(y_0 - \Dl_{x^*}) = y_0 - 2 \Dl_b.$$
From the analysis in~\cite{continuous_model_arxiv}, we have that
$$ \xs(y_0 - \Dl_{x^*}) \in (y_0 - \Dl_{x^*}, y_0),$$
and we obtain that
$$\Dl_b > 0.$$
Consider the function $\Dy$ given by
$$\Dy = || y - \xs(y)||, \qquad y \in I_C.$$
Furthermore, from the analysis in~\cite{continuous_model_arxiv} we have that the function $\Dy$ given by
$$\Dy = || y - \xs(y)||, \qquad y \in I_C,$$
is symmetric on $I_C$ with respect to $y_0$ we have that
$$ \xs(y_0 + \Dl_{x^*}) = y_0 + 2 \Dl_b.$$

Using $\Dl_b$ as defined above, let $\dl^{(1)}_0$, $ \dl^{(1)}_0 > 0$, be such that for all discrete interval communities on $I_C$ with distance $\dl < \dl_0^{(1)}$ we have that
\begin{enumerate}
\item $\PCd(x)$ is strictly concave in $x$ on $I_C$,
\item $\PCd(x)$ is strictly increasing in  $x$ on $[mid(I_C) - L_C, mid(I_C) - \Dl_b]$, and 
\item  $\PCd(x)$ is strictly decreasing in $x$ on $[mid(I_C) + \Dl_b, mid(I_C) + L_C]$.
\end{enumerate}
Note that such a constant  $\dl^{(1)}_0$ exists by Lemma~\ref{lemma:PCd}.

Furthermore, let  $\dl^{(2)}_0$, $ \dl^{(2)}_0 > 0$, be such that  for all discrete interval communities on $I_C$ with distance $\dl < \dl^{(2)}_0$ we have that
$$ || \xs(y) - \xd(y) || < \Dl_b, \qquad y \in I_C.$$
Note that such a constant  $\dl^{(2)}_0$ exists by Lemma~\ref{lemma:xd_xs}.

Let 
$$\dlO = \min\Bsl  \dl^{(1)}_0,  \dl^{(2)}_0 \Bsr.$$

We first prove Property (a) of the lemma. To do that, note that  by construction the function $\PCd(x)$ is strictly concave on $I_C$, strictly increasing on the interval  $[mid(I_C) - L_C, mid(I_C) - \Dl_b]$, and strictly decreasing on the interval  $[mid(I_C) + \Dl_b, mid(I_C) + \Dl_b]$. To prove Property (a) of the lemma, we can then use the same argument as given in~\cite{continuous_model_arxiv} to prove the result that under the assumption as given in the statement of the lemma we have that there exists a unique solution $\xd(y)$, $y \in I_C$, to the optimization problem
$$\xs(y) =  \arg \max_{x \in \setR} q(x|y) P_C(x), \qquad y \in I_C,$$
and $\xs(y)$ is given by the unique solution to the equation
$$  q'(x|y) P_C(x) +  q(x|y) P_C'(x) = 0,  \qquad x \in I_C,$$

Next we proof Property (b) and (d) for the case where
$$ y \in [mid(I_C) - L_C, mid(I_C) - \Dl_{x^*}].$$
 Property (c) and (d) for the case where $y \in [mid(I_C) + \Dl_{x^*}, mid(I_C) + L_C]$ can be proven using the same argument. 

From the analysis in~\cite{continuous_model_arxiv} we have  that the function $\xs(y)$ is strictly increasing on $I_C$, and we obtain that
$$ \xs(y) \leq \xs(y_0 - \Dl_{x^*}) =  mid(I_C) - 2\Dl_b, \qquad y \in [mid(I_C) - L_C,  mid(I_C) - \Dl_{x^*}].$$
Furthermore, by construction we have 
$$ || \xs(y) - \xd(y) || < \Dl_b, \qquad y \in I_C,$$
and it follows that
$$ \xd(y) <  mid(I_C) - \Dl_b, \qquad y \in [mid(I_C) - L_C,  mid(I_C) - \Dl_{x^*}].$$
Using this result with the fact that by construction the function $\PCd(x)$ is strictly increasing and strictly concave in $x$ on the interval  $[mid(I_C) - L_C, mid(I_C) - \Dl_b]$, we can then use the same argument to prove Property (b) as given in~\cite{continuous_model_arxiv} to prove the result that under the assumption as given in the statement of the lemma we have that
$$ \xs(y)  \in (y, mid(I_C)) \cap supp(q(\cdot|y)).$$
Furthermore, using the same argument as given in~\cite{continuous_model_arxiv} to show that the function $\xs(y)$  is strictly increasing and differentiable on the interval $I_C$, we can show  that $\xd(y)$ is strictly increasing and differentiable on the interval $[mid(I_C) - L_C, mid(I_C) - \Dl_{x^*}]$.

The result of the lemma then follows.
\end{proof}

-

\begin{lemma}\label{lemma:xd2}
Consider  a constant $E_p$, $0< E_p \leq 1$, and an interval
$$I_C  =   [ mid(I_C) - L_C, mid(I_C) + L_C) \subset \setR.$$ 
There exists a $\dlO >0$ such that for all discrete interval communities   $\dC=\Cd$ on $I_C$ with distance $\dl < \dlO$ the following is true.
If we have that
$$ \aCd(y) = E_p, \qquad y \in \intdd,$$
and 
$$L_C < \frac{L}{2},$$
then for
$$ y \in [mid(I_C) - L, mid(I_C) - L_C)$$
we have that
$$ \xd(y) \leq \xd \Bbl mid(I_C) - L_C \Bbr,$$
and for
$$ y \in (mid(I_C) + L_C, mid(I_C) + L)$$
we have that
$$ \xd(y) \geq \xd \Bbl mid(I_C) + L_C \Bbr.$$
\end{lemma}

\begin{proof}
We first consider the case where
$$ y \in [mid(I_C) - L, mid(I_C) - L_C).$$
Let
$$y_l = mid(I_C) - L_C.$$
From the analysis in~\cite{continuous_model_arxiv} we have that
\begin{equation}\label{eq:Dl_x}
\xs(y_l) \in [y_l,mid(I_C)) = [mid(I_C) - L_C, mid(I_C)),
\end{equation}
and hence
$$ \xs(y_l) < mid(I_C).$$
Using this result, let $\Dl_x$ be given by
$$\Dl_x = \frac{mid(I_C) - \xs(y_l)}{2}.$$
Note that by Eq.\eqref{eq:Dl_x} we have that
$$ 0 < \Dl_x <\frac{L_C}2.$$

Let $\dl^{(1)}_0 > 0$ be such that  for all discrete interval communities   $C=(\intdd,\intsd)$ on $I_C$ with distance $\dl < \dl^{(1)}_0$ we have that the function $\PCd(x)$ is strictly increasing and strictly concave on $[mid(I_C) - L_C, mid(I_C) -\Dl_x]$,
and that
$$ || \xd(y) - \xs(y) || \leq \Dl_x, \qquad y \in [mid(I_C) - L, mid(I_C)].$$
Note that such a $\dlO$ exists by Lemma~\ref{lemma:PCd}~and~\ref{lemma:xd_xs}.

By construction we have that
$$\xs(y_l) \leq mid(I_C) - 2 \Dl_x.$$
Furthermore from the analysis in~\cite{continuous_model_arxiv} we have that
$$\xs(y) \leq \xs(y_l), \qquad y \in [mid(I_C) - L, mid(I_C) - L_C).$$
Combining this result we the fact that
$$|| \xd(y) - \xs(y) || \leq \Dl_x, \qquad y \in [mid(I_C) - L, mid(I_C)],$$
we obtain that
$$\xd(y_l) <  mid(I_C) - \Dl_x$$
and
$$\xd(y) <  mid(I_C) - \Dl_x, \qquad  y \in [mid(I_C) - L, mid(I_C)- L_C).$$
As by construction we have that $\PCd(x)$ is strictly increasing and strictly concave on $[mid(I_C) - L_C, mid(I_C) - \Dl_x]$, we can show that under the assumptions given in the lemma we have for all discrete interval communities   $\dC=\Cd$ on $I_C$ with distance $\dl < \dl^{(1)}_0$ that
$$ \xd(y) \leq \xd \big ( mid(I_C) - L_C \big ), \qquad  y \in [mid(I_C) - L, mid(I_C)),$$
using the same argument as given in the analysis in~\cite{continuous_model_arxiv} to prove that interval communities   $C=(I_C,I_C)$ we have that
$$ \xs(y) \leq \xs \big ( mid(I_C) - L_C \big ), \qquad  y \in [mid(I_C) - L, mid(I_C)).$$

Using the same approach, we can show that there exists  $\dl^{(2)}_0$, $ \dl^{(2)}_0 > 0$, such that for all discrete interval communities on $I_C$ with distance $\dl < \dl^{(2)}_0$  we have for
$$ y \in (mid(I_C) + L_C, mid(I_C) + L)$$
that
$$ \xd(y) \geq \xd \Bbl mid(I_C) + L_C \Bbr.$$

The result of the lemma then follows by setting
$$\dlO = \min\{ \dl^{(1)}_0,  \dl^{(2)}_0 \}.$$
\end{proof}


\subsection{Properties of  $\Ddy$}
Having analyzed  the function $\xd(y)$ given by
$$\xd(y)= \arg\max_{x \in \setR} q(x|y) \PCd(x), \qquad y \in I_C.$$
we study next the properties of the function $\Ddy$, $y \in I_C$, given by
$$ \Ddy = || y - \xd(y)||.$$

We have the following result.

\begin{lemma}\label{lemma:Ddy}
Consider  a constant $E_p$, $0< E_p \leq 1$, and an interval
$$I_C  =   [ mid(I_C) - L_C, mid(I_C) + L_C) \subset \setR.$$ 
For every $\Dl_{x^*}$, $0 < \Dl_{x^*} < L_C$, there exists a $\dlO >0$ such that for all discrete interval communities $\dC=\Cd$ on $I_C$ with distance $\dl < \dlO$ the following is true.
If 
$$ \aCd(y) = E_p, \qquad y \in \intdd,$$
and 
$$L_C < \frac{L}{2},$$
then we have that  the function $\Ddy$ given by
$$ \Ddy = || y - \xd(y)||, \qquad y \in I_C,$$
where
$$\xd(y)= \arg\max_{x \in \setR} q(x|y) \PCd(x),$$
is  strictly decreasing and differentiable on  $[mid(I_C) - L_C, mid(I_C) -  \Dl_{x^*}]$, and strictly increasing and differentiable on  $[mid(I_C)  + \Dl_{x^*}, mid(I_C)]+L_C]$.
\end{lemma}

\begin{proof}
Let
$$y_0 = mid(I_C).$$
We first show that there exists  $\dl^{(1)}$, $ \dl^{(1)} > 0$, such that for all discrete interval communities on $I_C$ with distance $\dl < \dl_0^{(1)}$ the function $\Ddy$ is  strictly decreasing and differentiable on  $[y_0 - L_C, y_0 -  \Dl_{x^*}]$. To do this, let $\Dl_{x^*}$ be as in the statement of the lemma, and let then $\Dl_b$ be given by
$$ 2 \Dl_b =  | y_0 -  \xs(y_0 - \Dl_{x^*}) |.$$
From the analysis in~\cite{continuous_model_arxiv} we have that 
$$ \xs(y_0 - \Dl_{x^*}) \in ( y_0 - \Dl_{x^*}, y_0),$$
and it follows that
$$\Dl_b > 0.$$
We then choose $\dl_0^{(1)}$ as follows.
Let $\dl^{(1)}_0$, $ \dl^{(1)}_0 > 0$, be such that for all discrete interval communities on $I_C$ with distance $\dl < \dl_0^{(1)}$ we have that $\PCd(x)$ is strictly increasing and strictly concave in $x$ on $[y_0 - L_C, y_0 -  \Dl_{x^*}]$, and
$$ || \xs(y) - \xd(y) || < \Dl_b,  \qquad y \in [y_0 - L_C, y_0 -  \Dl_{x^*}].$$
Note that such a constant  $\dl^{(1)}_0$ exists by Lemma~\ref{lemma:PCd}~and~\ref{lemma:xd_xs}.

Using this construction, we have  for all discrete interval communities on $I_C$ with distance $\dl < \dl_0^{(1)}$ that 
$$ \xd(y) \in [y,mid(I_C) - \Dl_b], \qquad y \in [mid(I_C) - L_C,  mid(I_C) - \Dl_{x^*}].$$
To see this note the following. Recall that  $\Dl_b$ is given by
$$ 2 \Dl_b =  || y_0 -  \xs(y_0 - \Dl_{x^*}) ||.$$
From the analysis in~\cite{continuous_model_arxiv} we have that the function $\xs(y)$ is strictly increasing on $I_C$, and we have that
$$ \xs(y) \in [y,mid(I_C) - 2 \Dl_b] \cap supp(q(\cdot|y)), \qquad y \in [mid(I_C) - L_C,  mid(I_C) - \Dl_{x^*}].$$
Furthermore, by construction we have that
$$ || \xs(y) - \xd(y) || < \Dl_b, \qquad y \in I_C.$$
Combining these results, it then follows that
$$ \xd(y) \in [y,mid(I_C) - \Dl_b], \qquad y \in [mid(I_C) - L_C,  mid(I_C) - \Dl_{x^*}].$$
Combining this result with the fact that by construction the function $\PCd(x)$ is strictly increasing and strictly concave in $x$ on the interval  $[mid(I_C) - L_C, mid(I_C) - \Dl_b]$, and in order to show that $\Ddy$ is  strictly decreasing and differentiable on  $[mid(I_C) - L_C, mid(I_C) -  \Dl_{x^*}]$ we can use the same argument as given in~\cite{continuous_model_arxiv} to show the function $\Dy$ is  strictly decreasing and differentiable on  $[mid(I_C) - L_C, mid(I_C)]$ where
$$\Dy = || y - \xs(y)||, \qquad y \in I_C.$$

Using the same approach, we can show that there exists  $\dl^{(2)}_0$, $ \dl^{(2)}_0 > 0$, such that for all discrete interval communities on $I_C$ with distance $\dl < \dl^{(2)}_0$ the function $\Ddy$ strictly increasing and differentiable on  $[ mid(I_C) +  \Dl_{x^*}, mid(I_C) + L_C]$.

The result of the lemma then follows by setting
$$\dlO = \min\{ \dl^{(1)}_0,  \dl^{(2)}_0 \}.$$
\end{proof}


\subsection{Properties of $q(\xd(y)|y)$}
Finally, we characterize the function $q(\xd(y)|y)$
where
$$ \xd(y) = \arg \max_{x \in \setR} q(x|y) \PCd(x).$$
The first results states that if $\dl$ is small enough, then we have for a discrete interval communities $C=(\intdd,\intsd)$ on $I_C$ with distance $\dl$ that that the distance
$$ \Big |q(\xd(y)|y) - q(\xs(y)|y) \Big |$$
is small, where
$$ \xs(y) = \arg \max_{x \in \setR} q(x|y) P_C(x).$$

\begin{lemma}\label{lemma:qPCd_qPC}
Consider  a constant $E_p$, $0< E_p \leq 1$, and an interval
$$I_C  =   [ mid(I_C) - L_C, mid(I_C) + L_C) \subset \setR.$$. 
For every $\Dl_q$, $\Dl_q>0$, there exists a $\dlO >0$ such that for all discrete interval communities $\dC=\Cd$ on $I_C$ with distance $\dl < \dlO$ the following is true.
If
$$\aCd(y) = E_p, \qquad y \in \intdd,$$
and
$$L_C < \frac{L}{2},$$
then we have that
$$ \Big |q(\xd(y)|y) - q(\xs(y)|y) \Big | < \Dl_q, \qquad y \in I_C,$$
where
$$ \xd(y) = \arg \max_{x \in \setR} q(x|y) \PCd(x)$$
and
$$ \xs(y) = \arg \max_{x \in \setR} q(x|y) P_C(x).$$
\end{lemma}

\begin{proof}
Let $ \Dl_{x}$, $ \Dl_{x} >0$, be such that
$$ \Dl_{x} < \frac {\Dl_q}{ \max_{x \in supp(g)} |g'(x)|},$$
where $g$ is the function of Assumption~\ref{ass:fg}. 
Note that such a $\Dl_x$ exists by   Assumption~\ref{ass:fg}.

Let $\dlO$ be such that for all discrete interval communities on $I_C$ with distance $\dl < \dlO$  we have that
$$ ||\xd(y) - \xs(y) || < \Dl_x, \qquad y \in I_C.$$
Note that such a $\dlO$ exists by Lemma~\ref{lemma:xd_xs}.

We then have that
$$ \Big |q(\xd(y)|y) - q(\xs(y)|y) \Big | <  \max_{x \in supp(g)} |g'(x)|\cdot ||\xd(y) - \xs(y)||.$$
As by construction we have that
$$ \Dl_{x} < \frac {\Dl_q}{ \max_{x \in supp(g)} |g'(x)|}$$
and
$$ ||\xd(y) - \xs(y) || < \Dl_x, \qquad y \in I_C,$$
we obtain that 
$$  \Big |q(\xd(y)|y) - q(\xs(y)|y) \Big | < \Dl_q, \qquad y \in I_C.$$
The result of the lemma then follows.
\end{proof}

The next results provides properties of the function $q(\xd(y)|y)$ that we will use in our analysis.
\begin{lemma}
Consider  a constant  $E_p$, $0< E_p \leq 1$, an interval
$$I_C  =   [ mid(I_C) - L_C, mid(I_C) + L_C) \subset \setR.$$ 
For every $\Dl_q$, $0<\Dl_q>L_C$, there exists a $\dlO >0$ such that for all discrete interval communities $\dC=\Cd$ on $I_C$ with distance $\dl < \dlO$ the following is true.
If
$$\aCd(y) = E_p, \qquad y \in \intdd,$$
and
$$L_C < \frac{L}{2},$$
then we have for
$$ y \in [mid(I_C) - L_C, mid(I_C) - \Dl_q] \cup [mid(I_C) - L_C, mid(I_C) - \Dl_q]$$
that
$$ q(\xd(y)|y) > 0$$
and that $q(\xd(y)|y)$ is differentiable with with respect to $y$. 
\end{lemma}

\begin{proof}
Let $\dlO>0$ be such that for  all discrete interval communities $C=(\intdd,\intsd)$ on $I_C$ with distance $\dl < \dlO$ we have that the function
$$\Ddy = || y - \xd(y) ||$$
is differentiable in $y$ on
$$[mid(I_C) - L_C, mid(I_C) - \Dl_q] \cup [mid(I_C) - L_C, mid(I_C) - \Dl_q].$$
Note that such a $\dlO$ exists by Lemma~\ref{lemma:Ddy}.

Recall that
$$ q(\xd(y)|y) = g(\Ddy),$$
where $g$ is the function in Assumption~\ref{ass:fg}.

As we have that
$$\xd(y)  = \arg \max_{x \in \setR} q(x|y) \PCd(x)$$
it follows that
$$  q(\xd(y)|y) \PCd(\xd(y)) \geq q(y|y) \PCd(y) = g(0) \PCd(y)> 0,$$
as by Assumption~\ref{ass:fg} we have that
$$g(0) > 0$$
and by Lemma~\ref{lemma:PCd} we have that
$$  \PCd(y) > 0, \qquad y \in \setR.$$
This implies that
$$ q(\xd(y)|y) \PCd(\xd(y)) > 0, \qquad y \in \setR,$$
and hence
$$ q(\xd(y)|y)  > 0, \qquad y \in \setR.$$
From this result we obtain that
$$ \xd(y) \in supp(q(\cdot|y), \qquad y \in \setR.$$

By Assumption~\ref{ass:fg} we have 
that the function $g$ is differentiable on $supp(g)$, and by construction we have that  the function $\Ddy$ is differentiable in $y$ on
$$[mid(I_C) - L_C, mid(I_C) - \Dl_q] \cup [mid(I_C) - L_C, mid(I_C) - \Dl_q].$$
Combining these results, we obtain that the function
$$ q(\xd(y)|y) = g(\Ddy)$$
is differentiable in $y$ on
$$[mid(I_C) - L_C, mid(I_C) - \Dl_q] \cup [mid(I_C) - L_C, mid(I_C) - \Dl_q].$$

The result of the lemma then follows. 
\end{proof}


\section{Proof of Proposition~\ref{prop:xd}~and~\ref{prop:Ddy}}\label{app:proof_xd_Ddy}
In this appendix we prove Proposition~\ref{prop:xd}~and~\ref{prop:Ddy} using the analysis of Appendix~\ref{app:xd_Ddy}. Recall that Proposition~\ref{prop:xd} characterizes the optimal content $\xd(y)$ that a agent $y \in \intsd$ produces in a discrete interval community  $\dC = \Cd$ that is part of a \erNash as given by  Proposition~\ref{prop:dnash}, and  Proposition~\ref{prop:Ddy} characterizes the function
$$ \Ddy = || \xd(y) - y||,$$
i.e. by how much the agent adapts its optimal content towards the main interest of the community $\dC = \Cd$.

We obtain the results of  Proposition~\ref{prop:xd}~and~\ref{prop:Ddy} using the analysis of Appendix~\ref{app:xd_Ddy} as follows.
Let $\dcommStrucNash$ be a \erNash as given by Proposition~\ref{prop:dnash}. Then we have that for every discrete interval community $\dC = \Cd \in \sCdNash$ that
$$ \aCd(y) = E_p, \qquad y \in \intdd,$$
and
$$ L_C < \frac{L}{2},$$
where $L_C$ is the length of the interval $I_C$ such that
$$ \intdd = \Add \cap I_C.$$
As a result, the properties of a discrete interval community  $\dC = \Cd \in \sCdNash$ satisfy the assumptions in Lemma~\ref{lemma:xd}~and~\ref{lemma:Ddy}. Proposition  then follows immediately from Lemma~\ref{lemma:xd}~and~\ref{lemma:Ddy}.

\newpage
\section{Properties of $\bCsd(\cdot,y)$}\label{app:BCsd}
In this appendix, we characterize how agents in a discrete interval community optimally allocated their content production rate. More precisely, given an interval $I_C \subset \setR$ and a discrete interval community $\dC=\Cd$ on $I_C$ with distance $\dl$, we characterize how an agents $y \in \intsd$ allocate their content production rate $\bCd(\cdot|y)$ in community $C$. For this, we assume that there is a constraint on the maximal production rate, i.e. we have that
$$ ||\bCd(\cdot|y) || = \int_{\setR} \beta_C(x|y) dx  \leq \beta_C(y), \qquad y \in \intsd,$$
where $ \bCd(y)>0$ denotes the maximal production rate that agent $y \in \intdd$ can allocate to community $\dC$. In this case the optimal content production rate $\bCsd(\cdot|y)$ of agent $y \in \intd$ is given by
$$\bCsd(\cdot|y) = \underset{\bCd(\cdot|y): || \bCd(\cdot|y)|| \leq \bCd(y)}{\arg\max} \int_{\setR} \bCd(x|y) [ q(x|y) \PCd(x) - \aCd c] dx,$$
where
$$\aCd = \sum_{y \in \intdd} \aCd(y),$$
and
$$ \PCd(x) = \sum_{y \in \intdd} \aCd(y) p(x|y).$$

Our first result provides a characterization of the optimal content production rate allocation for discrete interval communities.
\begin{lemma}\label{lemma:optimal_production_rate}
Consider an interval $I_C  =   [ mid(I_C) - L_C, mid(I_C) + L_C) \subset \setR$, and a discrete interval community  $\dC =\Cd$ on $I_C$ with distance $\dl$.
For $y \in \intsd$ and $\beta_C(y) > 0$, let the function $\bCsd(\cdot|y)$  be given by
$$\bCsd(\cdot|y) = \underset{\bCd(\cdot|y): || \bCd(\cdot|y)|| \leq \bCd(y)}{\arg\max} \int_{x \in \setR} \bCd(x|y) \Big [ q(x|y) \PCd(x) - \aCd c \Big] dx$$
where
$$\aCd = \sum_{y \in \intdd} \aC(y),$$
and
$$ \PCd(x) = \sum_{y \in \intdd} p(x|y).$$
Furthermore,  let
$$ b^* = \max_{x \in \setR} q(x|y) \PCd(x),$$
and let the set $A \subseteq \setR$ be given by
$$A = \Bsl x \in \setR | q(x|y) \PCd(x) = b^*\Bsr.$$
If  we have that
$$\max_{x \in \setR} \Bbl q(x|y) \PCd(x) - \aCd c \Bbr > 0, \qquad y \in I_C,$$
then the optimal production rate allocation has the property that
$$ \int_{\in A} \bCsd(x|y) dx =  \int_{\setR} \bCsd(x|y) dx.$$
\end{lemma}

\begin{proof}
This lemma can be proved using the same argument as given in~\cite{continuous_model_arxiv} where the same result is proven for interval communities $C=(I_C,I_C)$ under the continuous agent model.
\end{proof}

The next result follows directly from Lemma~\ref{lemma:optimal_production_rate}. 
\begin{cor}\label{cor:optimal_production_rate}
Consider constants $E_p$, $0< E_p\leq 1$, and $E_q$, $0 < E_q$, and an interval $I_C  =   [ mid(I_C) - L_C, mid(I_C) + L_C) \subset \setR$.
For every $\Dl_y$, $0 < \Dl_y$, there exists a $\dlO >0$ such that for all discrete interval communities $\dC =\Cd$ on $I_C$ with distance $\dl < \dlO$ the following is true.
Suppose that
$$\aCd(y) = E_p, \qquad  y \in \intdd,$$
and
$$L_C < \frac{L}{2}.$$
If for agent $y \in \intdd$ we have that
$$\max_{x \in \setR} \Bbl q(x|y) \PCd(x) - \aCd c \Bbr > 0,$$
where
$$\aCd = \sum_{y \in \intdd} E_p,$$
and
$$ \PCd(x) = \sum_{y \in \intdd} E_p p(x|y),$$
then
$$\bCsd(x|y) = E_q \dl\big (x-\xd(y) \big ), \qquad x \in \setR,$$
where 
$$\xd(y) = \arg \max_{x \in \setR}  q(x|y) \PCd(x),$$
is an optimal solution for  
$$\bCsd(\cdot|y) = \underset{\bCd(\cdot|y): || \bCd(\cdot|y)|| \leq E_q}{\arg\max} \int_{x \in \setR} \bCd(x|y) \Big [ q(x|y) \PCd(x) - \aCd c \Big ] dx.$$
\end{cor}

\section{Properties of  $\QCsd(x)$}\label{app:QCd}
In the previous appendix we studied the properties of  the optimal content production rate allocation $\bCsd(\cdot|y)$ of an agent $y \in \intdd$ in a discrete interval community  $\dC =\Cd$. In this appendix, we characterize the resulting optimal content supply function $\QCsd(x)$, $x \in \setR$, given by
\begin{equation}\label{eq:QCsd}
\QCsd(x) = \sum_{y \in \intsd} E_q \dl\big (x-\xd(y) \big ) q(x|y)dx, \qquad x \in \setR,
\end{equation}
where
$$\xd(y) = \arg \max_{x \in \setR}  q(x|y) \PCd(x).$$

\begin{lemma}\label{lemma:QCd1}
Consider constants  $E_p$, $0< E_p \leq 1$, and $E_q$, $0<E_q$, and an interval $I_C  =   [ mid(I_C) - L_C, mid(I_C) + L_C) \subset \setR$. 
For every $\Dl_{Q}>0$ and $\Dl_x>0$, such that
$$\Dl_Q + \Dl_x < L_C,$$
there exists a $\dlO >0$ such that for all discrete interval communities $\dC =\Cd$ on $I_C$ with distance $\dl < \dlO$ the following is true.
Suppose that
$$\aCd(y) = E_p, \qquad  y \in \intdd,$$
and
$$\max_{x \in \setR} \Bbl q(x|y) \PCd(x) - \aCd c \Bbr > 0, \qquad y \in \intsd,$$
as well as
$$L_C < \frac{L}{2}.$$
Furthermore, let
$$\Id^{(1)} = [mid(I_C) - L_C, mid(I_C) - \Dl_Q] \cap \intsd,$$
$$\Id^{(2)} = (mid(I_C) - \Dl_Q, mid(I_C) + \Dl_Q) \cap \intsd,$$
and
$$\Id^{(3)} = [mid(I_C) +\Dl_Q, mid(I_C) + L_C] \cap \intsd.$$
Then the optimal content supply function $Q^*_C(x)$, $x \in \setR$, as given by Eq.~\eqref{eq:QCsd} is such that
\begin{eqnarray*}
\QCsd(x) &=& E_q \sum_{y \in \Id^{(1)}} \dl\big (x-\xd(y)) q(\xd(y)|y \big ) + \cdots\\
 && +  E_q \sum_{y \in \Id^{(2)}} \dl\big (x-\xd(y)) q(\xd(y)|y \big ) + \cdots\\
 && + E_q \sum_{y \in \Id^{(3)}}  \dl \big (x-\xd(y) \big ) q(\xd(y)|y),
\end{eqnarray*}
where $\dl(\cdot)$ is the Dirac delta function, and
\begin{enumerate}
\item[a)]  $\xd(y)$ is the unique solution to the optimization problem
$$\xd(y) = \arg \max_{x \in \setR}  q(x|y) \PCd(x).$$
\item[b)]  for $y \in \Id^{(1)}$ we have that
$$\xd(y) \in (y,mid(I_C)),$$
\item[c)]  for $y \in \Id^{(3)}$ we have that
$$\xd(y) \in (mid(I_C),y),$$
\item[d)]
and for $y \in \Id^{(2)}$ we have that
$$\xd(y) \in  (mid(I_C) - \Dl_Q - \Dl_x, mid(I_C) + \Dl_Q + \Dl_x).$$
\end{enumerate}
\end{lemma}

\begin{proof}
Let $\dlO$, $ \dlO > 0$, be such that for all discrete interval communities $\dC=\Cd$ on $I_C$ with distance $\dl < \dlO$ the following is true,
\begin{enumerate}
\item[1)] the optimization problem
$$\xd(y) = \arg \max_{x \in \setR}  q(x|y) \PCd(x)$$
has a unique solution for $y \in I_C$.
\item[2)] for $y \in \Id^{(1)}$ we have that
$$\xd(y) \in (y,mid(I_C)).$$
\item[3)] for $y \in \Id^{(3)}$ we have that
$$\xd(y) \in (mid(I_C),y).$$
\item[4)] we have that
$$ || \xs(y) - \xd(y) || < \Dl_x, \qquad y \in I_C,$$
where
$$\xs(y) = \arg \max_{x \in \setR} q(x|y)P_C(x)$$
and
$$P_C(x) = E_p \int_{I_C} p(x|y) dy.$$
\end{enumerate}
Note that such a $\dlO$ exists by Lemma~\ref{lemma:xd}~and~\ref{lemma:xd_xs}.

Using Property 1)-3) of the above construction, we obtain Property a)-c) of the lemma.

Furthermore, combining Property 4  of the above construction with the analysis in~\cite{continuous_model_arxiv} which shows that
$$ \xs(y) \in (y,mid(I_C)), \qquad y \in [mid(I_C) - L_C, mid(I_C)]$$
and
$$ \xs(y) \in (mid(I_C),y), \qquad y \in [mid(I_C), mid(I_C)+L_C],$$
it follows that
$$ \xd(y) \in  (mid(I_C) - \Dl_Q - \Dl_x, mid(I_C) + \Dl_Q + \Dl_x), \qquad y \in  \Id^{(2)}.$$
The result of the lemma then follows. 
\end{proof}

The next result states that the support of the optimal content supply function $\QCsd(x)$, $x \in \setR$, is contained in $I_C$ for all discrete interval communities on $I_C$ with a small enough distance $\dl$.
More precisely, we have the following result.
\begin{lemma}\label{lemma:QCd2}
Consider constants  $E_p$, $0 < E_p \leq 1$,  and $E_q$, $0< E_q$, and an interval $I_C  =   [ mid(I_C) - L_C, mid(I_C) + L_C) \subset \setR$. 
Then there exists a $\dlO >0$ such that for all discrete interval communities $\dC =\Cd$ on $I_C$ with distance $\dl < \dlO$ the following is true.
Suppose that
$$\aCd(y) = E_p, \qquad  y \in \intdd,$$
and
$$\int_{\setR} \bCsd(x|y) dx = E_q,   \qquad y \in \intsd,$$
as well as
$$L_C < \frac{L}{2}.$$
Then we have for the content supply function $\QCsd(x)$, $x \in \setR$, as given by Equation~\eqref{eq:QCsd} that
$$supp(\QCsd(\cdot)) \subseteq [mid(I_C) - \LSd, mid(I_C) + \LSd] \subset I_C,$$
where
$$0< \LSd < L_C.$$
\end{lemma}

\begin{proof}
From the analysis in~\cite{continuous_model_arxiv}  we have for the function $\xs(y)$, $y \in I_C$, given by
$$\xs(y) = \arg \max_{x \in \setR} q(x|y)P_C(x)$$
where
$$P_C(x) = E_p \int_{I_C} p(x|y) dy,$$
that the image $I^*_C$ of $\xs(y)$, $y \in I_C$, is such that
$$I^*_C \subseteq [mid(I_C) - L^*_C,mid(I_C) + L^*_C]$$
where
$$ 0 < L^*_C < L_C.$$
Using this result, let
$$\Dl_x = \frac{L_C - L^*_C}{2} > 0$$
and
$$\Dl_Q = \frac{L^*_C}{2} > 0.$$
Using this definition of $\Dl_x$, let
$$\LSd = L^*_C + \Dl_x.$$
Note that
$$ 0 < \LSd < L_C.$$

Using the above definition of $\Dl_x$, let $\dl_0^{(1)} > 0$ be such that for all discrete interval communities $C = \Cd$  on $I_C$ with distance $\dl < \dl_0^{(1)}$ we have that
$$ || \xd(y) - \xs(y) || < \Dl_x.$$
Note that such a $\dl_0^{(1)}$ exists by Lemma~\ref{lemma:xd_xs}. 

Using the above definitions of $\Dl_x$ and $\Dl_Q$, let for a given discrete interval community $\dC=\Cd$ on the interval $I_C$ with distance $\dl$ the sets $\Id^{(1)}$, $\Id^{(2)}$, and  $\Id^{(3)}$ be given by
$$\Id^{(1)} = [mid(I_C) - L_C, mid(I_C) - \Dl_Q] \cap \intsd,$$
$$\Id^{(2)} = (mid(I_C) - \Dl_Q, mid(I_C) + \Dl_Q) \cap \intsd,$$
and
$$\Id^{(3)} = [mid(I_C) +\Dl_Q, mid(I_C) + L_C] \cap \intsd.$$
Furthermore, using the definitions of $\Id^{(1)}$, $\Id^{(2)}$, and  $\Id^{(3)}$, let  $\dl_0^{(2)}>0$ be such that  for all discrete interval communities on $I_C$ with distance $\dl < \dl_0^{(2)}$ we have that
\begin{enumerate}
\item[a)]  for $y \in \Id^{(1)}$  we have that
$$\xd(y) \in (y,mid(I_C)),$$
\item[b)]  for $y \in \Id^{(3)}$ we have that
$$\xd(y) \in (mid(I_C),y),$$
\item[c)]
and for $y \in \Id^{(2)}$ we have that
$$  \xd(y) \in  (mid(I_C) - \Dl_Q - \Dl_x, mid(I_C) + \Dl_Q + \Dl_x).$$
\end{enumerate}
Note that such a  $\dl_0^{(2)}$ exists by Lemmas~\ref{lemma:xd}. 

Let
$$\dlO = \min \{ \dl_0^{(1)}, \dl_0^{(2)} \}.$$ 

It then follows that for all that  for all discrete interval communities on $I_C$ with distance  $\dl < \dlO$, we have that
$$ \xd(y) = [mid(I_C) - \LSd, mid(I_C) + \LSd], \qquad y \in \intsd,$$
and
$$supp(\QCsd(\cdot)) \subseteq [mid(I_C) - \LSd, mid(I_C) + \LSd].$$
As by construction we have that
$$ \LSd < L_C,$$
the result of the lemma then follows.
\end{proof}


\section{Proof of Proposition~\ref{prop:QCd}}\label{app:proof_QCd}
In this appendix we prove Proposition~\ref{prop:QCd} using the results of Appendix~\ref{app:QCd}. Recall that Proposition~\ref{prop:QCd} characterizes the content supply function $\QCd(x)$ for a discrete interval community  $\dC = \Cd$ that is part of a \erNash as given by  Proposition~\ref{prop:dnash}.

We obtain the result of  Proposition~\ref{prop:QCd} using the results of Appendix~\ref{app:QCd} as follows.
Let $\dcommStrucNash$ be a \erNash as given by Proposition~\ref{prop:dnash}. Then we have that for every discrete interval community $\dC = \Cd \in \sCdNash$ that
$$ L_C < \frac{L}{2}$$
where $L_C$ is the length of the interval $I_C$ such that
$$ \intdd = \Add \cap I_C,$$
as well as
$$ \aCd(y) = E_p, \qquad y \in \intdd,$$
and
$$ \bCdNash(\cdot|y) = E_q \delta(\xCd(y) - x), \qquad y \in \intsd$$
where
$$ \xCd(y) = \arg \max_{x \in \setR} q(x|y)\PCd(x),$$
and $\delta$ is the Dirac delta function.

As a result, the properties of a discrete interval community  $\dC = \Cd \in \sCdNash$ satisfy the assumption in Lemma~\ref{lemma:QCd1}~and~\ref{lemma:QCd2}. Proposition~\ref{prop:QCd} then follows immediately from Lemma~\ref{lemma:QCd1}~\ref{lemma:QCd2}.

\newpage
\section{Proof of Proposition~\ref{prop:UCdd}}\label{app:UCdd}
In this appendix we  prove Proposition~\ref{prop:UCdd}. To do that, we first study for a given discrete interval communities $\dC =\Cd$ the properties of the function $\FCdd(y)$  given by
$$ \FCdd(y) = E_p E_q \sum_{z \in \intsd} \int_{x \in \setR} \Big [ p(\xd(z)|y) q(\xd(z)|z) - c \Big ]dx, \qquad y \in \setR.$$
where $E_p$, $0 < E_p \leq 1$, and $E_q$, $E_q >0$, are the bounds on the agents' content consumption and production rates.

Our first lemma shows that $\dls \FCdd(y)$ closely approximates the function $\FCd(y)$ that we used for our analysis of interval communities $C = (I_C,I_C)$ in~\cite{continuous_model_arxiv}. More precisely, we have the following result.
\begin{lemma}\label{lemma:Riemann_FCdd}
Consider  constants  $E_p$, $0 < E_p \leq 1$,  and $E_q$, $0< E_q$, and an interval $I_C  =   [ mid(I_C) - L_C, mid(I_C) + L_C) \subset \setR$. 
For every $\Dl_U >0$ there exists a $\dlO >0$ such that for all discrete interval communities $\dC =\Cd$ on $I_C$ with distance $\dl < \dlO$ the following is true. Suppose that we have that
$$2 L_C < \min \{ b, L \}$$
where $b$ is the constant of Assumption~\ref{ass:fg},
as we well as
$$\aCd(y) = E_p, \qquad y \in \intdd.$$
Then  the function $\FCdd(y)$ given by
$$ \FCdd(y) = E_pE_q \sum_{z \in \intsd} \Big [ p(\xd(z)|y) q(\xd(z)|z) - c \Big ]dx, \qquad y \in \setR, $$
where
$$\xd(z) =  \arg\max_{x \in \setR} q(x|z) \PCd(x), \qquad z \in \intsd,$$
and
$$\PCd(x) = E_p \sum_{y \in \intdd} p(x|y),$$
has the property that
$$\Big |\dls \FCdd(y) - \FCd(y) \Big | < \Dl_U, \qquad y \in \setR,$$
for
$$\FCd(y) =  E_p E_q \int_{z \in I_C} \Big [ p(\xs(z)|y) q(\xs(z)|z) - c \Big ]dz$$
where
$$ \xs(z) =  \arg\max_{x \in \setR} q(x|z) P_C(x), \qquad z \in I_C,$$
and
$$P_C(x) = E_p \int_{I_C} p(x|y) dy.$$
\end{lemma}

\begin{proof}
Note that
\begin{eqnarray}\label{eq:err_I1}
&& \Big |\dls \FCdd(y) - \FCd(y) \Big | \\
&=& E_p E_q \left | \dls \sum_{z \in \intsd} \Bsbl  p(\xd(z)|y)q(\xd(z)|z) - c \Bsbr  -  \int_{z \in I_C} \Big [ p(\xs(z)|y) q(\xs(z)|z) - c \Big ]dz \right | \nonumber\\
&\leq& \dls \sum_{z \in \intsd} \Big | p(\xd(z)|y) q(\xd(z)|z) - p(\xs(z)|y) q(\xs(z)|z) \Big | + \cdots \nonumber\\
&&+   \left |  \dls \sum_{z \in \intsd} \Big [ p(\xs(z)|y) q(\xs(z)|z) - c \Big ] -  \int_{z \in I_C} \Big [ p(\xs(z)|y) q(\xs(z)|z) - c \Big ] dz \right |. \nonumber
\end{eqnarray}

In the following we show that  there exists a $\dlO >0$ such that for all discrete interval communities on $I_C$ with distance $\dl < \dlO$ we have that
\begin{equation}\label{eq:err_I1_1}
\dls \sum_{z \in \intsd} \Big | p(\xd(z)|y) q(\xd(z)|z) - p(\xs(z)|y) q(\xs(z)|z) \Big | < \frac{\Dl_U}{2E_p E_q}
\end{equation}
and
\begin{equation}\label{eq:err_I1_2}
\left |  \dls \sum_{z \in \intsd} \Big [ p(\xs(z)|y) q(\xs(z)|z) - c \Big ] -  \int_{I_C} \Big [ p(\xs(z)|y) q(\xs(z)|z) - c \Big ] dz \right | < \frac{\Dl_U}{2 E_p E_q}.
\end{equation}
Using Eq.~\ref{eq:err_I1}, these two results establish that
$$E_p E_q \left | \dls \sum_{z \in \intsd}   \Bsbl p(\xd(z)|y)q(\xd(z)|z)  - c \Bsbr -  \int_{I_C} \Big [ p(\xs(z)|y) q(\xs(z)|z) - c \Big ]dz \right | <  \Dl_U,$$
which is the bound that we wish to obtain.

In order to derive these two results, we first consider the term given by Eq.~\eqref{eq:err_I1_1}
for which we have that
\begin{eqnarray*}
&\dls& \sum_{z \in \intsd} \Big | p(\xd(z)|y) q(\xd(z)|z) - p(\xs(z)|y) q(\xs(z)|z) \Big | \\
&\leq& \dls \sum_{z \in \intsd} p(\xd(z)|y) \Big | q(\xd(z)|z) - q(\xs(z)|z) \Big | +   \dls \sum_{z \in \intsd} q(\xs(z)|z) \Big | p(\xd(z)|y) - p(\xs(z)|y) \Big |  \\
&\leq& \dls \sum_{z \in \intsd}  \Big | q(\xd(z)|z) - q(\xs(z)|z) \Big | +   \dl \sum_{z \in \intsd} \Big | p(\xd(z)|y) - p(\xs(z)|y) \Big |    \\
&\leq& \dls \max_{x \in supp(g(\cdot))} |g'(x)|\sum_{z \in \intsd} ||\xd(z) - \xs(z)| | + \max_{x \in [0,L]} |f'(x)| \sum_{z \in \intsd} ||\xd(z) - \xs(z) ||,
\end{eqnarray*}
where we used the fact that by Assumption~\ref{ass:fg} the derivatives of the function $f$ and $g$ that define $p(x|y)$ and $q(x|y)$ are bounded. In particular, by Assumption~\ref{ass:fg} the functions $f$ and $g$ have bounded derivatives, i.e. we have that
$$ \left | f'(x) \right | < M_f, \qquad x \in [0,L],$$
and
$$ \left | g'(x) \right | < M_g, \qquad x \in supp(g(\cdot)).$$
It then follows that
\begin{eqnarray}\label{eq:I_1_1}
 \dls \sum_{z \in \intsd} \Big | p(\xd(z)|y) q(\xd(z)|z) - p(\xs(z)|y) q(\xs(z)|z) \Big |
 < \dls ( M_f +  M_g)  \sum_{z \in \intsd} ||\xd(z) - \xs(z) ||.
\end{eqnarray}

Next we consider the term given by Eq.~\eqref{eq:err_I1_2},
$$ \left |  \dls \sum_{z \in \intsd} \Big [ p(\xs(z)|y) q(\xs(z)|z) - c \Big ] -  \int_{z \in I_C} \Big [ p(\xs(z)|y) q(\xs(z)|z) - c \Big ] dz \right |.
$$
Note that the expression
$$ \dls \sum_{z \in \intsd} \Big [ p(\xs(z)|y) q(\xs(z)|z) - c \Big ]$$
can be interpreted as  the Riemann sum approximation of the integral
$$\int_{z \in I_C} \Big [ p(\xs(z)|y) q(\xs(z)|z) - c \Big ] dz,$$
and we can use the same argument as given in Lemma~\ref{lemma:Riemann_PCd} to establish an upper bound on the  term given by Eq.~\eqref{eq:err_I1_2}.

More precisely, by Assumption~\ref{ass:fg} the functions $f$ and $g$ have bounded derivatives, i.e. we have that
$$ \left | f'(x) \right | < M_f, \qquad x \in [0,L],$$
and
$$ \left | g'(x) \right | < M_g, \qquad x \in supp(g(\cdot)).$$
Furthermore, from the analysis in~\cite{continuous_model_arxiv} we have that there exists a constant $M_z > 0$ such that
$$ \left | \frac{d}{dz} \xs(z) \right | < M_z, \qquad y \in I_C.$$
We then have that
\begin{eqnarray*}
\left | \frac{d}{dz}  p(\xs(z)|y) q(\xs(z)|z) \right |
&\leq&    \left | f'(||\xs(z) - y||) \right | \cdot  \left| \frac{d}{dz} \xs(z) \right |+ \\
&& \left  | g'(||\xs(z) - z||) \right | \cdot \left ( \left| \frac{d}{dz} \xs(z) \right | + 1 \right ) \\
&<& M_fM_z + M_g(M_z+1).
\end{eqnarray*}
Using a similar argument as given in the proof of Lemma~\ref{lemma:Riemann_PCd}, one can then show that for
$$M =   M_fM_z + M_g(M_z+1),$$
we have that
\begin{eqnarray*}
 \left |  \dls \sum_{z \in  \intsd} \Big [ p(\xs(z)|y) q(\xs(z)|z) - c \Big ] -  \int_{z \in I_C} \Big [ p(\xs(z)|y) q(\xs(z)|z) - c \Big ] dz \right | \\ 
< 2 (ML_C + 1) \dls.
\end{eqnarray*}
As by definition we have that
$$\dls < \dl,$$
it follows that
\begin{eqnarray}\label{eq:I_1_2} \nonumber
 \left |  \dls \sum_{z \in  \intsd} \Big [ p(\xs(z)|y) q(\xs(z)|z) - c \Big ] -  \int_{z \in I_C} \Big [ p(\xs(z)|y) q(\xs(z)|z) - c \Big ] dz \right | \\ 
< 2 (ML_C + 1) \dl.
\end{eqnarray}

Let $\dlO$,
$$ 0 < \dlO < \frac{\Dl_U}{4 E_p E_q(ML_C + 1)} $$
be such that for  all discrete interval communities on $I_C$ with distance $\dl < \dlO$ we have that
$$||\xd(z) - \xs(z)|| < \frac{\Dl_U}{8 E_p E_qL_C(M_f+M_g)}, \qquad z \in \setR.$$
Note that such  a $\dlO$ exists by Lemma~\ref{lemma:xd_xs}.
Using this choice of $\dlO$, it the follows from Eq.~\eqref{eq:I_1_2} that for the error term given by Eq.~\eqref{eq:err_I1_2} we have that
$$  \left |  \dls \sum_{z \in \intsd} \Big [ p(\xs(z)|y) q(\xs(z)|z) - c \Big ] -  \int_{z \in I_C} \Big [ p(\xs(z)|y) q(\xs(z)|z) - c \Big ] dz \right | < \frac{\Dl_U}{2E_p E_q},$$

Therefore, in order to prove the lemma it remains to show that  for the above choice of $\dlO$ we obtain for the error term in Eq.~\eqref{eq:err_I1_1} that
$$ \dls \sum_{z \in \intsd} \Big | p(\xd(z)|y) q(\xd(z)|z) - p(\xs(z)|y) q(\xs(z)|z) \Big | < \frac{\Dl_U}{2E_p E_q}.$$
As by construction we have  for  all discrete interval communities on $I_C$ with distance $\dl < \dlO$ that
$$||\xd(z) - \xs(z)|| < \frac{\Dl_U}{8 E_p E_qL_C(M_f+M_g)}, \qquad z \in \setR,$$
it follows from  Eq.~\eqref{eq:I_1_2} that
$$
 \dls \sum_{z \in \intsd} \Big | p(\xd(z)|y) q(\xd(z)|z) - p(\xs(z)|y) q(\xs(z)|z) \Big |
 < \frac{\Dl_U}{8 E_p E_q L_C}   \sum_{z \in \intsd} \dl_s.
$$
Note that we have that
$$\sum_{z \in \intsd} \dls < 2L_C + \dls.$$
As by assumption we have that
$$ 2L_C > \dls,$$
we obtain that
$$\sum_{z \in \intsd} \dls < 4 L_C$$
and 
$$ \dls \sum_{z \in \intsd} \Big | p(\xd(z)|y) q(\xd(z)|z) - p(\xs(z)|y) q(\xs(z)|z) \Big | < \frac{\Dl_U}{2E_p E_q}.$$
The result of the lemma then follows.

\end{proof}

The next lemma provides additional properties of the function $\FCdd(y)$, $y \in \setR$.

\begin{lemma}\label{lemma:FCdd}
Consider  constants  $E_p$, $0< E_p\leq1 $, and $E_q$, $0< E_q$, and an interval $I_C  =   [ mid(I_C) - L_C, mid(I_C) + L_C) \subset \setR$. 
For every $\Dl_U >0$ there exists a $\dlO >0$ such that for all discrete interval communities $\dC =\Cd$ on $I_C$ with distance $\dl < \dlO$ the following is true.
 Suppose that we have that
$$2 L_C < \min \{ b, L \}$$
where $b$ is the constant of Assumption~\ref{ass:fg},
as we well as
$$\aCd(y) = E_p, \qquad y \in \intdd.$$
Then  the function $\FCdd(y)$ given by
$$\FCdd(y) =  E_p E_q \sum_{z \in \intsd} \Big [ p(\xd(z)|y) q(\xd(z)|z) - c \Big ], \qquad y \in \setR, $$
where
$$\xd(z) =  \arg\max_{x \in \setR} q(x|z) \PCd(x), \qquad z \in \intsd,$$
and
$$\PCd(x) = E_p \sum_{y \in \intdd} p(x|y),$$
has the properties that
\begin{enumerate}
\item[a)]  $\FCdd(y)$ is differentiable in on $\setR$. 
\item[b)]  $\FCdd(y)$ is strictly increasing on $[mid(I_C) - L,mid(I_C) - \Dl_U]$, and we have that
$$\frac{d}{dy} \FCdd(y) > 0, \qquad y \in [mid(I_C) - L,mid(I_C) - \Dl_U].$$
\item[c)]  $\FCd(y)$ is strictly decreasing on $[mid(I_C) + \Dl_U, mid(I_C) + L)$, and we have that
$$\frac{d}{dy} \FCdd(y) < 0, \qquad y \in  [mid(I_C) + \Dl_U, mid(I_C) + L).$$
\end{enumerate}
\end{lemma}

\begin{proof}
We first prove part a) of the lemma. To do this, we note that by Assumption~\ref{ass:fg} the function $p(x|y)$, $x \in \setR$, is differentiable in $y$ for $y \in \setR$. Using this result, it then follows that the function $ \FCdd(y)$ is differentiable in $y$ on $\setR$.

We next prove part b) of the lemma. 
To do that, we consider the function $\FCd(y)$, $y \in \setR$, for the interval community $C = (I_C,I_C)$ given by
$$\FCd(y) =  E_p E_q \int_{z \in I_C} \Big [ p(\xs(z)|y) q(\xs(z)|z) - c \Big ] dz,$$
where
$$ \xs(z) =  \arg\max_{x \in \setR} q(x|z) P_C(x), \qquad z \in I_C.$$
From the analysis in~\cite{continuous_model_arxiv} we have that the function $\FCd(y)$ is continuous on $y$ and strictly increasing in $y$ on $[mid(I_C) - L,mid(I_C))$. This implies that  there exists a constant  $B > 0$ such that
$$\frac{d}{dy} \FCd(y) > B, \qquad y \in [mid(I_C) - L,mid(I_C) - \Dl_U].$$
Using this result, it follows that if we have that
$$ \left | \dls \frac{d}{dy} \FCdd(y) - \frac{d}{dy} \FCd(y) \right | < \frac{B}{2}, \qquad y \in  [mid(I_C) - L,mid(I_C) - \Dl_U],$$
then we have that
$$ \dls \frac{d}{dy} \FCdd(y) > \frac{B}{2}, \qquad y \in [mid(I_C) - L,mid(I_C) - \Dl_U],$$
and the function $\FCdd(y)$ is strictly increasing on $ [mid(I_C) - L,mid(I_C) - \Dl_U]$. 

By Assumption~\ref{ass:fg} we have that the function $f$ has bounded first and second derivatives, and  we can use  a similar argument as given in proof of Lemma~\ref{lemma:Riemann_PCd} to show that there exists a $\dl^{(1)}_0 > 0$
such that for all discrete interval communities on $I_C$ with distance $\dl < \dl^{(1)}_0$ we have that
$$\dls \sum_{z \in \intsd} \left | \frac{d}{dy}p(\xd(z)|y) q(\xd(z)|z) - \frac{d}{dy}p(\xs(z)|y) q(\xs(z)|z) \right | < \frac{B}{2E_p E_q}$$
and
$$
\left |  \dls \sum_{z \in \intsd} \Big [ \frac{d}{dy}p(\xs(z)|y) q(\xs(z)|z) - c \Big ] -  \int_{I_C} \Big [ \frac{d}{dy} p(\xs(z)|y) q(\xs(z)|z) - c \Big ] dz \right | < \frac{B}{2 E_p E_q}.$$
Note that this results establishes part b) of the lemma.

Using the same argument as given for part b), we can prove part c) of the lemma, i.e. prove that there exists a $\dl^{(2)}_0 >0$
such that for all discrete interval communities on $I_C$ with distance $\dl < \dl^{(2)}_0$ we have that the $\FCd(y)$ is strictly decreasing on $[mid(I_C) + \Dl_U, mid(I_C) + L)$ and we have that
$$\frac{d}{dy} \FCdd(y) < 0, \qquad y \in  [mid(I_C) + \Dl_U, mid(I_C) + L).$$

The result of the lemma then follows by setting
$$\dlO = \min \Bsl \dl^{(1)}_0, \dl^{(2)}_0 \Bsr.$$
\end{proof}


\subsection{Proof of Proposition~\ref{prop:UCdd}}\label{app:proof_UCdd}
We are now in the position to prove Proposition~\ref{prop:UCdd}. Recall that Proposition~\ref{prop:UCdd} characterizes the content consumption  utility rate function  $\UCdd(y)$ for a discrete interval community  $\dC = \Cd$ that is part of a \erNash as given by  Proposition~\ref{prop:dnash}.

We obtain the result of  Proposition~\ref{prop:UCdd} using the results of Lemma~\ref{lemma:FCdd} as follows.
Let $\dcommStrucNash$ be a \erNash as given by Proposition~\ref{prop:dnash}. Then we have that for every discrete interval community $\dC = \Cd \in \sCdNash$ that
$$ \aCd(y) = E_p, \qquad y \in \intdd,$$
and
$$ L_C < \frac{L}{2},$$
where $L_C$ is the length of the interval $I_C$ such that
$$ \intdd = \Add \cap I_C.$$
Furthermore, for a  discrete interval community  $\dC = \Cd \in \sCdNash$ as given  by Proposition~\ref{prop:dnash} we have that
$$\UCdd(y) =  E_p E_q \sum_{z \in \intsd} \Big [ p(\xd(z)|y) q(\xd(z)|z) - c \Big ], \qquad y \in \intdd, $$
where
$$\xd(z) =  \arg\max_{x \in \setR} q(x|z) \PCd(x), \qquad z \in \intsd,$$
and
$$\PCd(x) = E_p \sum_{y \in \intdd} p(x|y).$$
As a result, the properties of a discrete interval community  $\dC = \Cd \in \sCdNash$ satisfy the assumptions in Lemma~\ref{lemma:FCdd} and we have that
$$\UCdd(y) =  \FCdd(y) = E_p E_q \sum_{z \in \intsd} \Big [ p(\xd(z)|y) q(\xd(z)|z) - c \Big ], \qquad y \in \intdd,$$
where the function $\FCdd(y)$ is as given in Lemma~\ref{lemma:FCdd}.
Proposition  then follows immediately from Lemma~\ref{lemma:FCdd}.

\newpage
\section{Proof of Proposition~\ref{prop:UCsd}}\label{app:UCsd}
In this appendix we prove  Proposition~\ref{prop:UCsd}. To do this, we first study  for a given discrete interval communities $\dC =\Cd$ the properties of the function $\FCdd(y)$  given by
$$\FCsd(y) = E_q [  q(\xd(y)|y) \PCd(\xd(y)- \aCd c], \qquad  x \in \setR,$$
where $E_q$, $0 < E_q$, is the bound on the content production rate of agent $y \in \Add$, $\PCd(x)$ is the demand function given by
$$ \PCd(x) = E_p \sum_{y \in \intdd} p(x|y),$$
as well as
$$\xd(y) = \arg \max_{x \in \setR} q(x|y)\PCd(x)$$
and
$$\aCd = \sum_{y \in \intdd} E_p.$$

Our first lemma shows that $\dld \FCsd(y)$ closely approximates the function $\FCd(y)$ that we used for our analysis of interval communities $C = (I_C,I_C)$ in~\cite{continuous_model_arxiv}. More precisely, we have the following result.

\begin{lemma}\label{lemma:Riemann_FCsd}
Consider  constants  $E_p$, $0 < E_p \leq 1$, and $E_q$, $0< E_q$, and an interval $I_C  =   [ mid(I_C) - L_C, mid(I_C) + L_C) \subset \setR$. 
For every $\Dl_U >0$ there exists a $\dlO >0$ such that for all discrete interval communities $\dC =\Cd$ on $I_C$ with distance $\dl < \dlO$ the following is true. If we have that
$$\aCd(y) = E_p, \qquad y \in \intdd,$$
and
$$2 L_C < \min \{ b, L \}$$
where $b$ is the constant of Assumption~\ref{ass:fg}, then  the function $\FCsd(y)$ given by
$$ \FCsd(y) = E_q \Big [  q\big (\xd(y)|y \big ) \PCd(\xd(y))- \aCd c \Big ], \qquad  y \in \setR,$$
where
$$\xd(y) = \arg \max_{x \in \setR} q(x|y)\PCd(x)$$
and
$$\aCd = \sum_{y \in \intdd} E_p,$$
has the property that
$$\Big |\dld \FCsd(y) - \FCd(y) \Big | < \Dl_U, \qquad y \in I_C,$$
where
$$ \FCd(y) = E_q \Big [  q \big (\xs(y)|y \big ) P_C(\xs(y))- \aC c \Big ]$$
with
$$\xs(y) = \arg \max_{x \in \setR} q(x|y)P_C(x)$$
and
$$\aC = 2 L_C E_p.$$
\end{lemma}


\begin{proof}
Let $M_f$ and $M_g$ be such that
$$|f'(x)| < M_f, \qquad x \in [0,L]$$
and
$$|g'(x)| < M_g, \qquad x \in supp(g(\cdot)).$$
Note that such constants $M_f$ and $M_g$ exist by Assumption~\ref{ass:fg}.

We then have that
$$\left | \frac{d}{dx} P_C(x) \right | = \int_{y \in I_C} \aC(y) \left | \frac{d}{dx} p(x|y) \right | dy
< E_p \int_{y \in I_C} M_f dy = E_p 2 L_C M_f.$$
Let
$$B_P =  2 E_p L_C M_f$$
and
$$B = \max\Bsl  P_C \big (mid (I_C) \big ) M_g, B_P\Bsr.$$
Using this definition,  let $\dlO$,
$$0 < \dlO < \frac{\Dl_U}{16 E_p E_q c},$$
be such that 
for all discrete interval communities on $I_C$ with distance $\dl < \dlO$, we have that
$$|| \xs(y) - \xd(y)|| < \frac{\Dl_U}{4  E_q B}, \qquad y \in \setR,$$
and
$$\big | P_C(\xd(y)) -  \dld P_{\dl,C}(\xd(y)) \big | < \frac{\Dl_U}{4 E_q }, \qquad y \in \setR.$$
Note that such a $\dlO$ exists by Lemma~\ref{lemma:Riemann_PCd} and Lemma~\ref{lemma:xd_xs}.

Then we have for $y \in \setR$ that 
\begin{eqnarray*}
&&  \Big | P_C(\xs(y))q(\xs(y)|y) - \dld \PCd(\xd(y)) q(\xd(y)|y) \Big | \\
&& \hsa \leq \hsb \Big | P_C(\xs(y))q(\xs(y)|y) - P_C(\xd(y)) q(\xd(y)|y) \Big |  + \cdots \\
&& \hsc + \Big |P_C(\xd(y)) q(\xd(y)|y) -  \dld \PCd(\xd(y)) q(\xd(y|y) \Big | \\
&& \hsa \leq  \hsb q(\xs(y)|y) \Big | P_C(\xs(y)) - P_C(\xd(y)) \Big | + \cdots \\
&& \hsc +  P_C(\xs(y)) \Big |q(\xs(y)|y) - q(\xd(y)|y) \Big | + \cdots \\
&&  \hsc +  q(\xd(y)|y)   | P_C(\xd(y) -  \dld \PCd(\xd(y))| \\
&& \hsa \leq  \hsb  \big | P_C(\xs(y)) - P_C(\xd(y))\big | + \cdots \\
&& \hsc +  P_C(\xs(y)) \Big |q(\xs(y)|y) - q(\xd(y)|y) \Big | + \cdots \\
&& \hsc +   \big | P_C(\xd(y) -  \dld \PCd(\xd(y))\big | \\
&& \hsa <  \hsb B_P ||\xs(y) - \xd(y)|| + \cdots \\
&& \hsc + P_C \big (mid(I_C) \big )  M_g ||\xs(y) - \xd(y)|| + \cdots \\
&& \hsc +  \big | P_C(\xd(y) -  \dld \PCd(\xd(y))\big |,
\end{eqnarray*}
where we used the result from the analysis in~\cite{continuous_model_arxiv} which states that 
$$mid(I_C) = \arg \max_{x \in \setR} P_C(x).$$

By construction we have for $\dl < \dlO$ that
\begin{eqnarray*}
&& B_P ||\xs(y) - \xd(y)|| \\
&& \hsc + P_C \big (mid(I_C) \big ) M_g ||\xs(y) - \xd(y)|| \\
&& \hsc +  \big | P_C(\xd(y) - \dld \PCd(\xd(y))\big | \\
&& \hsa < \hsb  B \frac{\Dl_U}{4 E_q B} + B \frac{\Dl_U}{4 E_q B} + \frac{ \Dl_U}{4 E_q} \\
&& \hsa = \hsb   \frac{3}{4}\frac{\Dl_U}{ E_q }.
\end{eqnarray*}
Recall that we have that
$$\aC = 2 L_C E_p$$
and
$$
\dld \aCd = \dld \sum_{y \in \intdd} E_p = 2 \Ld E_p + \dld E_p = 2 E_p \Big [ \Ld + \dld/2 \Big ].$$
Combining the above result with the fact that
$$\left |L_C - \Ldd \right | \leq \dld$$
and
$$0 < \dl < \dlO < \frac{\Dl_U}{16 E_p E_q c},$$
we obtain for $y \in \setR$ that
\begin{eqnarray*}
&& \left | \dld \FCsd(y) - \FCd(y) \right | \\ 
&& \hsa \leq \hsb   E_q \Big | P_C(x^*)q(\xs(y)|y) - \dld \PCd(\xd(y)) q(\xd(y)|y) \Big | + \cdots \\
&& \hsc + E_q c \big |\aC - \dld\aCd \big | \\
&& \hsa < \hsb   E_q  \frac{3}{4} \frac{\Dl_U}{ E_q }
+ E_q E_p c 2 \Big |L_C - \Ldd - \dld/2 \Big | \\
&& \hsa \leq \hsb     \frac{3}{4} \Dl_U
+ 2 E_q E_p c \Big |L_C - \Ldd - \dld/2 \Big |\\
&& \hsa \leq \hsb    \frac{3}{4} \Dl_U
+ 2 E_q E_p c \dld + E_q E_p c  \\
&& \hsa < \hsb    \frac{3}{4} \Dl_U
+ 4 E_q E_p c \dld \\
\end{eqnarray*}
As by definition we have that
$$\dld < \dl$$
and by construction we have that
$$0 < \dl < \dlO < \frac{\Dl_U}{16 E_p E_q c},$$
it follows from the above result that
\begin{eqnarray*}
&& \left | \dld \FCsd(y) - \FCd(y) \right | \\
&& \hsa < \hsb    \frac{3}{4} \Dl_U
+ 4 E_q E_p c \dl \\
&& \hsa < \hsb    \frac{3}{4}\Dl_U
+ 4 E_q E_p c \frac{\Dl_U}{ 16 E_q E_p c } \\
&& \hsa = \hsb  \Dl_U.
\end{eqnarray*}
The result of the lemma then follows.
\end{proof}


The next two lemmas provide additional properties of the function $\FCsd(y)$, $y \in \setR$.
\begin{lemma}\label{lemma:FCsd}
Consider  constants  $E_p$, $0 < E_p \leq 1$, and $E_q$, $0< E_q$, and an interval $I_C  =   [ mid(I_C) - L_C, mid(I_C) + L_C) \subset \setR$. 
For every $\Dl_U$, $0 < \Dl_U <L_C$ there exists a $\dlO >0$ such that for all discrete interval communities $\dC =\Cd$ on $I_C$ with distance $\dl < \dlO$, the following is true.
If we have that
$$\aCd(y) = E_p, \qquad y \in \intdd,$$
and
$$2 L_C < \min \{ b, L \}$$
where $b$ is the constant of Assumption~\ref{ass:fg}, then  the function $\FCsd(y)$ given by
$$ \FCsd(y) = E_q \Big [  q\big (\xd(y)|y \big ) \PCd(\xd(y))- \aCd c \Big ], \qquad  y \in \setR,$$
where
$$\xd(y) = \arg \max_{x \in \setR} q(x|y)\PCd(x)$$
and
$$\aCd = \sum_{y \in \intdd} E_p,$$
has the properties that
\begin{enumerate}
\item[a)]  $\FCsd(y)$ is strictly increasing and twice continuously differentiable on $[mid(I_C) - L_C,mid(I_C) - \Dl_U]$, and we have that
$$\frac{d}{dy} \FCsd(y) > 0, \qquad y \in [mid(I_C)-L_C,mid(I_C) - \Dl_U].$$
\item[b)]  $\FCsd(y)$ is strictly decreasing and twice continuously differentiable on $[mid(I_C)+\Dl_U, mid(I_C) + L_C]$, and we have that
$$\frac{d}{dy} \FCsd(y) < 0, \qquad y \in  [mid(I_C)+\Dl_U, mid(I_C) + L_C].$$
\end{enumerate}

\end{lemma}

\begin{proof}
We prove part a) of the lemma. Part b) can be shown using the same argument.

Let
$$y_0 = mid(I_C),$$
and let
$$y_l = mid(I_C) - \Dl_U.$$

Using these definitions, let $\dlO>0$ be such that  for all discrete interval communities on $I_C$ with distance $\dl < \dlO$ we have that
\begin{enumerate}
\item[a)] there exists a unique solution to the optimization problem
$$\xd(y) =  \arg \max_{x \in \setR} q(x|y) \PCd(x), \qquad y \in [y_0-L_C,y_0 - \Dl_U]$$
and $\xd(y)$ is twice continuously differentiable on $[y_0-L_C,y_0 - \Dl_U]$.
\item[b)] $\xd(y_l) < y_0$ and
$$\xd(y) \in [y, \xs(y_l)],  \qquad y \in [y_0-L_C,y_0 - \Dl_U].$$
\item[c)] the function $\PCd(x)$ is increasing and strictly concave on $[y_0 - L_C,y_0 - \xd(y_l)]$.
\end{enumerate}
Note that such a $\dlO$ exists by Lemma~\ref{lemma:PCd}~and~\ref{lemma:xd}. 

Using this construction,
we show next that for
$$ 0< \dl < \dlO$$
the function $\FCsd(y)$ is strictly increasing and twice continuously differentiable on $[y_0 - L_C,y_l]$

To do this, we use the fact that by construction  there exists a unique solution to the optimization problem
$$\xd(y) =  \arg \max_{x \in \setR} q(x|y) \PCd(x), \qquad y \in [y_0-L_C,y_0 - \Dl_U],$$
and we have that 
$$\xd(y) \in [y, \xs(y_l)],  \qquad y \in [y_0-L_C,y_0 - \Dl_U]$$
with
$$\xd(y_l) < y_0.$$
Furthermore, by construction the function $\PCd(x)$ is increasing and strictly concave on $[y_0 - L_C,y_0 - \xd(y_l)]$. Using these properties, we can  show that the function  $\FCsd(y)$ is strictly increasing and twice continuously differentiable for
$$ y \in [y_0 - L_C,y_0 - \Dl_U],$$
using the same argument as given in~\cite{continuous_model_arxiv} to show that the function  $\FCs(y)$ is strictly increasing and twice continuously differentiable for
$$ y \in [y_0 - L_C,y_0].$$
The result of the lemma then follows.
\end{proof}

\begin{lemma}\label{lemma:FCsd2}
Consider  constants  $E_p$, $0 < E_p \leq 1$, and $E_q$, $0< E_q$, and an interval $I_C  =   [ mid(I_C) - L_C, mid(I_C) + L_C) \subset \setR$. 
For every $\Dl_U$, $0 < \Dl_U <L_C$ there exists a $\dlO >0$ such that for all discrete interval communities $\dC =\Cd$ on $I_C$ with distance $\dl < \dlO$, the following is true.
If we have that
$$\aCd(y) = E_p, \qquad y \in \intdd,$$
and
$$2 L_C < \min \{ b, L \}$$
where $b$ is the constant of Assumption~\ref{ass:fg}, then  the function $\FCsd(y)$ given by
$$ \FCsd(y) = E_q \Big [  q\big (\xd(y)|y \big ) \PCd(\xd(y))- \aCd c \Big ], \qquad  y \in \setR,$$
where
$$\xd(y) = \arg \max_{x \in \setR} q(x|y)\PCd(x)$$
and
$$\aCd = \sum_{y \in \intdd} E_p,$$
has the properties that
\begin{enumerate}
\item[a)]  $\FCsd(y)$ is non-decreasing on $[mid(I_C) - L,mid(I_C) - L_C]$.
\item[b)]  $\FCsd(y)$ is non-increasing on $[mid(I_C)+L_C, mid(I_C) + L)$.
\end{enumerate}
\end{lemma}

\begin{proof}
We prove part a) of the lemma. Part b) can be shown using the same argument.

Let
$$y_l = mid(I_C) - L_C.$$
Furthermore, let $\dlO>0$ be such that  for all discrete interval communities on $I_C$ with distance $\dl < \dlO$ we have that
\begin{enumerate}
\item[a)] there exists a unique solution to the optimization problem
$$\xd(y_l) =  \arg \max_{x \in \setR} q(x|y) \PCd(x),$$
and
$$\xd(y_l) < mid(I_C),$$
\item[b)] the function $\PCd(x)$ is strictly concave on $[mid(I_C) - L_C,mid(I_C) - \xd(y_l)]$ and increasing on $[mid(I_C) - L,mid(I_C) - L_C]$,
\item[c)] for
$$ y \in [mid(I_C) - L,mid(I_C) - L_C],$$
we have
$$\xd(y) \leq \xd(y_l), \qquad y \in [mid(I_C) - L,mid(I_C) - L_C].$$
\end{enumerate}
Note that such a $\dlO$ exists by Lemma~\ref{lemma:PCd},~\ref{lemma:xd} and~\ref{lemma:xd2}. 

Using the fact that by construction we have that
$$\xd(y) \leq \xd(y_l) < mid(I_C), \qquad y \in [mid(I_C) - L,mid(I_C) - L_C],$$
and that the function $\PCd(x)$ is strictly concave on $[mid(I_C) - L_C,mid(I_C) - \xd(y_l)]$ and increasing on $[mid(I_C) - L,mid(I_C) - L_C]$, we can show that the function  $\FCsd(y)$ is non-decreasing on $[mid(I_C) - L,mid(I_C) - L_C]$
using the same argument as given in~\cite{continuous_model_arxiv} to show that the function  $\FCs(y)$ is non-decreasing on $[mid(I_C) - L,mid(I_C) - L_C]$.

The result of the lemma then follows.
\end{proof}


\subsection{Proof of Proposition~\ref{prop:UCsd}}
We are now in the position to prove Proposition~\ref{prop:UCsd}. Recall that Proposition~\ref{prop:UCsd} characterizes the content production utility rate function  $\UCsd(y)$ for a discrete interval community  $\dC = \Cd$ that is part of a \erNash as given by  Proposition~\ref{prop:dnash}.

We obtain the result of  Proposition~\ref{prop:UCsd} using the results of Lemma~\ref{lemma:FCsd} as follows.
Let $\dcommStrucNash$ be a \erNash as given by Proposition~\ref{prop:dnash}. Then we have that for every discrete interval community $\dC = \Cd \in \sCdNash$ that
$$ \aCd(y) = E_p, \qquad y \in \intdd,$$
and
$$ L_C < \frac{L}{2},$$
where $L_C$ is the length of the interval $I_C$ such that
$$ \intdd = \Add \cap I_C.$$
Furthermore, for a  discrete interval community  $\dC = \Cd \in \sCdNash$ as given  by Proposition~\ref{prop:dnash} we have that
$$\UCdd(y) =  E_p E_q \sum_{z \in \intsd} \Big [ p(\xd(z)|y) q(\xd(z)|z) - c \Big ], \qquad y \in \intdd, $$
where
$$\xd(z) =  \arg\max_{x \in \setR} q(x|z) \PCd(x), \qquad z \in \intsd,$$
and
$$\PCd(x) = E_p \sum_{y \in \intdd} p(x|y).$$
As a result, the properties of a discrete interval community  $\dC = \Cd \in \sCdNash$ satisfy the assumptions in Lemma~\ref{lemma:FCsd} and we have that
$$\UCdd(y) =  \FCdd(y) = E_p E_q \sum_{z \in \intsd} \Big [ p(\xd(z)|y) q(\xd(z)|z) - c \Big ], \qquad y \in \intdd,$$
where the function $\FCdd(y)$ is as given in Lemma~\ref{lemma:FCsd}.
Proposition  then follows immediately from Lemma~\ref{lemma:FCsd}.

\newpage
\section{Optimal Community Selection for Content Consumers and Producers}
In Appendices~\ref{app:PCd}~and~\ref{app:QCd}, we derived for a given discrete interval community $\dC=\Cd$  the properties of content demand function $\PCd(x)$ given by
$$\PCd(x) = \sum_{y \in \intdd} \aCd(y) p(x_y)$$
and content supply $\QCsd(x)$ given by
$$\QCsd(x) = \sum_{y \in \intsd} \bCsd(x|y) q(x|y),$$
under the  assumptions that
$$\aCd(y) = E_p, \qquad y \in \intdd,$$
and
$$||\bCsd(\cdot|y)|| = E_q, \qquad y \in \intsd,$$
where
$$\bCsd(\cdot|y) = \underset{\bCd(\cdot|y): || \bCd(\cdot|y)|| \leq E_q}{\arg\max} \int_{x \in \setR} \bCd(x|y) \Bsbl q(x|y) \PCd(x) - \aCd c \Bsbr dx$$
and
$$\aCd = \sum_{y \in \intdd} E_p.$$
Or in other words  we derived for a given discrete interval community $\dC=\Cd$  the properties of content demand function $\PCd(x)$ and content supply $\QCsd(x)$ under the assumptions that agents $y \in \intdd$, $y \in \intsd$ in community $C$ allocate all their production and consumption rate to community $C$. In this appendix we show that this is indeed optimal for an agent to allocate its production and consumption rate to a single community. 

To do this, we consider a given community structure $\CSd = \left ( \sCd, \{\aSCd(y)\}_{y \in \Add},  \{\bSCd(\cdot|y)\}_{y \in \Ads} \right )$ as defined in Section~\ref{section:dnash}; in particular, we have that $\sCd$ defines the set of communities $\dC = (\Cdd,\Cds)$,
$$ \Cdd \subset \Add \mbox{ and } \Cds \subset \Ads,$$
of the community structure $\CSd$, and  
$$\aSCd(y) = \{ \aCd(y) \}_{\dC \in \sCd}, \qquad y \in \Add,$$
and
$$\bSCd(y) = \{ \bCd(\cdot |y) \}_{\dC \in \sCd}, \qquad y \in \Ads,$$
indicate the  rates that content consumers and producers allocate to the different communities $\dC \in \sCd$.  We assume that the total content consumption and production rates of each agent can not exceed a given threshold, and we have that 
$$|| \aSCd(y) || = \sum_{\dC \in \sCd} \aCd(y) \leq E_p, \qquad y \in \Add,$$
and
$$|| \bSCd(y) || = \sum_{\dC \in \sCd} || \bCd(\cdot|y) || \leq E_q, \qquad y \in \Ads,$$
where
$$ || \bCd(\cdot|y) || = \int_{x \in R}  \bCd(x|y) dx.$$

More precisely, given a community structure $\CSd = \dcommStruc$, we  analyze in the following the situation where an agent $y \in \Add$ changes its rate allocation $\aSCd(y)$ to $\aSCd'(y)$ given by
$$\aSCd'(y)=  \{ \aCd'(y) \}_{\dC \in \sCd},$$
where $\aCd'(y)$ is the rate for content consumption allocated to community $\dC \in \sCd$ such that
$$ \sum_{\dC \in \sCd} \aCd'(y) \leq E_p.$$
In particular, we want to analyze  whether this change will increase the utility for contention consumption  of agent $y$, assuming that all other agents keep their rate allocation fixed. Given a community structure  $\CSd = \dcommStruc$, let $ \UCSdd(\aSCd'(y)|y)$ be the utility that agent $y \in \Add$ receives under the new allocation $\aSCd'(y)$, while all other agents keep their rate allocation fixed, i.e. we have that
$$  \UCSdd(\aSCd'(y)|y) =
\sum_{\dC \in \sCd} \aCd'(y) \int_{\setR} \Bsbl p(x|y) \QCd(x) - \bCd (x) c \Bsbr dx$$
with
$$\QCd(x) = \sum_{z \in \Ads} q(x|z) \bCd(x|z)$$
and
$$\bCd (x) = \sum_{z \in \Ads} \bCd(x|z).$$

Similarly, given a community structure $\CSd = \dcommStruc$, we  want to analyze the situation where an agent $y \in \Ads$ changes its rate allocation $\bSCd(y)$ to $\bSCd'(y)$ given by
$$\bSCd'(y)=  \{ \bCd'(\cdot|y) \}_{\dC \in \sCd},$$
where $\bCd'(\cdot|y)$ is the rate allocated for content production in community $\dC \in \sCd$ such that
$$ || \bSC'd(y) || \leq E_q.$$
Again, we analyze  whether this change will increase the utility for contention consumption  of agent $y$, assuming that all other agents keep their rate allocation fixed. Given a community structure  $\CSd = \dcommStruc$, let $ \UCSsd(\bSCd'(y)|y)$ be the utility that agent $y$ receives under the new allocation $\bSCd'(y)$, while all other agents keep their rate allocation fixed, i.e. we have that
$$  \UCSsd(\bSCd'(y)|y) = \sum_{\dC \in \sCd} \int_{\setR} \bCd'(x|y) \Bsbl q(x|y) \PCd(x) - \aCd c \Bsbr dx$$
with
$$\PCd(x) = \sum_{z \in \Add} \aCd(z) p(x|y)$$
and
$$ \aCd  = \sum_{z \in \Add} \aCd(z).$$

Our first result states that it is indeed optimal for agent $y \in \Add$ to allocate all its content consumption rate to a single community. That is,  given a community structure  $\CSd = \dcommStruc$ there exists a community 
$$\dC^*(y) \in \sCd$$
such that an optimal rate allocation  $\aSCd^*(y)$,
$$ \aSCsd(y) = \arg \max_{\aSCd'(y): || \aSCd'(y) || \leq E_p}  \UCSdd(\aSCd'(y)|y), \qquad y \in \Add, $$
is given by
$$ \aCsd(y) = \left \{
\begin{array}{rl}
E_p, & \dC = \dC^*(y) \\
0, & \dC \neq \dC^*(y).
\end{array} \right .$$
Or in other words,  given a community structure  $\CSd = \dcommStruc$ there exists an optimal rate allocation  $\aSCd^*(y)$ for agent $y \in \Add$ such that $y$ allocates all its rate to a single community $\dC^* \in \sCd$. 

\begin{lemma}\label{lemma:single_community_consumption}
Let $\dcommStruc$ be a given community structure, and let $y \in \Add$ be a given content consumer, and let $E_p$, $0 < E_p \leq 1$, be a positive scalar. If we have that
$$ \max_{\dC \in \sCd} \int_{x \in \setR}  \Bsbl p(x|y) \QCd(x) - \bCd(x) c \Bsbr > 0,$$
then an optimal content consumption rate allocation  $\aSCd^*(y)$,
$$ \aSCsd(y) = \arg \max_{\aSCd'(y): || \aSCd'(y) || \leq E_p}  \UCSdd(\aSCd'(y)|y),$$
for agent $y$ is given by
$$ \aCsd(y) = \left \{
\begin{array}{rl}
E_p, & \dC = \dC^*(y) \\
0, & \dC \neq \dC^*(y),
\end{array} \right .$$
where
$$\dC^*(y) = \arg \max_{\dC \in \sCd} \int_{x \in \setR} \Bsbl p(x|y) \QCd(x) - \bCd(x) c \Bsbr .$$
\end{lemma}

\begin{proof}
This lemma can be proved using the same argument as given in~\cite{continuous_model_arxiv} where we proved the same result for community structures  $\CS=\commStruc$ under the continuous agent model.
\end{proof}

Similarly, the next result states that it is optimal for an agent $y \in \Ads$ to allocate all its content production rate to a single community. That is,  given a community structure  $\CSd = \dcommStruc$ there exists a community
$$\dC^*(y) \in \sCd$$
such that an optimal rate allocation $\bSCd^*(y)$, 
$$ \bSCsd(y) = \arg \max_{\bSCd'(y): || \bSCd'(y) || \leq E_q}  \UCSsd(\bSCd'(y)|y), \qquad y \in \Add, $$
is given by
$$ || \bCsd(y) || = \left \{
\begin{array}{rl}
E_q, & \dC = \dC^*(y) \\
0, & \dC \neq \dC^*(y).
\end{array} \right .$$
Or in other words,  given a community structure  $\CSd = \dcommStruc$ there exists an optimal rate allocation  $\bSCd^*(y)$ for agent $y \in \Ads$ such that $y$ allocates all its rate to a single community $\dC^* \in \sCd$.

\begin{lemma}\label{lemma:single_community_production}
Let $\dcommStruc$ be a given community structure, and let $y \in \Ads$ be a given content producer, and let $E_q >0$ be a positive scalar. If we have that
$$ \max_{\dC \in \sCd} \Bsbl q(\xCd(y)|y)\PCd(x) - \aCd c \Big ]  > 0$$
where
$$\xCd(y) = \arg \max_{x \in \setR} q(x|y) \PCd(x),$$
then an optimal content consumption rate allocation  $\aSCd^*(y)$,
$$ \bSCsd(y) = \arg \max_{\bSCd'(y): || \bSCd'(y) || \leq E_q}  \UCSsd(\bSCd'(y)|y), $$
for agent $y$ is given by
$$ \bCsd(x|y) = \left \{
\begin{array}{rl}
E_q \dl(x^*_{\dl,\dC}(y) -x) , & \dC = \dC^*(y) \\
0, & \mbox{otherwise,}
\end{array} \right .$$
where
$$\dC^*(y) = \arg \max_{\dC \in \sCd} q(\xCd(y)|y)\PCd(x^*_{C,y})$$
and $\dl(\cdot)$ is the Dirac delta function.
\end{lemma}

\begin{proof}
This lemma can be proved using the same argument as given in~\cite{continuous_model_arxiv} where we proved the same result for community structures  $\CS=\commStruc$ under the continuous agent model.
\end{proof}

\newpage
\section{Proof of Proposition~\ref{prop:dnash}}\label{app:dnash}

In this appendix we prove  Proposition~\ref{prop:dnash}. To do that, we make a connection with the analysis for the continuous agent model in~\cite{continuous_model_arxiv}. In~\cite{continuous_model_arxiv}  we showed that there always exists a Nash equilibrium under the continuous model. Our next results states that for each  Nash equilibrium that we obtained in~\cite{continuous_model_arxiv} there exists a corresponding \erNash for the two-sided model, given that the distances $\dld$ and $\dls$ that define  the agent sets $\Add$ and $\Ads$ are small enough.  More precisely we have the following result. 

\begin{lemma}\label{lemma:dnash}
Suppose that
$$ f(0) g(0) - c > 0.$$
Furthermore let $E_p$, $0 < E_p \leq 1$, and $E_q$, $0< E_q$, be given positive scalars, and let $\CSNash = \commStrucNash$ be a Nash equilibrium  as given in~\cite{continuous_model_arxiv}. That is, the Nash equilibrium $\commStrucNash$ has the properties that
\begin{enumerate}
\item[(a)] each community $C=(I_C,I_C) \in \setC$ is an interval community
where the interval $I_C \subset \setR$ is given by
$$I_C = [mid(I_C) - L_C, mid(I_C) + L_C).$$
\item[(b)] the set $\{I_C\}_{C \in \setC}$ is a set of mutually non-overlapping intervals that covers $\setR$, i.e. we have that
  $$ I_C \cap I_{C'} = \emptyset, \qquad C,C' \in \setC, C \neq C',$$
  and 
  $$\cup_{C \in \setC} I_C = \setR.$$
\item[(c)] there exists a constant $L_I > 0$ such that
$$ L_C = L_I, \qquad C \in \setC*$$
where
$$ 2 L_I < \min \{ b, L \}$$
and  $b$ is the constant of Assumption~\ref{ass:fg}.
\end{enumerate}
Furthermore, for each community $C=(I_C,I_C) \in \setC$ we have that
\begin{enumerate}
\item[(a)] $\alpha^*_C(y) = E_p$, $y \in I_C$.
\item[(b)] $ \beta^*_C(\cdot|y) = E_q \delta(\xs(y) - x)$, $y \in I_C$, where  $\delta$ is the Dirac delta function and $\xs(y)$ is given by
$$ \xs(y) = \arg \max_{x \in \setR} q(x|y)P_C(x)$$
with
$$P_C(x) = E_p \int_{I_C} p(x|y) dy.$$
\item[(c)] $\UCd(y) = E_pE_q \int_{I_C} \Bsbl q(\xs(z)|y)p(\xs(z)|y) - c \Bsbr dz  > 0$, $y \in I_C$.
\item[(d)]  $\UCs(y) = E_q \Bsbl q(\xs(y)|y) P_C(\xs(y) - \aC c \Bsbr  > 0$, $y \in I_C$, where
$$\aC = 2L_CE_p.$$
\end{enumerate}

Then for every $\epsilon > 0$ there exists a $\dlO > 0$ such that for all discrete agent sets $\Add$ and $\Ads$ with 
$$ 0 < \dld,\dls < \dlO,$$
the following discrete interval community structure $\CSddNash = \ddcommStrucNash$ is a $\epsilon$-relative Nash equilibrium with
$$\UCSdNash(y) > 0, \qquad y \in \Add,$$
and
$$\UCSsNash(y) > 0, \qquad y \in \Ads,$$
where $\UCSdNash(y)$ and $\UCSsNash(y)$ are the utility rates for content consumption, and production, under the discrete community structure $\CSddNash = \ddcommStrucNash$.

The  discrete interval community structure $\CSddNash = \ddcommStrucNash$ is given as follows. If there exists a community $C = \CI$ in the Nash equilibrium  $\CSNash = \commStrucNash$ under the continuous model, then the discrete interval community $\ddC = \Cd$ on $I_C$ with distance $\dl$ given by
$$ \Cdd = \intdd = \Add \cap I_C$$
and
$$ \Cds = \intsd = \Ads \cap I_C,$$
exists in the  community structure $\CSddNash  = \ddcommStrucNash$.

Furthermore, for a discrete interval community $\ddC = \Cd \in \sCdNash$ we have that
$$\aCd(y) = E_p, \qquad y \in \intdd,$$
and
$$ \bCd(x|y) = E_q \delta(\xd(y) - x), \qquad y \in \intsd,$$
where
$$\xd(y) = \arg \max_{x \in \setR} q(x|y)\PCd(x).$$
\end{lemma}
Before we prove Lemma~\ref{lemma:dnash}, we make the following observation.  In~\cite{continuous_model_arxiv} we showed that there always exists a Nash equilibrium  $\CSNash = \commStrucNash$ as given in in the statement of Lemma~\ref{lemma:dnash}. Furthermore, the \erNash  $\CSddNash = \ddcommStrucNash$ given in the statement of the Lemma~\ref{lemma:dnash} has the same properties as the \erNash in  Proposition~\ref{prop:dnash}. As a result, we have that Proposition~\ref{prop:dnash} follows immediately from  Lemma~\ref{lemma:dnash}. Or in other words, proving Lemma~\ref{lemma:dnash} is sufficient to prove  Proposition~\ref{prop:dnash}.


\begin{proof}
Let $\commStrucNash$ be a Nash equilibrium as given in the statement of the lemma.  From the analysis in~\cite{continuous_model_arxiv} we have for all communities
 $C = (I_C,I_C) \in \sCNash$ that the functions $\UCd(y)$ and  $\UCs(y)$ are continuous and symmetric with respect to $mid(I_C)$ on $[mid(I_C) - L_C, mid(I_C) + L_C]$.
Furthermore, by Property (a) in the statement of the lemma,  we have  we have for all communities $C = (I_C,I_C) \in \sCNash$ that
$$ I_C = [mid(I_C) - L_C, mid(I_C) + L_C).$$
Combining these two results, it follows that there exists a $\Dl_U$,
$$ 0 < \Dl_U,$$
such that for all communities
 $C = (I_C,I_C) \in \sCNash$ we have that
$$\UCd(y) > 2\Dl_U >0, \qquad y \in I_C,$$
and
$$\UCs(y) > 2\Dl_U >0, \qquad y \in I_C.$$
Let $\dlO>O$ be a positive scalar such that the following is true.
For each interval community $C = (I_C,I_C) \in \sCNash$ in the Nash equilibrium $\CSNash$ given in the statement of the lemma, and each discrete interval community  $\ddC = \Cd$ on $I_C$ with distance $\dl$ such that
$$ 0 < \dl < \dlO,$$
we have that
$$ \left | \dls \FCdd(y) - \FCd(y) \right | < \min\left \{ \frac{\epsilon}{2}, \Dl_U \right \}, \qquad y \in \setR$$
and
$$ \left | \dld \FCsd(y) - \FCs(y) \right | < \min\left \{ \frac{\epsilon}{2}, \Dl_U \right \}, \qquad y \in \setR,$$
where the functions $\FCdd(y)$, $\FCd(y$, and $\FCsd(y)$ and $\FCs(y)$, are as defined in Section~\ref{app:UCdd}~and~\ref{app:UCsd}. 
Note that that such a $\dlO$ exists by Lemma~\ref{lemma:Riemann_FCdd}~and~\ref{lemma:Riemann_FCsd}. Furthermore, from the analysis in~\cite{continuous_model_arxiv} we have for each community $C = (I_C,I_C \in \CSNash$ that
$$ \UCd(y) = \FCd(y), \qquad y \in I_C,$$
and
$$ \UCs(y) = \FCs(y), \qquad y \in I_C.$$
It then follows that if the sets $\Add$ and $\Ads$ are such that
$$ 0 < \dld,\dls < \dlO,$$
then we have for all communities $\ddC = \Cd \in \sCdNash$ in the discrete interval community structure $\CSddNash = \ddcommStrucNash$ given in the statement of the lemma that
$$\UCdd(y) = \FCdd(y) =
E_p E_q \sum_{z \in \intdd} \Big [ p(\xd(z)|y) q(\xd(z)|z) - c \Big ]dz > \Dl_U >0, \qquad y \in \intdd,$$
where
$$ \xd(z) =  \arg\max_{x \in \setR} q(x|z) \PCd(x),$$
and
$$\UCsd(y) = \FCsd(y) =
 E_q \Big [  q \big (\xd(y)|y \big ) \PCd(\xd(y))- \aCd c \Big ]  > \Dl_U >0, \qquad y \in \intsd,$$
where
$$\aCd = \sum_{y \in \intdd} E_p.$$

Let $\ddC = \Cd \in \sCdNash$ be a community in the community structure $\ddcommStrucNash$ as given in the statement of the lemma. By definition,

Furthermore, by the definition we have for the the discrete interval community structure  $\CSddNash = \ddcommStrucNash$ given in the statement of the lemma that for the following is true. Let $\ddC = \Cd \in \sCdNash$ be a community in the community structure $\ddcommStrucNash$, then for all agents $y \in \intdd$ we have that
$$\alpha_{\ddC'}(y) = 0, \qquad \ddC' \in \sCdNash, \ddC' \neq \ddC$$
and
$$||\beta_{\ddC'}(\cdot|y)|| = 0,  \qquad  \ddC' \in \sCdNash, \ddC' \neq \ddC.$$
It then follows  that 
$$ \UCSddNash(y) = \UCdd(y), \qquad y \in \intdd$$
and
$$\UCSsdNash(y) = \UCsd(y), \qquad y \in \intsd.$$
Combining the above results, we obtain that  for the  discrete interval community structure $\CSddNash = \ddcommStrucNash$ given in the statement of the lemma we have
$$\UCSddNash(y) > 0, \qquad y \in \Add,$$
and
$$\UCSsdNash(y) > 0, \qquad y \in \Ads.$$

In order to prove the lemma, it then remains to show that the discrete interval community structure  $\CSddNash = \ddcommStrucNash$ is a \erNash, i.e. we have to show that
\begin{enumerate}
\item[a)]
$$\UCSddNash(\aCd(y)|y) - \UCSddNash(y)< \epsilon, \qquad y \in \intdd.$$
where
$$ \aSCdd(y) = \arg \max_{\aSCd'(y): || \aSCd'(y) || \leq E_p}  \UCSddNash(\aSCd'(y)|y), $$
and
\item[b)]
$$\UCSsdNash(\bSCd(y)|y) - \UCSsdNash(y)< \epsilon, \qquad y \in \intsd.$$
where
$$ \bSCdd(y) = \arg \max_{\bSCd'(y): || \bSCd'(y) || \leq E_q}  \UCSsdNash(\bSCd'(y)|y). $$
\end{enumerate}
As by definition, we have for all communities $\ddC = \Cd \in \sCdNash$ in the discrete interval community structure $\ddcommStrucNash$ as given in the statement of the lemma for all agents $y \in \intdd$ we have that
$$\alpha_{\ddC'}(y) = 0, \qquad \ddC' \in \sCdNash, \ddC' \neq \ddC$$
and
$$||\beta_{\ddC'}(\cdot|y)|| = 0,  \qquad  \ddC' \in \sCdNash, \ddC' \neq \ddC,$$
this is equivalent to show that for each community $\ddC = \Cd \in \sCdNash$ in the community structure $\ddcommStrucNash$ given in the statement of the lemma, we have that
\begin{enumerate}
\item[a)]
$$\UCSddNash(\aCd(y)|y) - \UCdd(y)< \epsilon, \qquad y \in \intdd.$$
where
$$ \aSCd(y) = \arg \max_{\aSCd'(y): || \aSCd'(y) || \leq E_p}  \UCSddNash(\aSCd'(y)|y), $$
and
\item[b)]
$$\UCSsdNash(\bSCd(y)|y) - \UCsd(y)< \epsilon, \qquad y \in \intsd.$$
where
$$ \bSCd(y) = \arg \max_{\bSCd'(y): || \bSCd'(y) || \leq E_q}  \UCSsdNash(\bSCd'(y)|y). $$
\end{enumerate}

We first prove property that  for all communities $\ddC = \Cd \in \sCdNash$ we have
$$\UCSddNash(\aCd(y)|y) - \UCdd(y)< \epsilon, \qquad y \in \intdd.$$
To do this, we consider a given community  $\tdC = \tCd \in \sCdNash$ in the discrete interval community structure $\ddcommStrucNash$, and a given agent $\ty \in \tintdd$.  In order to prove the above result, it suffices to show for the given  community $\tdC$ and given agent $\ty$ we have that
$$\UCSddNash(\aCd(\ty)|\ty) - \tUCdd(\ty)< \epsilon.$$

By Lemma~\ref{lemma:single_community_consumption}, we have that
$$ \aSCd(\ty) = \arg \max_{\aSCd'(\ty): || \aSCd'(\ty) || \leq E_p}  \UCSddNash(\aSCd(\ty)|\ty)$$
is given by
$$ \aCd(\ty) = \left \{
\begin{array}{rl}
E_p, & \ddC = \ddCs(\ty) \\
0, & \ddC \neq \ddCs(\ty),
\end{array} \right .$$
where
$$\ddCs(\ty) = \arg \max_{\ddC \in \sCd} \int_{x \in \setR} \Bsbl p(x|\ty) \QCd(x) - \bCd(x) c \Bsbr$$
and
$$\QCd(x) = \sum_{z \in \intsd} \bCd(x|z) q(x|z).$$ 
To simplify the notation let
$$ \ddCs =  \ddCs(\ty),$$
and let 
 $\intsdd$ and $\intssd$ be the set of agents that consume and produce content in $\ddCs$, i.e. we have that
$$\ddCs = \ddCs(\ty) = \Csd.$$
It then follows that
$$\UCSddNash(\aCd(\ty)|\ty) = \FsCdd(\ty)
= E_p E_q \sum_{z \in \intssd} \Big [ p(\xd(z)|\ty) q(\xd(z)|z) - c \Big ]dx,$$
where
$$\xd(z) = \arg \max_{x \in \setR} q(x|z)P_{\dl,\ddCs}(x).$$

Let the interval community $C = \CI$ in the Nash equilibrium  $\commStrucNash$ given in the statement of the lemma be such that for the given  $\tdC = \tCd \in \sCdNash$ we have that
$$ \tintdd = \Add \cap I_C$$
and
$$ \tintsd = \Ads \cap I_C.$$
Furthermore, let the interval community $C^* = \Cs$ in the Nash equilibrium  $\commStrucNash$ given in the statement of the lemma be such that for the community
$$\ddCs =  \ddCs(\ty) =  \Csd$$
we have that
$$ \intsdd = \Add \cap I_{C^*},$$
and
$$ \intssd = \Ads \cap I_{C^*}.$$
Note that such interval communities  $C = \CI$ and   $C^* = \Cs$ by the construction of the community structure  $\ddcommStrucNash$ given in the statement of the lemma.

Using these definition, we then have that for the given community $\tdC$ and given agent $\ty$ that
\begin{eqnarray}
\dls \Bbl \UCSddNash(\aCd(\ty)|\ty)  - \tUCdd(\ty) \Bbr
&=& \dls \UCSddNash(\aCd(\ty)|\ty) - \FsCd(\ty)  +  \cdots \label{eq:diff_UCSddNash}\\ \nonumber
&& \FsCd(\ty) -  \FCd(\ty) + \cdots \\ \nonumber
&& \FCd(\ty) - \dls  \tUCdd(\ty) \\ \nonumber
&=& \dls \FsCdd(\ty) - \FsCd(\ty)  +  \cdots \\ \nonumber
&& \FsCd(\ty) -  \FCd(\ty) + \cdots \\ \nonumber
&& \FCd(\ty) - \dls  \tFCdd(\ty)
\end{eqnarray}
By the properties of a Nash equilibrium as defined in~\cite{continuous_model_arxiv}, we have for each community $C = \CI$ in the Nash equilibrium $\commStrucNash$ given in the statement of the lemma that 
$$ \FCd(y) \geq F^{(d)}_{C'}(y), \qquad y \in I_C, C' \in \sCNash.$$
It then follows that in Eq.~\eqref{eq:diff_UCSddNash} we have
$$ \FsCd(\ty) -  \FCd(\ty) \leq 0,$$
and we obtain that
\begin{eqnarray*}
\dls \Bbl \UCSddNash(\aCd(\ty)|\ty)  - \tUCdd(\ty) \Bbr
&\leq & \Big | \dls \FsCdd(\ty) - \FsCd(\ty) \Big | +  \cdots \\
&& \Big |  \FCd(\ty) - \dls \tFCdd(\ty) \Big |.
\end{eqnarray*}
By construction we have for
$$ 0 < \dl < \dlO$$
that
$$\Big | \dls \FsCdd(\ty) - \FsCd(\ty) \Big | <  \frac{\epsilon}{2}$$
and
$$\Big |  \FCd(\ty) - \dls \FCdd(\ty) \Big | <  \frac{\epsilon}{2}.$$
Combining the above results, we obtain that
\begin{eqnarray*}
\dls \Bbl \UCSddNash(\aCd(\ty)|\ty)  - \UCdd(\ty) \Bbr
&\leq & \Big | \dls \FsCdd(\ty) - \FsCd(\ty) \Big | +  \cdots \\
&& \Big |  \FCd(\ty) - \dls \FCdd(\ty) \Big |\\
&<& 2 \frac{\epsilon}{2} \\
&=& \epsilon,
\end{eqnarray*}
and
$$\UCSddNash(\bCd(\ty)|y) - \UCSddNash(\ty)< \epsilon, \qquad y \in \intdd.$$
It then follows that  for community structure $\CSddNash = \ddcommStrucNash$ given in the statement of the lemma we have 
$$\UCSddNash(\aCd(y)|y) - \UCSddNash(y)< \epsilon, \qquad y \in \Add.$$

Using the same line of argument, we can show that for community structure $\CSddNash = \ddcommStrucNash$ given in the statement of the lemma we have 
$$\UCSsdNash(\bCd(y)|y) - \UCSsdNash(y)< \epsilon, \qquad y \in \Ads.$$
It then follows that the  community structure $\CSddNash = \ddcommStrucNash$ given in the statement of the lemma is indeed a \erNash. This completes the proof of the lemma.
\end{proof}

\end{document}